\documentclass[journal]{IEEEtran}

\usepackage{tipa}
\usepackage{amssymb}
\usepackage{bbm}
\usepackage{amsfonts}
\usepackage{mathrsfs}
\usepackage{amsmath,amsthm}
\usepackage{amsfonts,amssymb}
\usepackage{graphicx}
\usepackage{epsf}
\usepackage{multicol}
\usepackage{delarray}
\usepackage{amsopn,epsfig,psfrag}
\usepackage{float}
\usepackage{subfigure}

\allowdisplaybreaks

\newtheorem{theorem}{Theorem}

\newtheorem{definition}{Definition}
\newtheorem{lemma}{Lemma}

\newtheorem{remark}{Remark}

\title{Perfectly Controllable Multi-Agent Networks
\thanks{This work was supported by the National Natural Science Foundation of China under Grants 61374062, 61603288, 61573203 and 61673013 and by the Natural Science Foundation of Shandong Province for Distinguished Young Scholars under Grant JQ201419.}}

\author{Shaobin Cao
\thanks{Corresponding author: Zhijian Ji (e-mail: jizhijian@qdu.edu.cn)}
\thanks{Shaobin Cao, Zhijian Ji and Haisheng Yu are with the College of Automation and Electrical Engineering, Qingdao University, Qingdao, 266071,
China. E-mails: {\it csbin2016@163.com} (S. Cao); {\it yhsh@qdu.edu.cn} (H. Yu).}
\quad Zhijian Ji \quad Hai Lin  \quad Haisheng Yu
\thanks{Hai Lin is with the Department of Electrical Engineering, University
of Notre Dame, IN 46556 USA. E-mail: {\it
hlin1@nd.edu} (H. Lin).}
}

\begin{document}

\maketitle

\begin{abstract}
This note investigates how to design topology structures to ensure the controllability of multi-agent networks (MASs) under any selection of leaders. We put forward a concept of perfect controllability, which means that a multi-agent system is controllable with no matter how the leaders are chosen. In this situation, both the number and the locations of leader agents are arbitrary. A necessary and sufficient condition is derived for the perfect controllability. Moreover,
a step-by-step design procedure is proposed by which topologies are constructed and are proved to be perfectly controllable. The principle of the proposed design method is interpreted by schematic diagrams along with the corresponding topology structures from simple to complex. We show that the results are valid for any number and any location of leaders. Both the construction process and the corresponding topology structures are clearly outlined.
\end{abstract}



\section{Introduction}

Recently, distributed coordination and control of networks of dynamic agents has attracted constant and extensive attention (see e.g. \cite{XiaoFTAC2018, MaJTIE2017,LTianAccess2018,ZhenMW2018,CaiN2017Complexity,XiAccess2018,LiuJiCyb2017,QiqingITC,NCaiJSSC,YuanGTAC2017,NiuLiuCyber2018}). Among them, the study of graphical characterizations of multi-agent controllability has received special attention (see e.g. \cite{JiWCon,TanOn,JiLL,JiYA,LiuJCon,LiuCRNC2012}). It is worth mentioning that the results along this line are usually established only with respect to the pre-selected leaders; and the restriction on leader agents is a prerequisite for reaching a conclusion. Beyond that, the results on a straightforward identification of controllable topologies are very few, let alone how to design controllable graphs directly. What is commonly seen is the construction of uncontrollable topologies, which actually forms a necessary condition on the determination of controllability for a multi-agent system. Our efforts in the note, however, are free from both of these two constraints, i.e., the requirement of pre-selected leaders and the direct design of uncontrollable topologies rather than controllable ones.

In subsequent sections, we show step by step that it is really feasible to design controllable topology graphs in a direct way regardless of the selection of leaders' locations and the number of leader agents. Here graphs with this controllability property of any leaders are called perfectly controllable to distinguish it from the usual notion of controllability with respect to a particular group of leaders. To this end, we came up with a concept of schematic graphs, and the design process is illustrated by combing the schematic graphs with topology graphs. The schematic diagrams help to explain in a particular way the mechanism of the proposed principle upon which the topology structures are desiged. The use of schematics is a feature of this design method.

Now, an increasing number of scholars from different angles are making elaborate effort on the investigation of
graphical characterizations for multi-agent controllability. Recent results in this regard include, e.g., \cite{JiYA,LiuboIET2017,HsuLetter2016,GuanWangletter2018,ChaoIMA,ZhaoGuan2017,HsuIJRNC2017,LuINJRNC2018}. One of the starting points of graphical characterizations is to uncover graphical interpretations for algebraic conditions on controllability \cite{JiWCon,TanOn}. The related results were mainly obtained by graph partitioning method \cite{JLYTAC2015,AgGhAuto2017,JiLL,YWETac2016,GuanRNC2017}; the concept of controllability destructive nodes \cite{JiYA};  the balancing set approach \cite{MHMTAC2018}; the eigenvector based technique \cite{GGCo,LiuJCon,JiWCon}, and the notion of graph controllability classes \cite{GGCoTAC2015}, etc. Reading controllability from the level of graph theory is quite involved even for simple topology structures such as path and tree graphs \cite{ParlanNTAC2012,LiuJCon,JiLL}.

The contribution of the note is twofold. Firstly, we find a kind of interesting topology which is controllable, while independent of the leaders' selection. An algebraic necessary and sufficient condition is developed for determining
this kind of controllable topology graphs. Secondly, a systematic design procedure is proposed by which a large number of perfectly controllable graphs can be constructed. To illustrate this design process more clearly, a kind of schematic diagrams is employed creatively in the description. It can be imagined that it is not possible to exhaust all  perfectly controllable graphs by one construction method alone. Even so, some of the rules followed in the construction are verified to be effective, at least for the situations mentioned in the note. It can be expected that the rules are instructive in dealing with more situations.

The note is organized as follows. Section II contains preliminary notions and problem setup. In Section III, main results are reported. Section IV concludes the note.

\section{Preliminaries}

\subsection{Graph theory}
Let $G = (V,E,A)$ represent a graph. A MAS consists of $n$ single integrator dynamics agents, which are labeled by the vertex set $V = \{ {v_1},{v_2}, \cdots ,{v_n}\} $. The edge set $E = \{ ({v_i},{v_j}):{v_i},{v_j} \in V\}$ denotes the connection between two nodes, if agent $j$ can directly get information from agent $i$, i.e., $({v_i},{v_j})\in E$, the edge is an ordered pair with the nodes in set $V$. When $({v_i},{v_j}) \in E \Leftrightarrow ({v_j},{v_i}) \in E$, graph $G$ is called an undirected graph, otherwise directed graph. Assuming there is no self-loop in the graph, i.e., all edges of $E$ are in accord with $i \ne j$. The $n\times n$ adjacency matrix is expressed as $A = \{ {a_{ij}}\} $. ${a_{ij}} \in \{ 0, - 1\}$, if ${v_i},{v_j} \in V({v_i},{v_j})\in E$, agent $v_{i}$ is called a neighbor of agent $v_{j}$. If ${a_{ij}} =  1$, agent $j$ is called a neighbor of agent $i$, $e_{ij}\in E_{-}$, else ${v_i},{v_j} \in V({v_i},{v_j})\notin E$, then ${a_{ij}} = 0$, agent $j$ is not a neighbor of agent $i$. In particular, any ${a_{ii}} = 0$ means that each diagonal elements of $A$ is $0$. ${N_i} = \{ {v_j} \in V:({v_i},{v_j}) \in E\} $ represents the neighborhood set of vertex $i$. The number of neighbors is the degree of the node. The degree of vertex $v_{i}$ is defined as ${d_i} = \sum\nolimits_{j = 1}^n {\left| {{a_{ij}}} \right|} $. The valency matrix of a graph $G$ is a diagonal matrix which is defined as $\triangle(G) = diag({d_i}) \in {\mathbb{R}^{n \times n}}$. A graph is called a complete graph if any two nodes are neighbors. The Laplacian matrix of graph $G$ is given by
\begin{equation}
L(G) = \triangle(G) - A(G)
\end{equation}
 If graph $G$ is given precisely, it can be abbreviated as $L$. If graph $G $ is an undirected graph, the Laplacian matrix $L$ and the adjacency matrix $A$ are symmetrical.
Hence the entries of $L$ can be written as
\begin{equation}
{[L]_{ij}} = \left\{ {\begin{array}{*{20}{c}}
{\sum\limits_{j \in {N_{i}}} {\left| {{a_{ij}}} \right|} }&{}&{i = j}\\
{ - {a_{ij}}}&{}&{i \ne j,j \in {N_{i}}}\\
0&{}&{{\rm{otherwise}}}
\end{array}} \right.
\end{equation}

\subsection{Problem formulation and a preliminary result}
Consider a single integrator multi-agent system with dynamics of each agent described by
\begin{equation}
{{\dot x}_i} = {u_i},\quad i \in \mathcal{I}_{n+l},\label{sys1}
\end{equation}
where $\mathcal{I}_{n+l}\mathop  = \limits^\Delta\{1,\ldots, n+l\},$ $n$ and $l$ are the numbers of followers and leaders, respectively. The information flow between agents is modeled by an undirected graph $\mathcal{G}.$ Its Laplacian is defined by $L=D-A,$ where $D$ is the diagonal degree matrix and $A$ is the adjacency matrix.
For each agent, the neighbor-based rule is
\begin{equation}
{u_i} =  \sum\limits_{{j} \in {N_i}} {({x_j} - {x_i})}. \label{pro1}
\end{equation}
Then the closed-loop system of (\ref{sys1}) and (\ref{pro1}) reads
\begin{equation}
\dot x =  - Lx, \label{orgisys1}
\end{equation}
where $x$ is the stack vector of all $x_i.$

With the last $l$ agents being selected as leaders, $L$ can be partitioned as
$
L = \left[ {\begin{array}{*{20}{c}}
   {{L_f}} & {{L_{fl}}}  \\
   {L_{lf}} & {L_l}  \\
 \end{array} } \right],
$
where ${L_f}$ and ${L_l}$ correspond to the indices of followers and leaders, respectively. Denote by $x_f$ and $x_l$ the stack vectors of all followers' and leaders' states. Then ${x^T} = [x_f^T,x_l^T].$ It follows from (\ref{orgisys1}) that the followers' dynamics is
\begin{equation}
\dot x_f =  - {L_f}x_f - {L_{fl}}{x_l}. \label{dynfollower}
\end{equation}
Since the followers' dynamics is captured by (\ref{dynfollower}), a multi-agent system with dynamics (\ref{sys1}) is controllable if and only if so is (\ref{dynfollower}). The following lemma is for the determination of controllability in association with the selection of leaders.

\begin{lemma}\cite{JiWCon}
The multi-agent system with dynamics (\ref{sys1}) is controllable if and only if there is no eigenvector of $\mathcal{G}$ taking $0$ on the elements corresponding to the leaders.\label{prelemm}
\end{lemma}

\begin{definition}
A multi-agent system is said to be perfectly controllable if it is controllable under any selection of leaders. Here both the number and the locations of leaders are arbitrary.
\end{definition}

In what follows, we don't discriminate the controllability of a multi-agent system with that of its corresponding interconnection graph.

\section{Main Results}

The goal of this section is to construct communication graphs under which the multi-agent system is perfectly controllable. We first give a necessary and sufficient condition for the perfect controllability.

\subsection{Algebraic condition}

\begin{theorem}\label{The1Alge}
The multi-agent system (\ref{sys1}) along with protocol (\ref{pro1}) is perfectly controllable if and only if all the eigenvalues of the Laplacian matrix $L$ of the associated interconnection graph are distinct and the eigenvector of each eigenvalue of $L$ has no zero elements.
\end{theorem}
\begin{proof}
Necessity: Suppose by contradiction that there is an eigenvalue $\lambda$ of $L$ with its algebraic multiplicity $n_\lambda$ greater than $1,$ i.e., $n_\lambda>1.$ Since the interconnection graph is undirected, $L$ is symmetric and accordingly there are $n_\lambda$ linearly independent eigenvectors of $\lambda.$ Let $\eta_1,\eta_2$ be any two eigenvectors of $\lambda.$ There is a linear combination $\eta=\alpha_1\eta_1+\alpha_2\eta_2$ of $\eta_1$ and $\eta_2$ so that $\eta$ has at least one zero element, where $\alpha_1, \alpha_2$ are real numbers. Note that $\eta$ is also an eigenvector of $\lambda$ since so is for both $\eta_1$ and $\eta_2.$ Then, by Lemma \ref{prelemm}, the system is uncontrollable if leaders are selected from the agents corresponding to the zero elements of $\eta,$ which is a contradiction to the perfect controllability. The above discussion also yields that any eigenvector of $L$ has no zero elements because the system is perfectly controllable.

Sufficiency: Since each eigenvalue of $L$ is different from all the others, and the eigenvector of each eigenvalue has no zero elements, the system is controllable with any selection of leaders. Otherwise, by Lemma \ref{prelemm}, the uncontrollability of the system under a group of leaders implies that $L$ has an eigenvector with the elements corresponding to the leader agents taking the value of zero. This is a contradiction to the assumption.
\end{proof}

\begin{remark}\label{Rem1}
An algebraic necessary and sufficient condition is presented in Theorem \ref{The1Alge}. The next objective is to cast the algebraic condition on the meaning of topology structures, especially finding out topology graphs satisfying the algebraic condition. Intuitively, it is not possible to uncover all these topology graphs because the number of  topology structures is countless; and the more number of nodes, the harder in finding out all the topologies. Hence, designing topology structures just in a certain circumstance is a realizable approach in the investigation.
\end{remark}

\subsection{Design procedure}

Bearing Remark \ref{Rem1} in mind and for the convenience of presentation, the construction procedure starts with a simple topology structure depicted by (b) of Fig. \ref{Step1con} which consists of eight nodes.

\textbf{Step 1}
The eight nodes $1$ to $8$ are divided into two groups. One is $\Omega_1=\{1, 2, 3, 4\},$ and the other is $\Omega_2=\{5, 6, 7, 8\}.$ The corresponding schematic diagram is (a) of Fig. \ref{Step1con}.

Select arbitrarily one node, respectively, from $\Omega_1$ and $\Omega_2,$ say $1\in\Omega_1$ and $5\in\Omega_2$ to constitute a double nodes set, which is depicted by the first oval shape in (a) of Fig.\ref{Step1con}. In this way, four double nodes sets are written out, as shown in (a) of Fig.\ref{Step1con}, and each node belongs to one and only one double node set.

The original graph consisting of the eight nodes are depicted by (b) of Fig.\ref{Step1con}, which has a simple topology structure and is a candidate for subsequent construction of perfectly controllable graphs.

\begin{remark}
The graph (a) of Fig.\ref{Step1con} and its subsequent evolutionary graphs are schematic diagrams which are used to explain the rules of design. The graph (b) of Fig.\ref{Step1con} and its subsequent evolutionary graphs indicate the resulting topology structures in each step. These two types of graphs together form a description of the design process.
\end{remark}

\begin{figure}[H]
\begin{center}
\subfigure[]{\includegraphics[width=2.63cm]{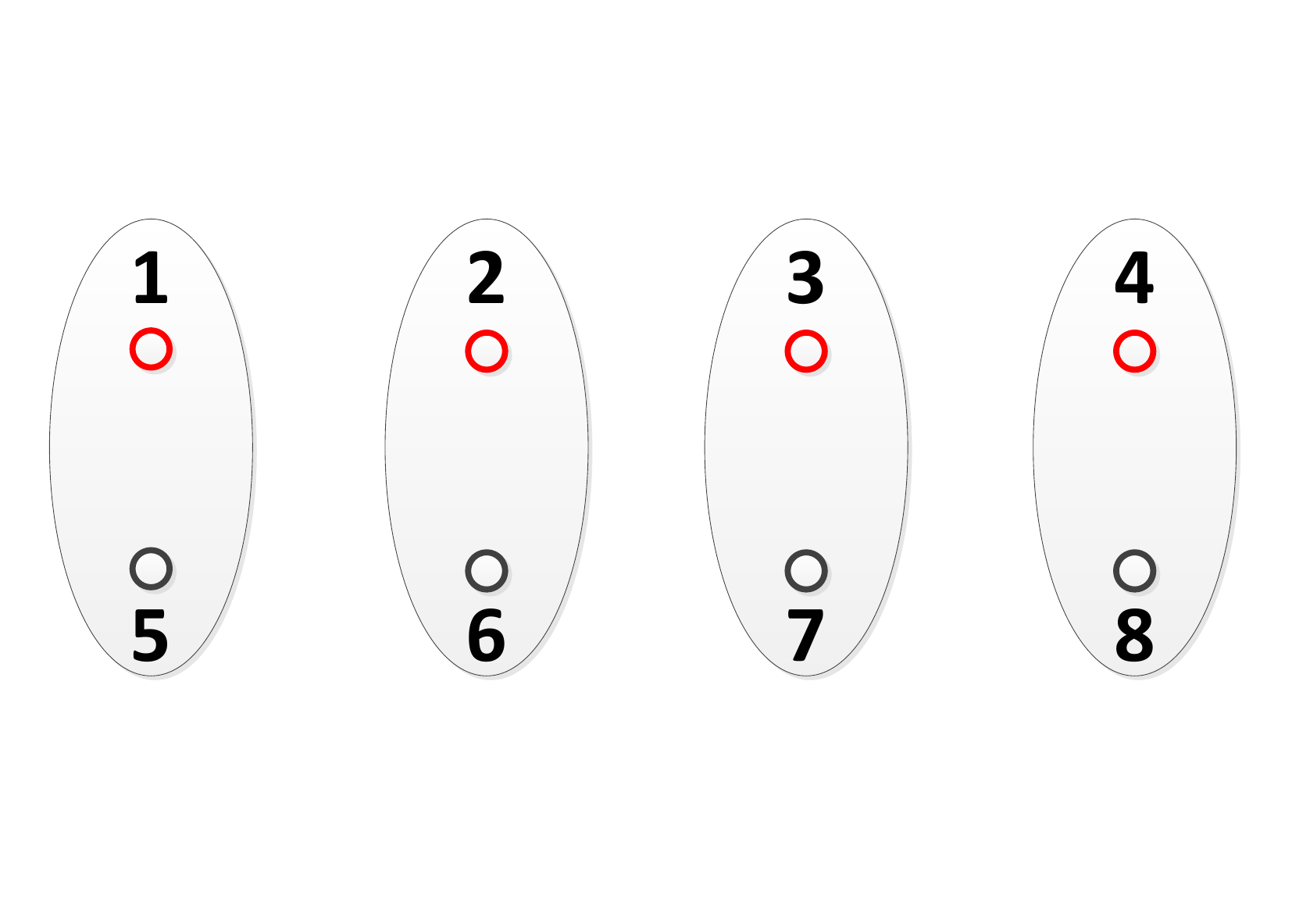}}\qquad
\subfigure[]{\includegraphics[width=2.26cm]{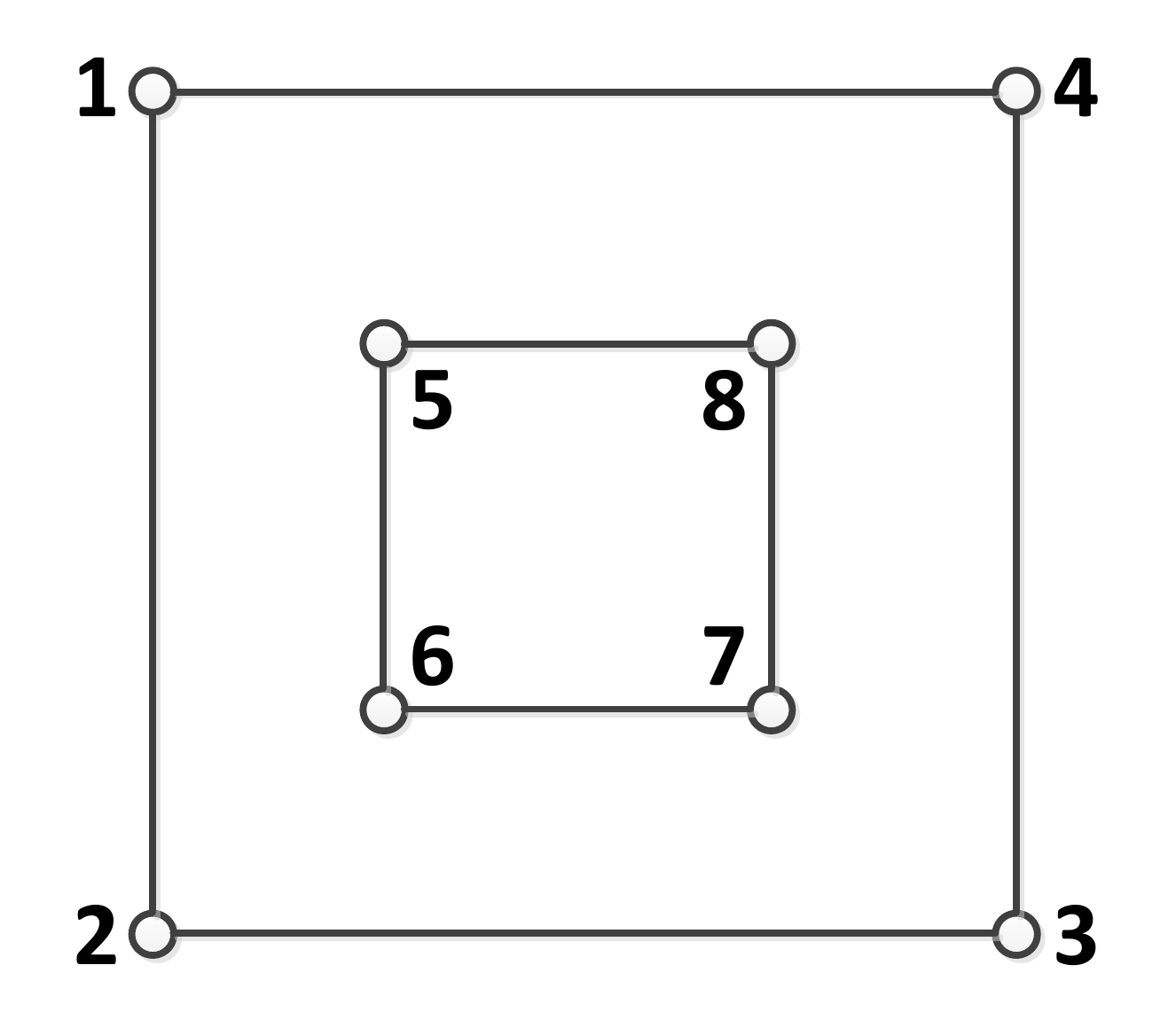}}
\caption{Four double node sets (a); and the original graph (b).}\label{Step1con}
\end{center}
\end{figure}



\textbf{Step 2}
Take any three of the four double nodes sets, say the first, the second and the fourth one; and establish a connection between the two nodes in each of these selected three node sets. This addition of the communication edges is illustrated by (a) of Fig.\ref{Step2con}. The corresponding topology structure is described by (b) of Fig.\ref{Step2con}.

\begin{figure}[H]
\begin{center}
\subfigure[]{\includegraphics[width=2.63cm]{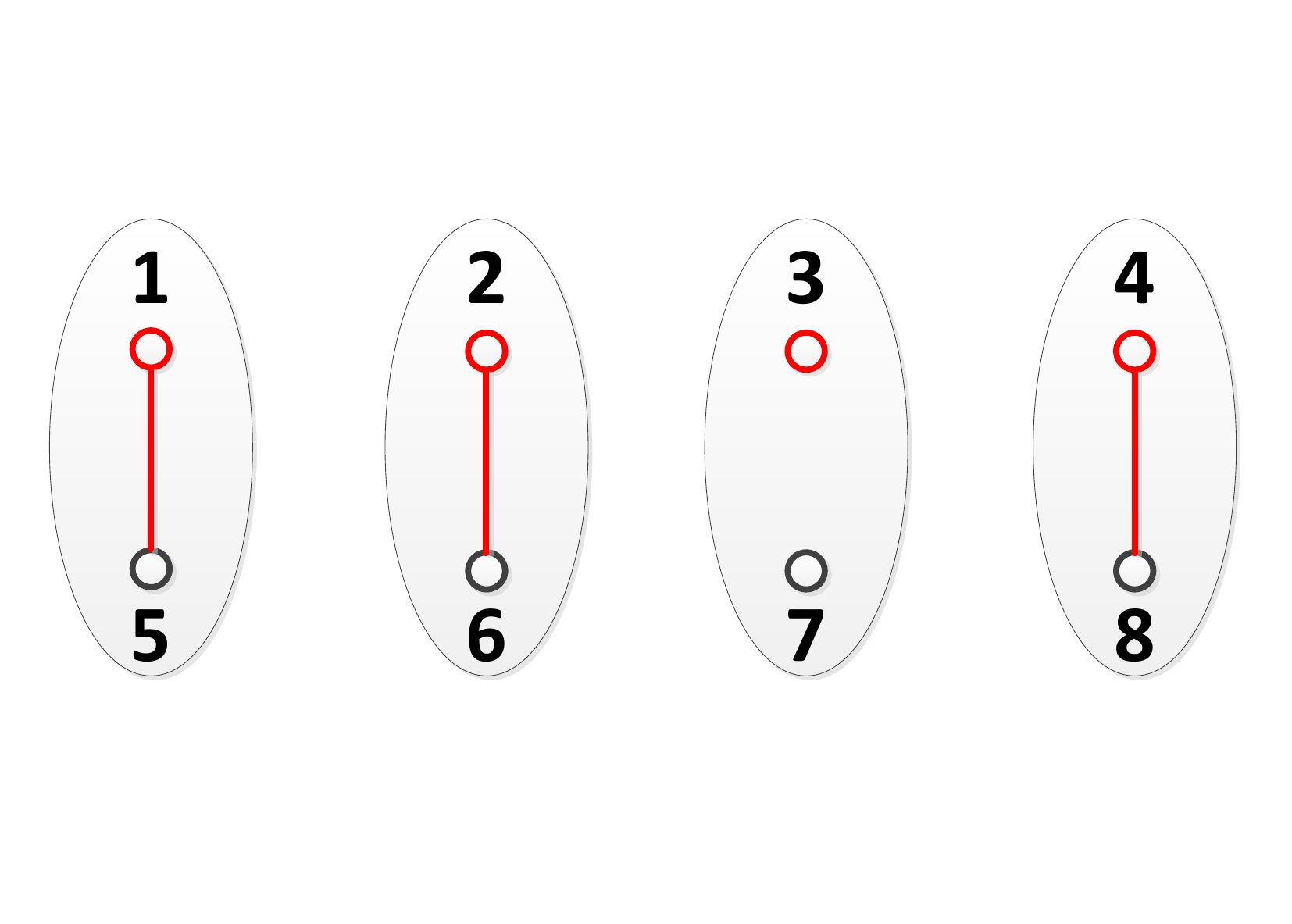}}\qquad
\subfigure[]{\includegraphics[width=2.26cm]{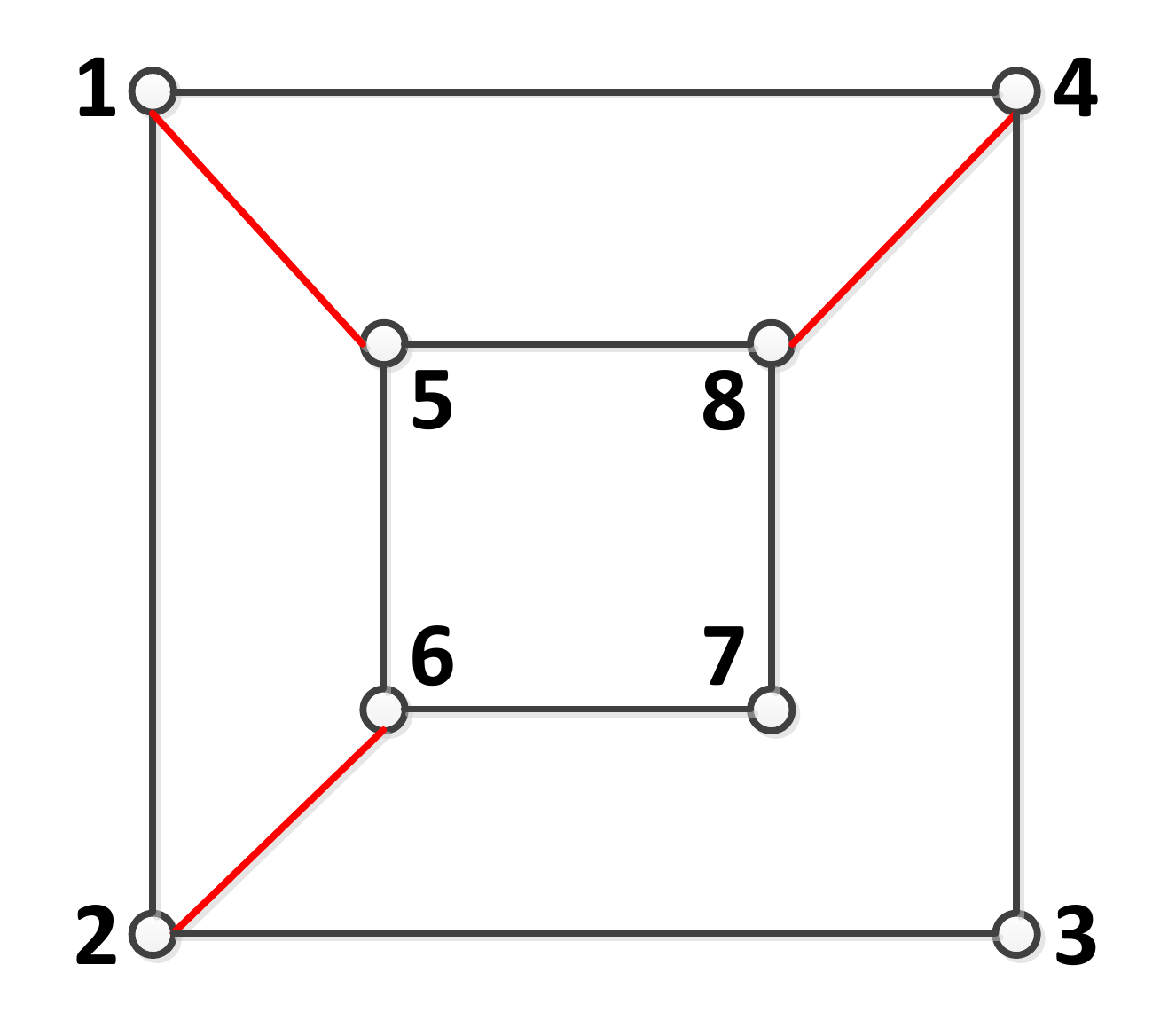}}
\caption{(a) is the addition of the communication edges; (b) is the corresponding topology structure.}\label{Step2con}
\end{center}
\end{figure}

\textbf{Step 3}
Focus on the double nodes set in which the two nodes are not interlinked. For these two nodes, let us consider the one belonging to $\Omega_2.$ Here this node is $7.$

Now, connect $7$ to $4.$ The rules that this connection follows are listed below:
\begin{itemize}
\item[i)] The two nodes $4$ and $7$ to be connected belong to different sets $\Omega_i, i=1,2;$
\item[ii)] The two nodes $4$ and $7$ to be connected are not in the same double nodes set;
\item[iii)] The newly added edge $e_{47}$ between $4$ and $7$ does not intersect with any other edge.
\end{itemize}
The newly added edge and the corresponding topology structure are depicted, respectively,  by (a) and (b) of Fig.\ref{Step3con}.

In this way, each of edges $e_{27}, e_{36}, e_{38}$ is also the right alternative.
\begin{figure}[H]
\begin{center}
\subfigure[]{\includegraphics[width=3cm]{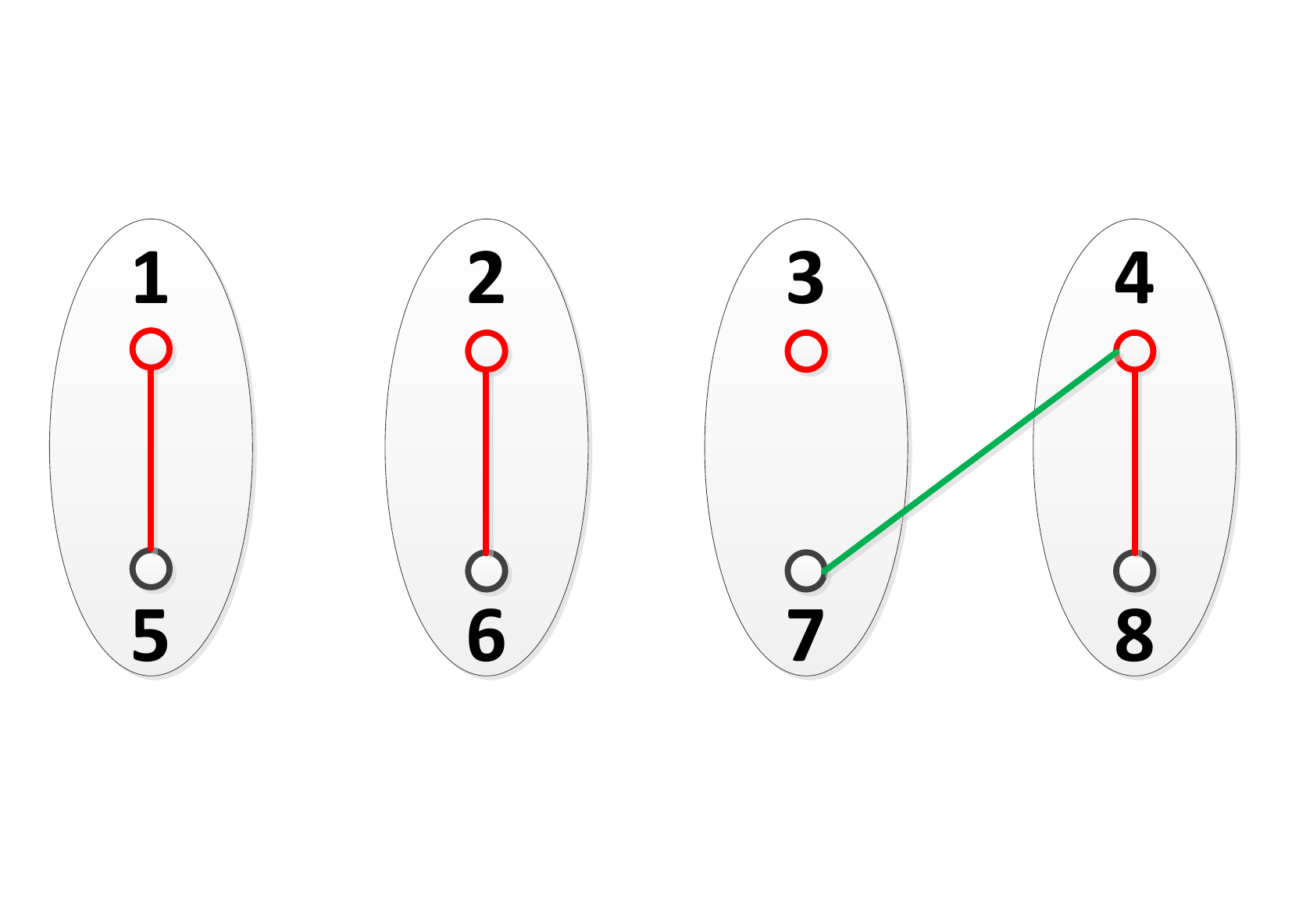}}\qquad
\subfigure[]{\includegraphics[width=2.3cm]{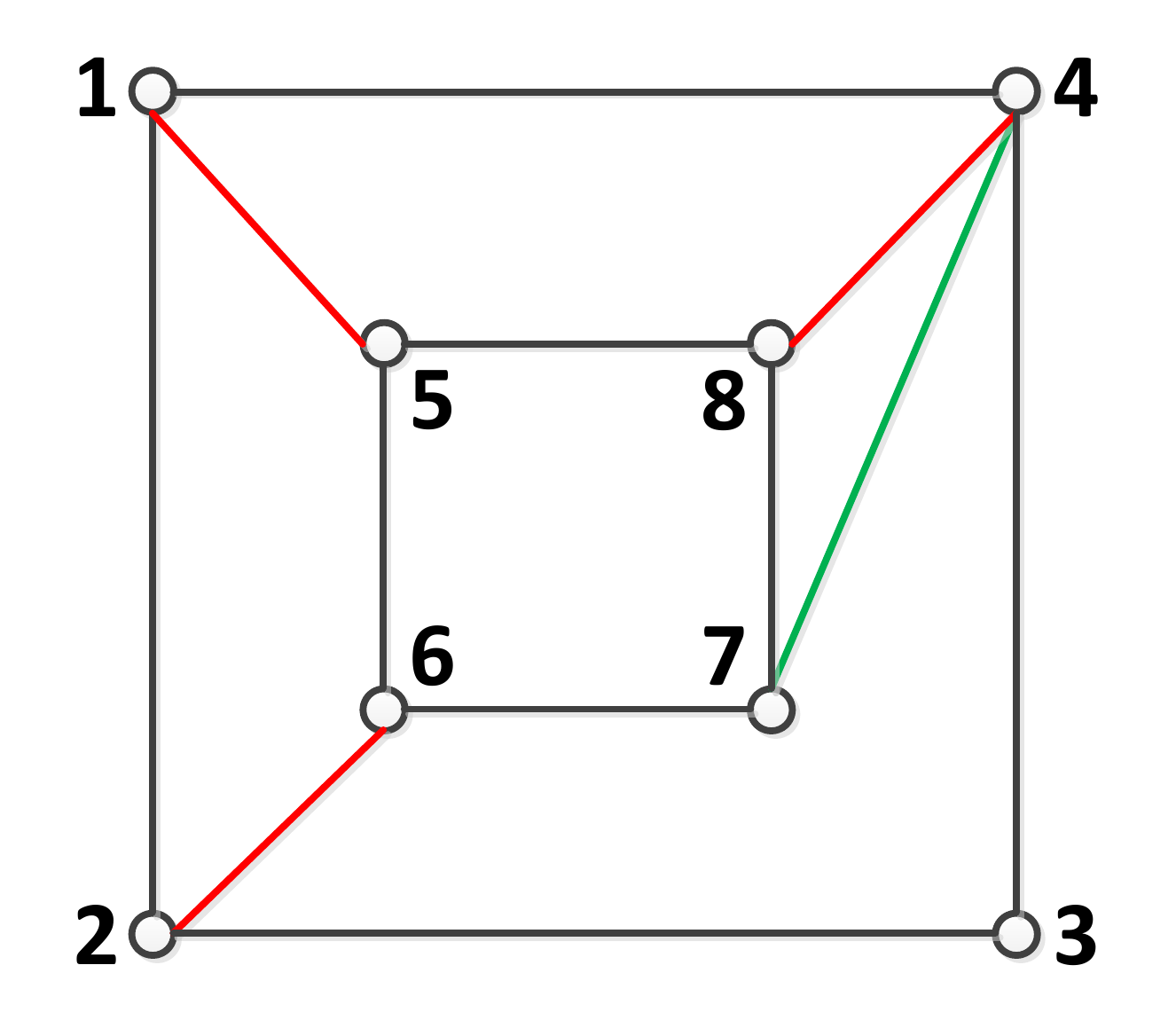}}
\caption{(a) indicates the edge to be added; (b) is the corresponding topology structure.}\label{Step3con}
\end{center}
\end{figure}

\textbf{Step 4}
In this step, the two double nodes sets linked by the newly added edge $e_{47}$ are fixed. In what follows, three  scenarios are stated.

\textbf{Case a)} By following the same rules i) to iii) in Step 3, a new edge, say $e_{16}$ or $e_{25}$ can be created to link the other two unfixed double nodes sets. This manipulation is illustrated by Fig. \ref{Step4-1con}.
\begin{figure}[H]
\begin{center}
\subfigure[]{\includegraphics[width=2.43cm]{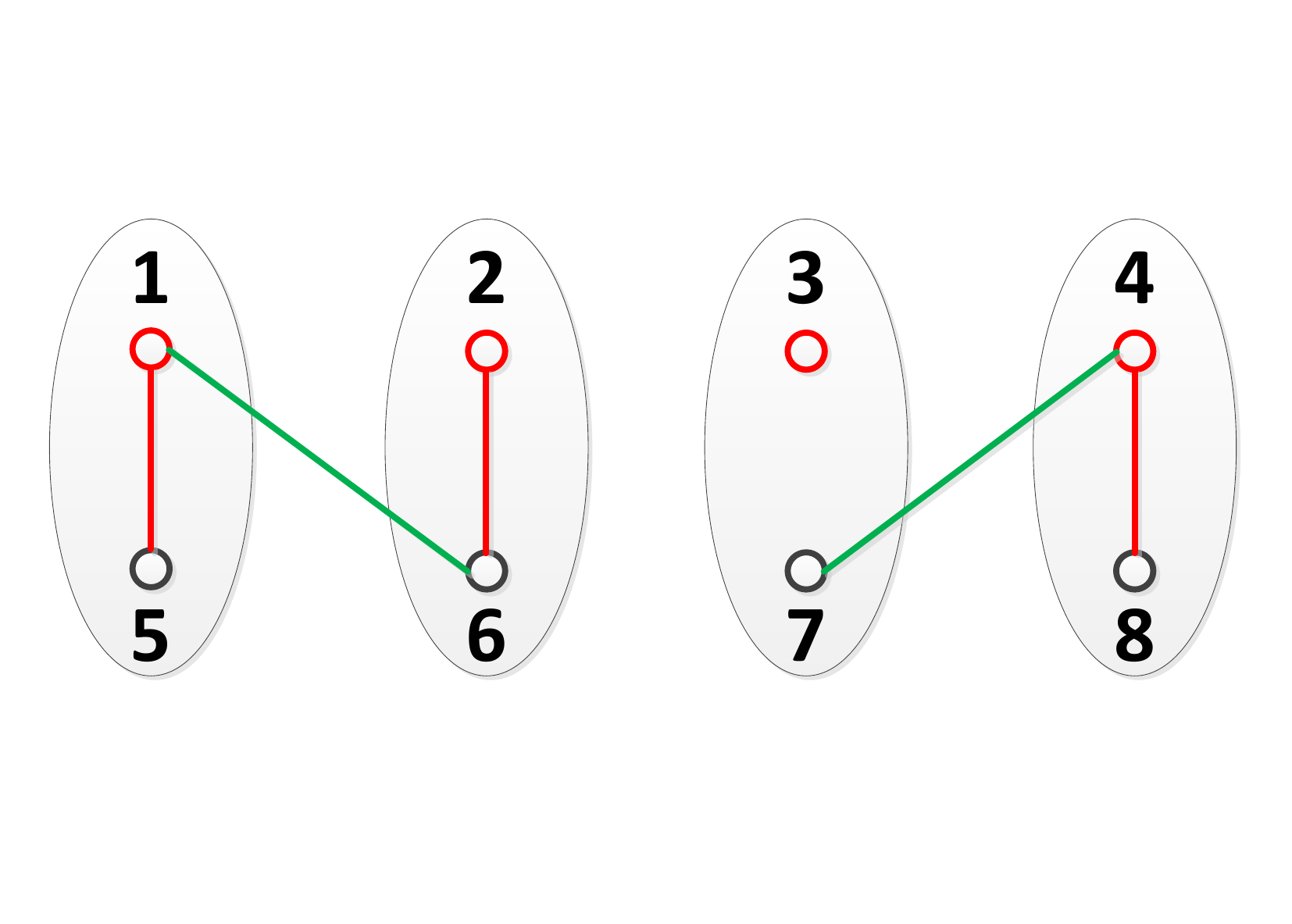}}\qquad
\subfigure[]{\includegraphics[width=2.43cm]{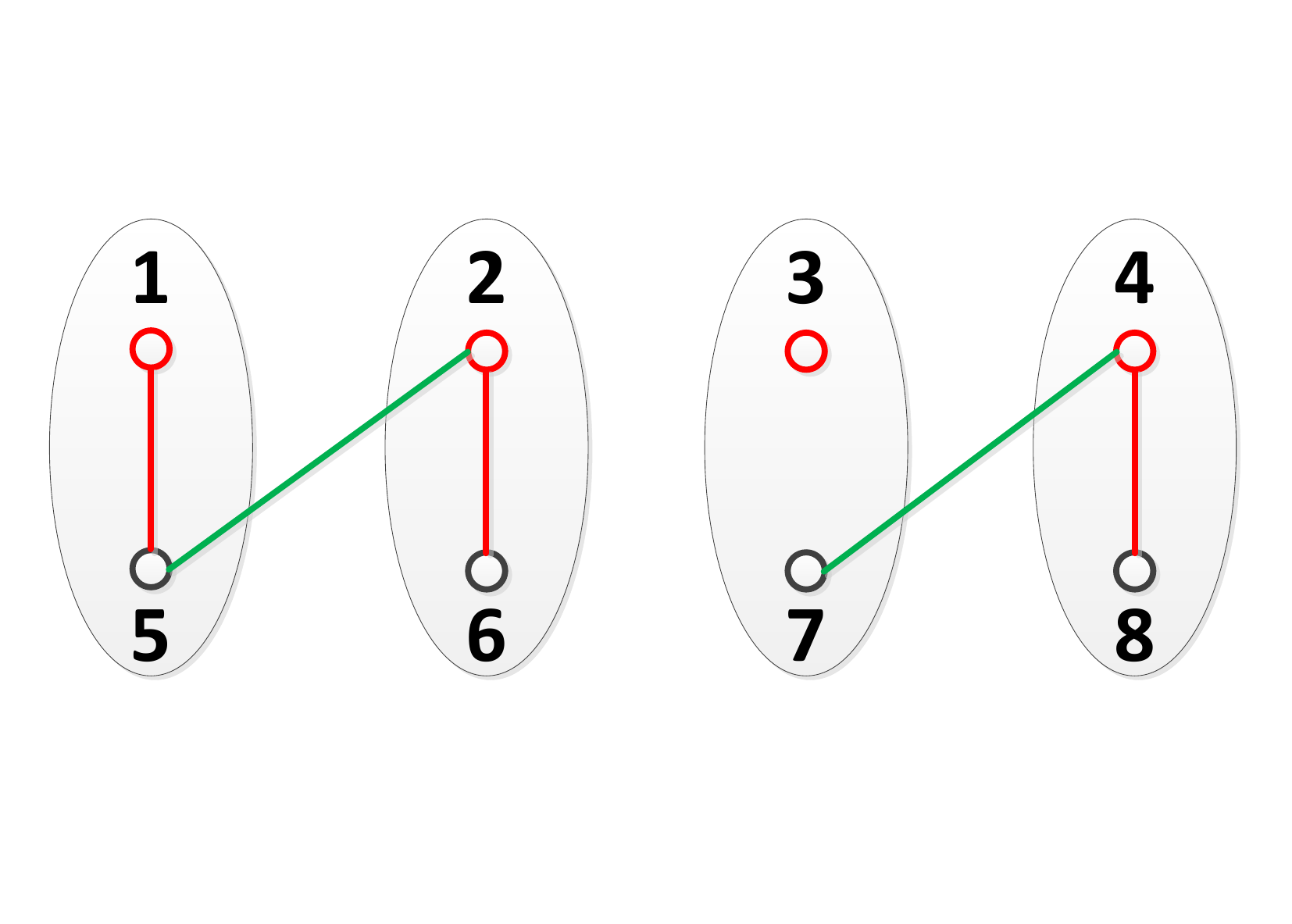}}
\caption{(a) and (b) are the figures associated with the newly added edges between the two unfixed double nodes sets.}\label{Step4-1con}
\end{center}
\end{figure}

\textbf{Case b)}  Let us still consider the two unfixed double nodes sets. In this case, one new node labeled as $9$ and two associated new edges are added to the graphs designed by following Step 3. There are three circumstances.
\begin{itemize}
\item For the nodes $1$ and $2,$ which belong to the same $\Omega_1,$ and meanwhile, respectively, belong to one of the two unfixed double nodes sets,  the two new edges are constructed by connecting the new node $9,$ respectively, to one of these two nodes $1$ and $2.$ This is illustrated by (a) of Fig. \ref{Step4-2con}. The same operation can be applied to the nodes $5$ and $6,$ as shown in (b) of Fig. \ref{Step4-2con}.
\item For the nodes $1$ and $6,$ which neither belong to the same $\Omega_i, i=1,2,$ nor belong to the same one of unfixed double nodes sets, two new edges can also be designed by connecting node $9,$ respectively, to one of the nodes $1$ and $6,$ as depicted by  (c) and (d) of Fig. \ref{Step4-2con}.
\item For one of the two unfixed double nodes sets, say the one consisting of nodes $1$ and $5,$ two new edges can be constructed by connecting the newly added node $9,$ respectively, to the nodes  $1$ and $5.$ This is illustrated by (e) of Fig. \ref{Step4-2con}. For the other unfixed double nodes set consisting of nodes  $2$ and $6,$ the same operation can be applied to it to get another group of two new edges, as shown in (f) of Fig. \ref{Step4-2con}.
\end{itemize}

\begin{figure}[H]
\begin{center}
\subfigure[]{\includegraphics[width=2.43cm]{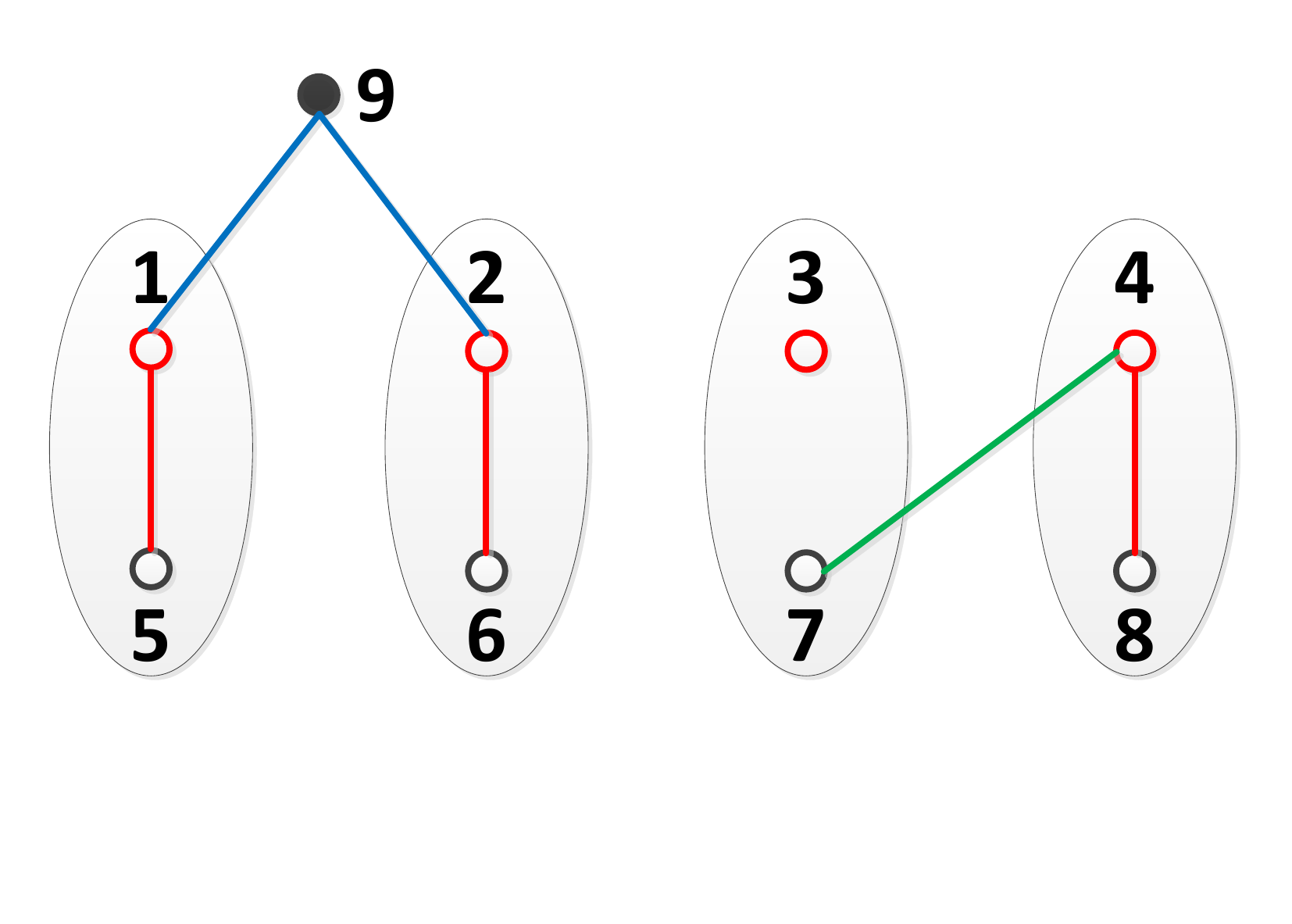}}\quad\;
\subfigure[]{\includegraphics[width=2.43cm]{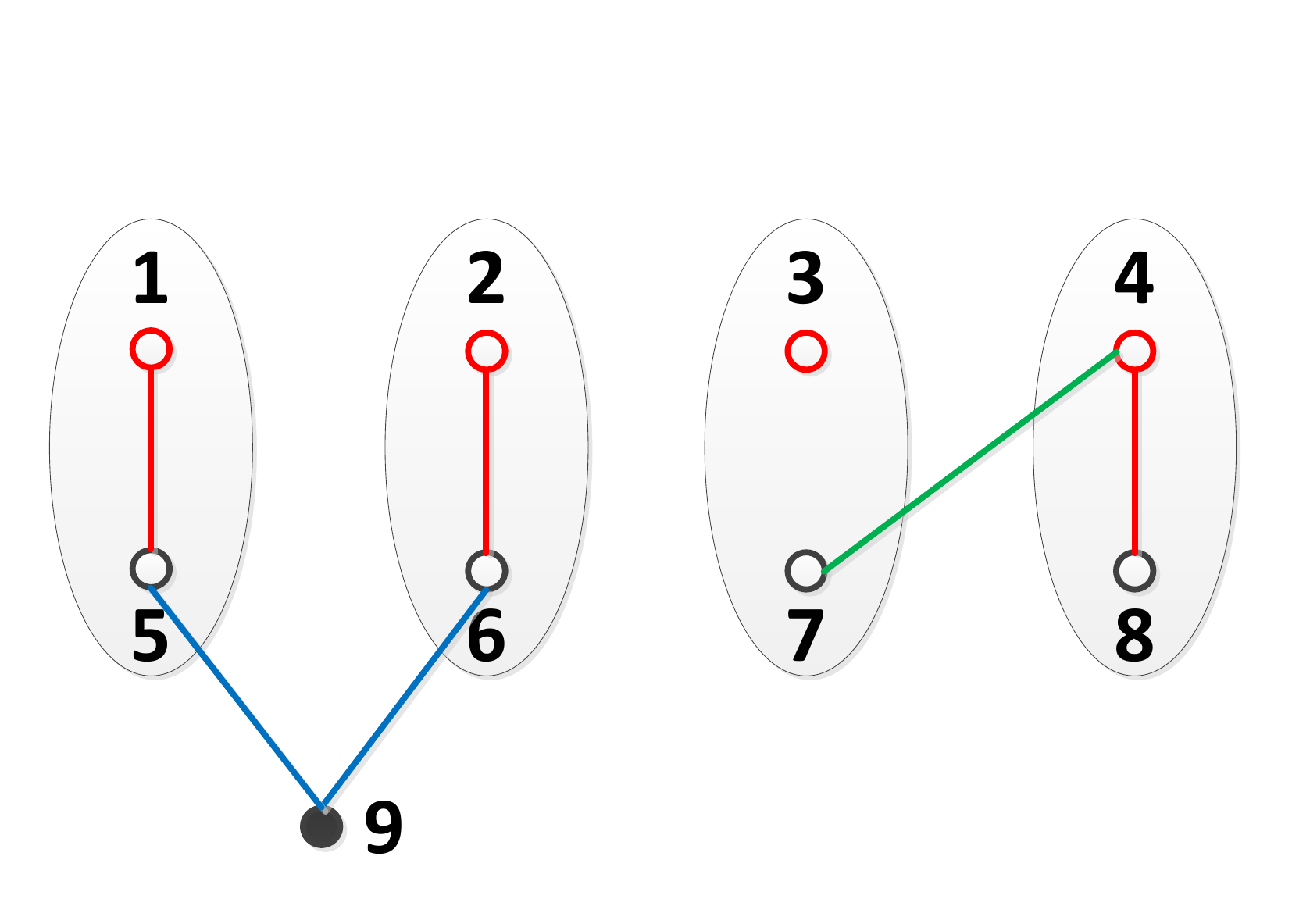}}\quad\;
\subfigure[]{\includegraphics[width=2.43cm]{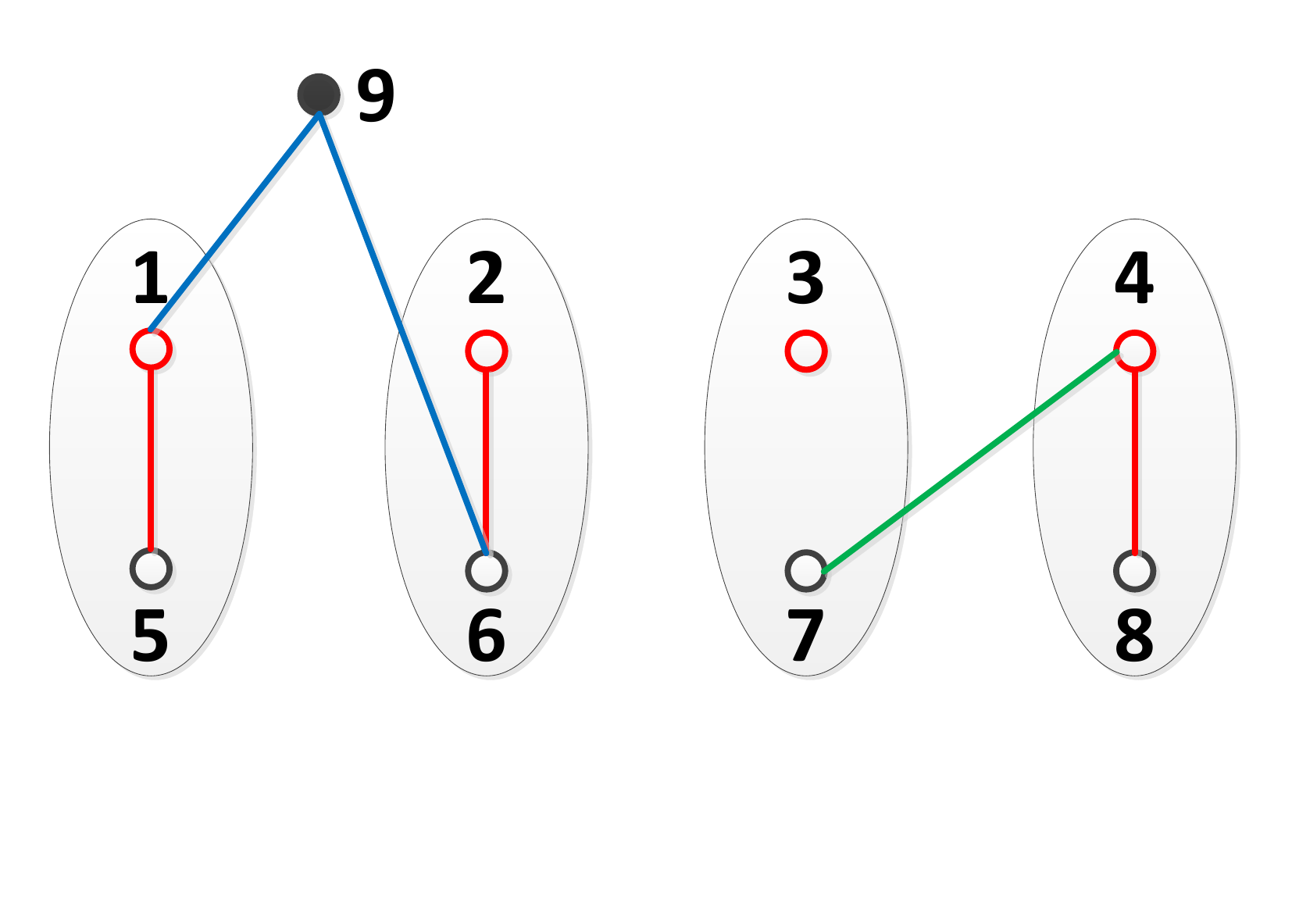}}
\subfigure[]{\includegraphics[width=2.43cm]{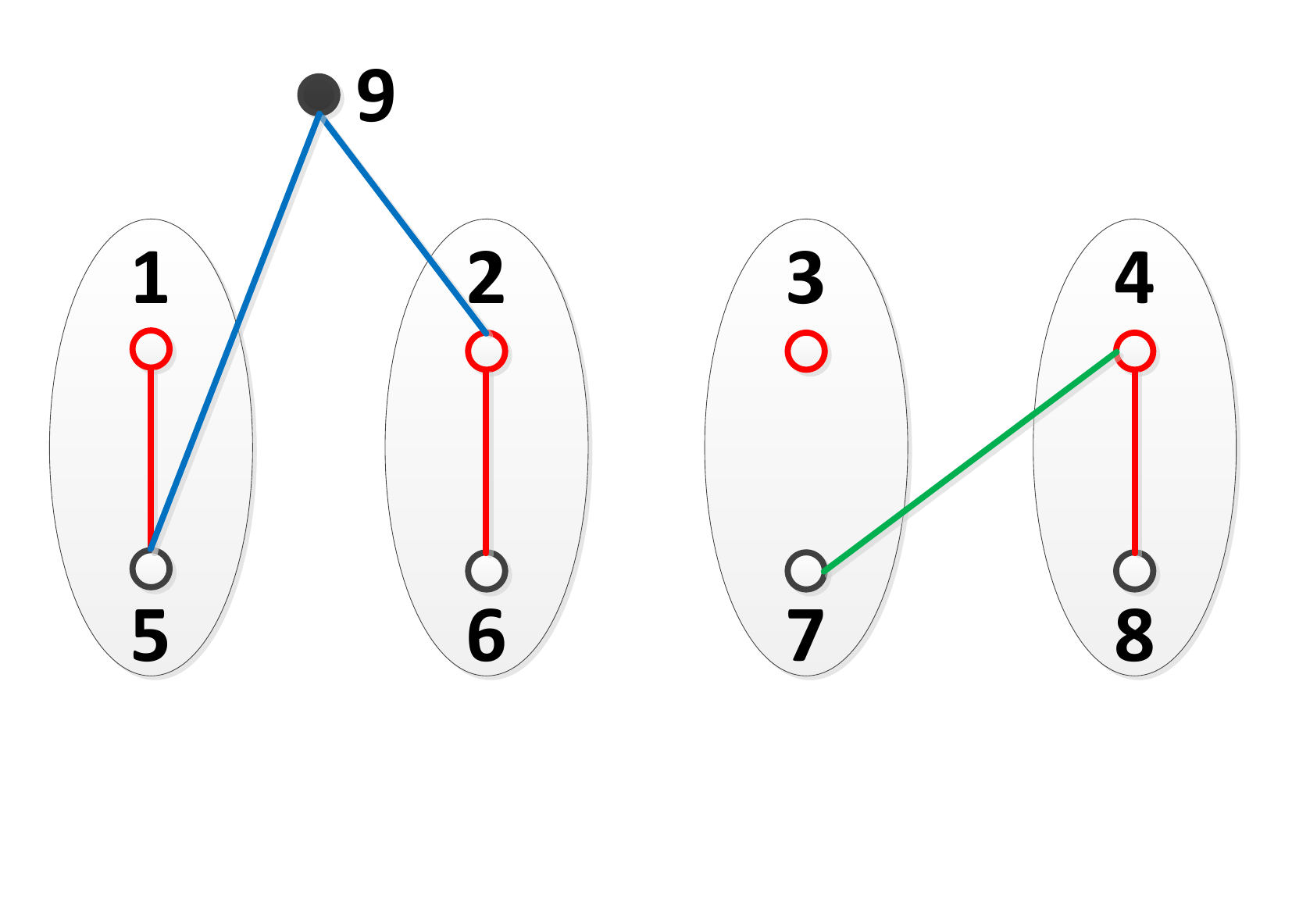}}\quad\;
\subfigure[]{\includegraphics[width=2.43cm]{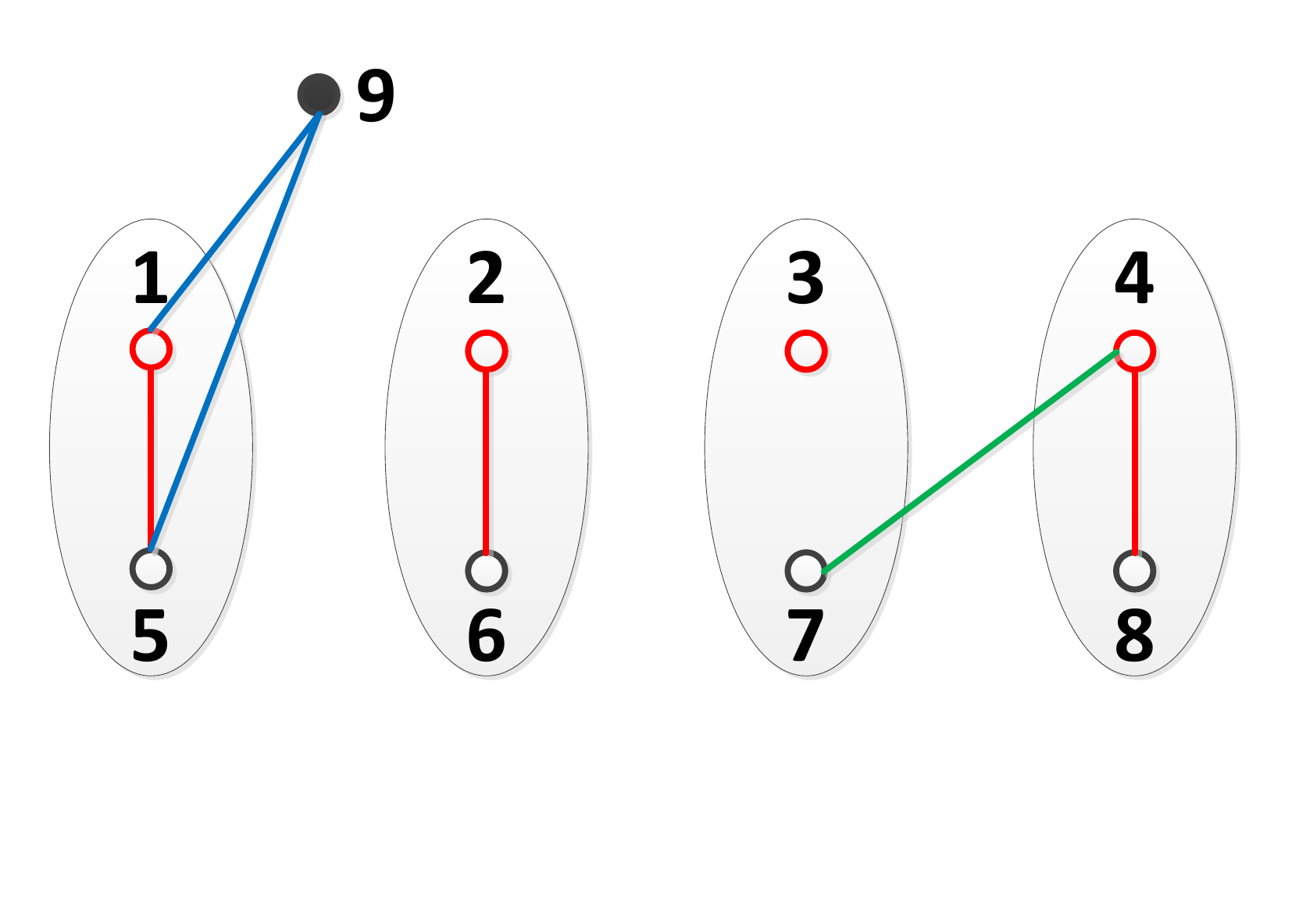}}\quad\;
\subfigure[]{\includegraphics[width=2.43cm]{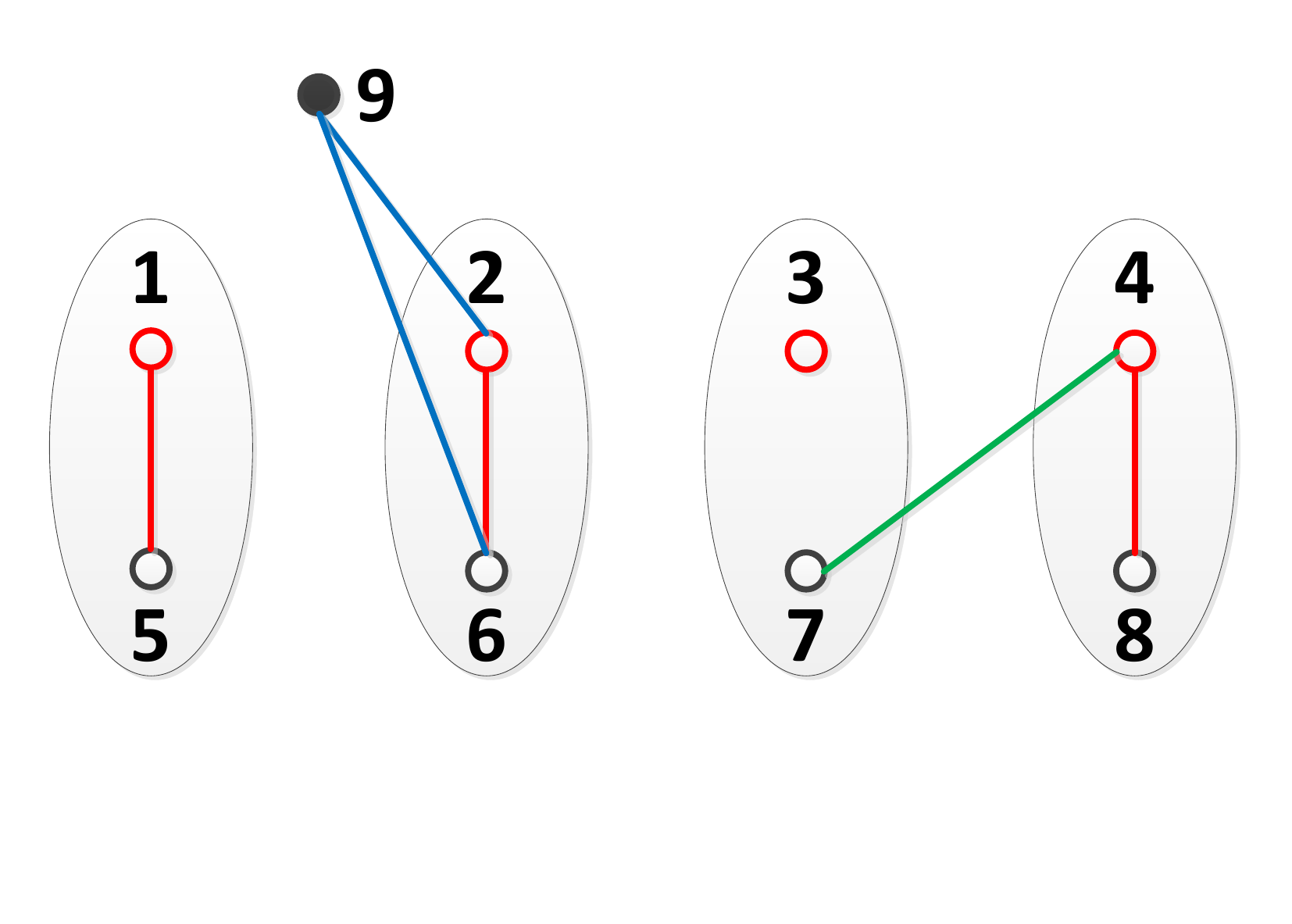}}
\caption{(a) to (f) are the figures associated with both the newly added nodes and the corresponding two newly constructed edges.}\label{Step4-2con}
\end{center}
\end{figure}
\begin{figure}[H]
\begin{center}
\subfigure[]{\includegraphics[width=2.3cm]{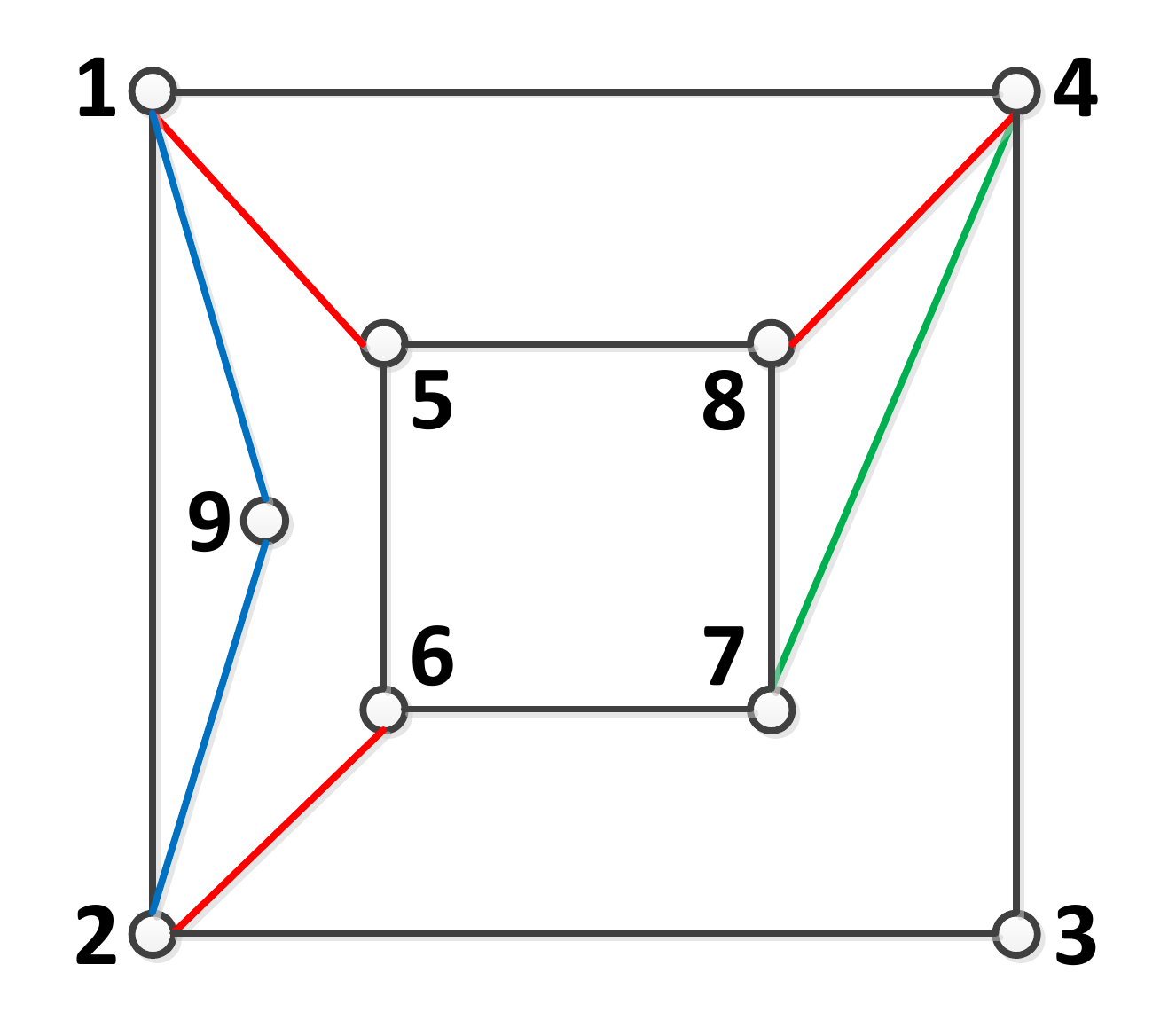}}\quad
\subfigure[]{\includegraphics[width=2.3cm]{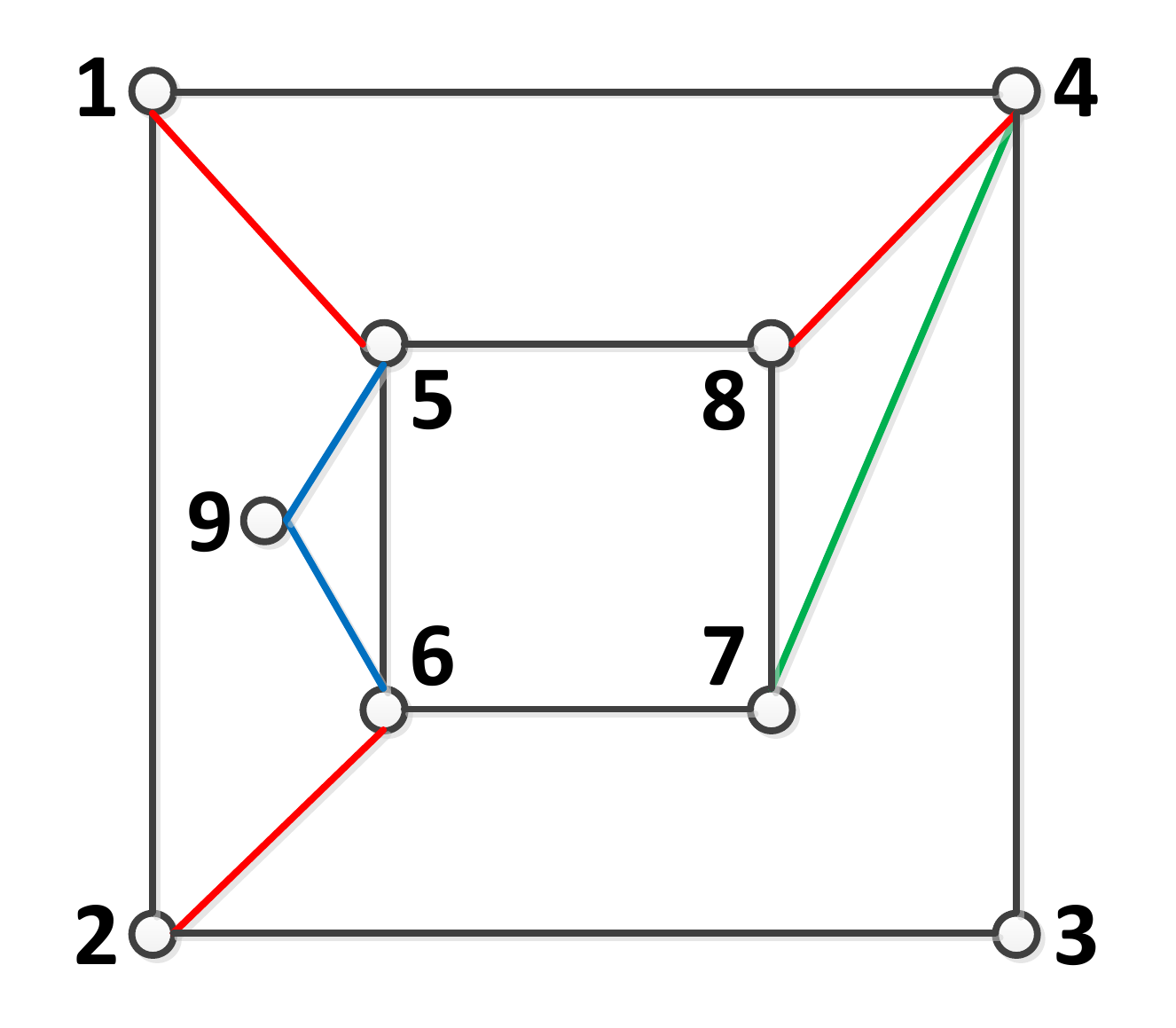}}\quad
\subfigure[]{\includegraphics[width=2.3cm]{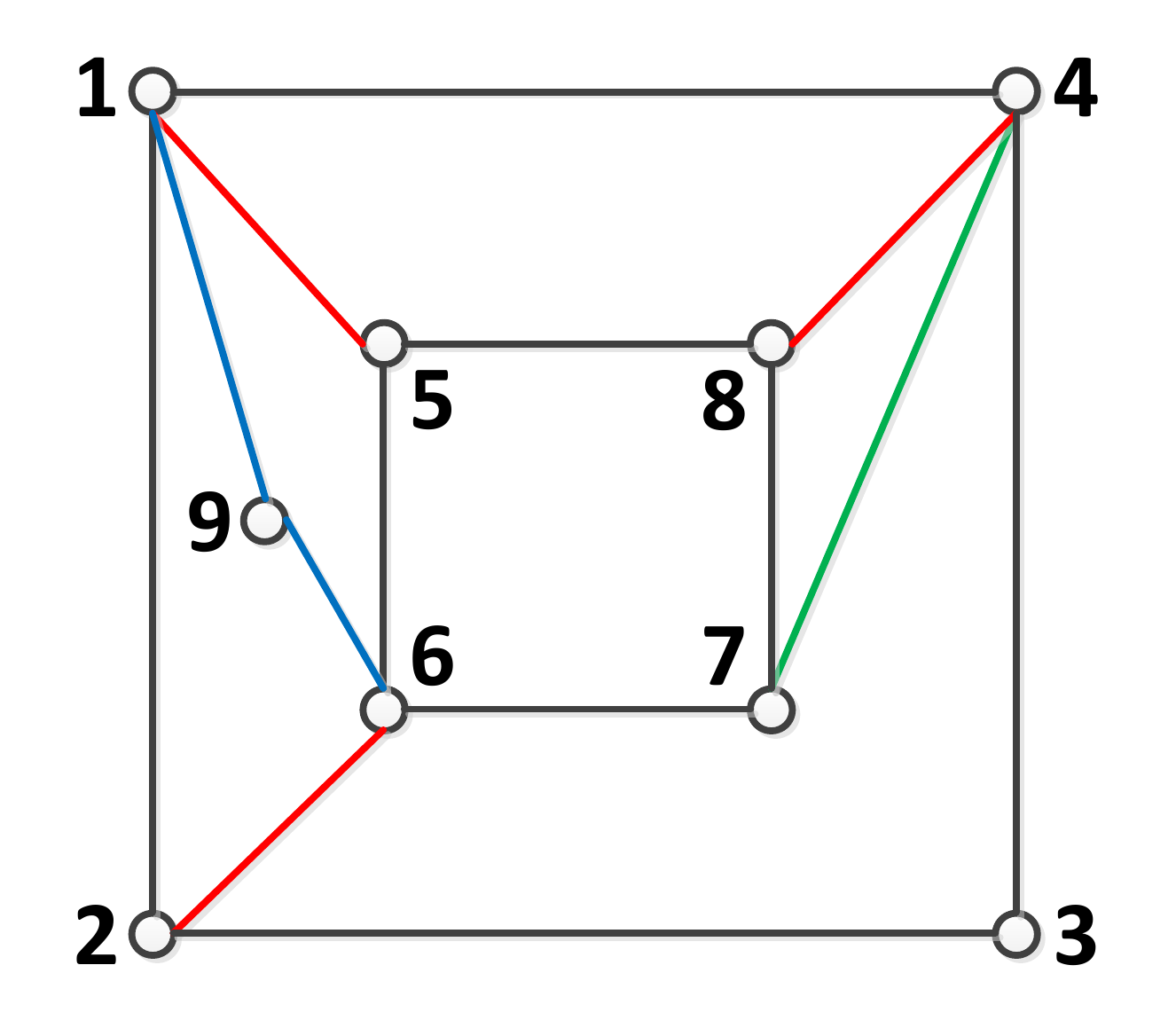}}\quad
\subfigure[]{\includegraphics[width=2.3cm]{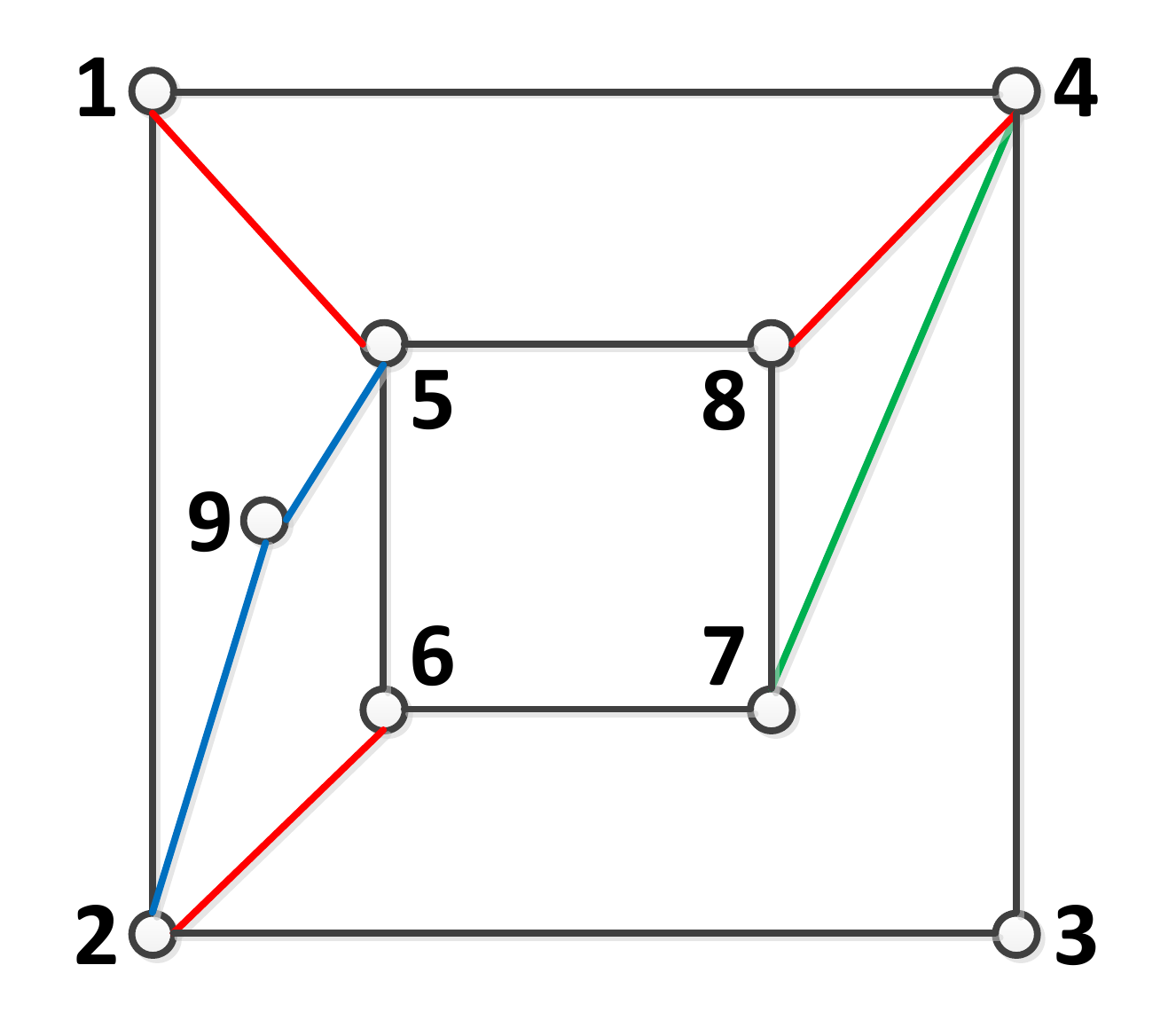}}\quad
\subfigure[]{\includegraphics[width=2.3cm]{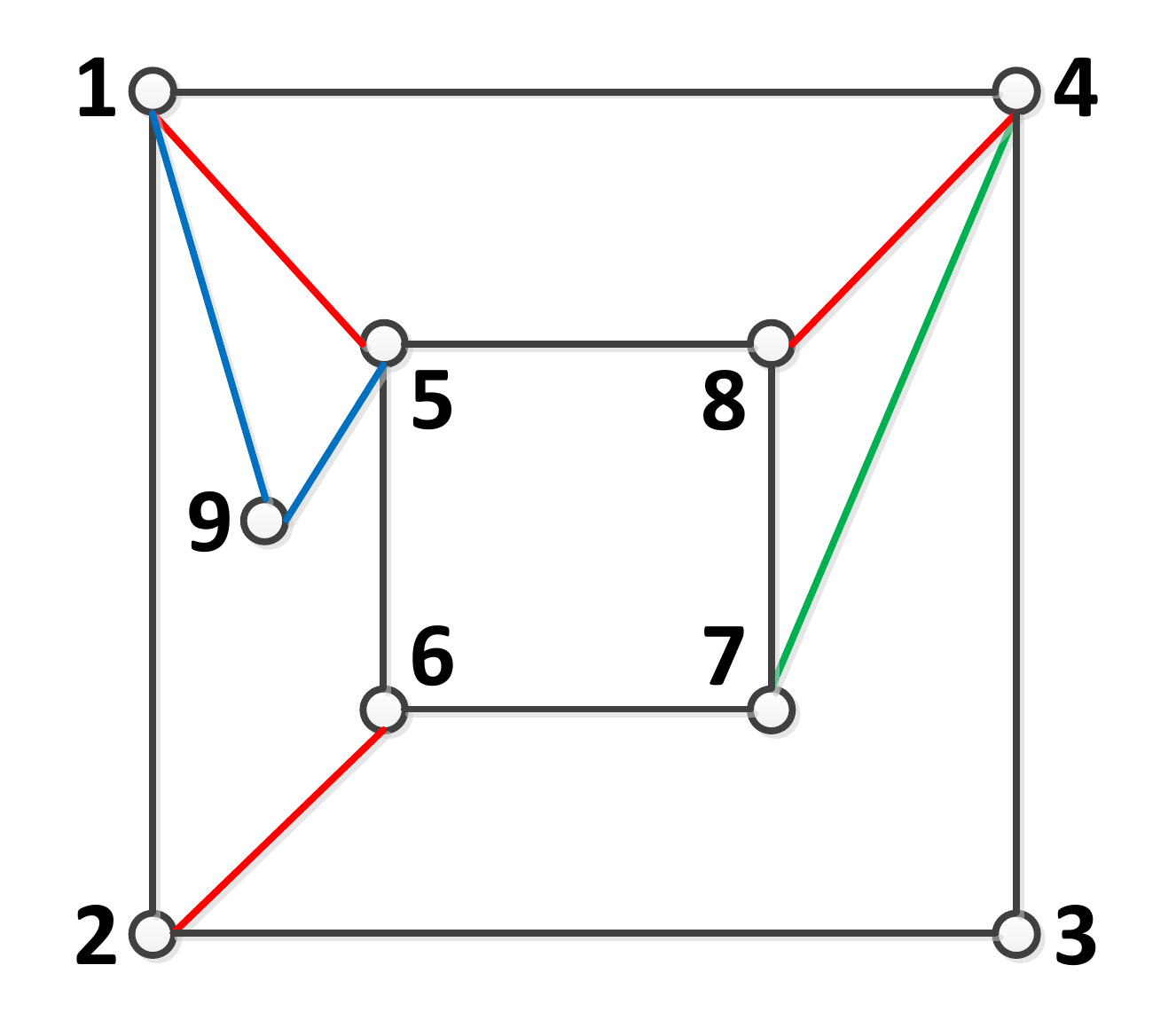}}\quad
\subfigure[]{\includegraphics[width=2.3cm]{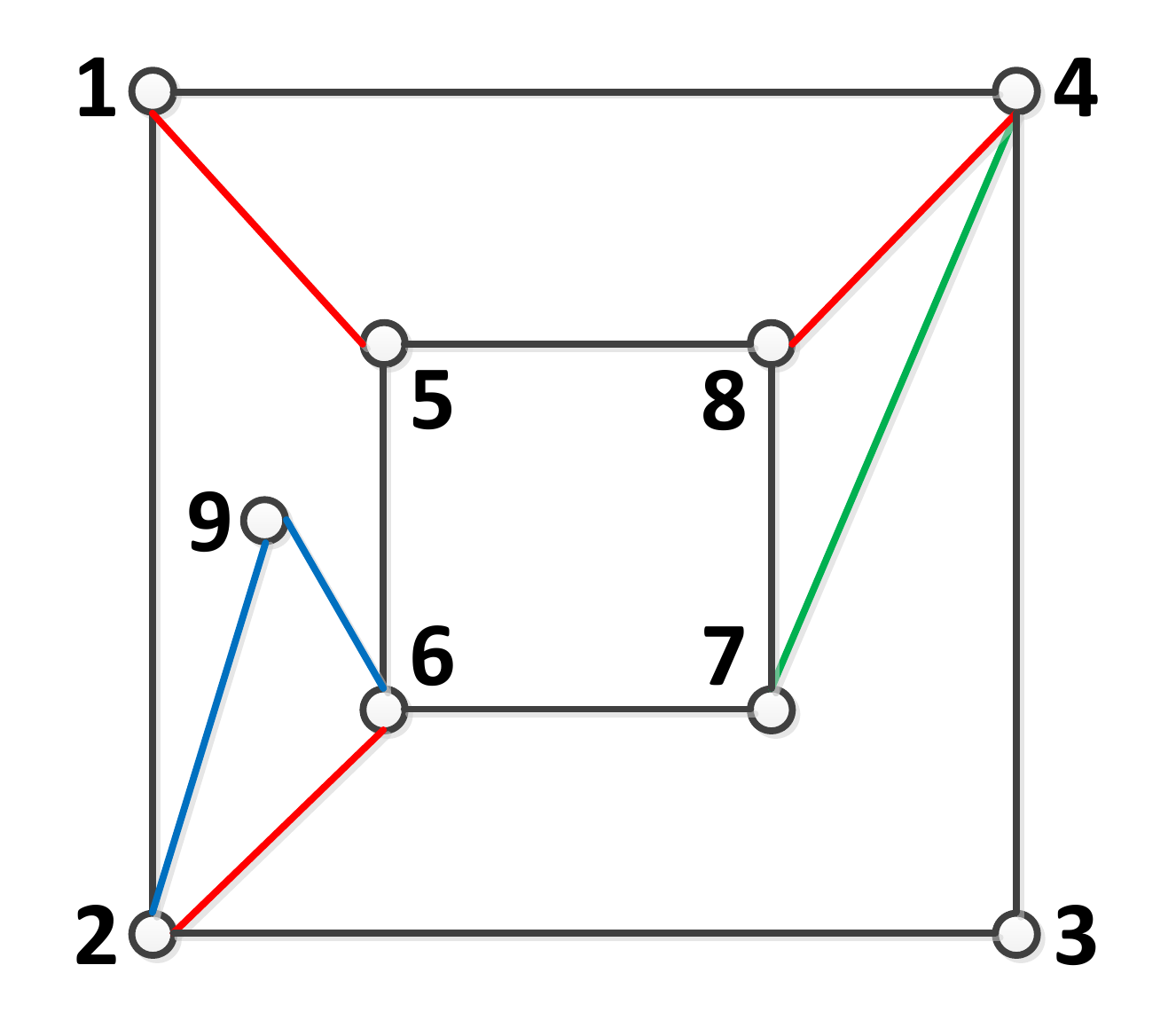}}\quad
\caption{(a) to (f) are the topology structures corresponding, respectively, to each of the figures in Fig. \ref{Step4-2con}.}\label{Step4-2conN}
\end{center}
\end{figure}

\textbf{Case c)}
In this case, we are to design the third new edge. It can be seen from Fig. \ref{Step4-3con}  that the third edge is designed by connecting the new node $9$ to one of the nodes in the unfixed double nodes set, where any node in this unfixed double nodes set previously has no linking edge with the new node $9.$

The topology structures corresponding, respectively, to (a) to (d) of Fig. \ref{Step4-3con} are depicted by (a) to (d) of Fig. \ref{Step4-3conN}.

Although the following Fig. \ref{Step4-3con} indicates that the difference between case b) and this case c) is only the third new edge, the two cases actually are parallel in the sense that both cases can be employed below to further produce different completely controllable graphs. That is why we make the distinction between the two cases.

\begin{figure}[H]
\begin{center}
\subfigure[]{\includegraphics[width=2.34cm]{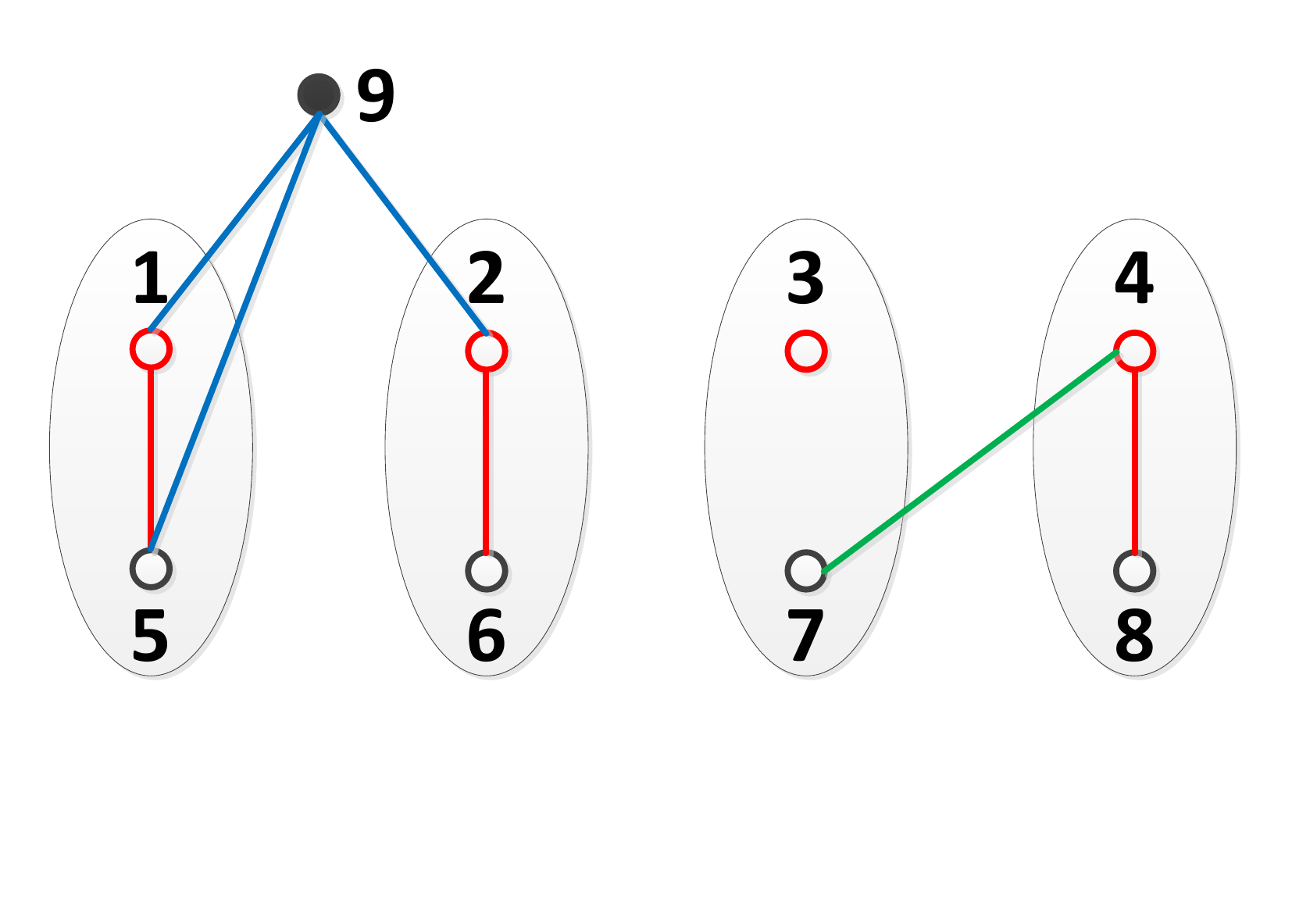}}\qquad
\subfigure[]{\includegraphics[width=2.34cm]{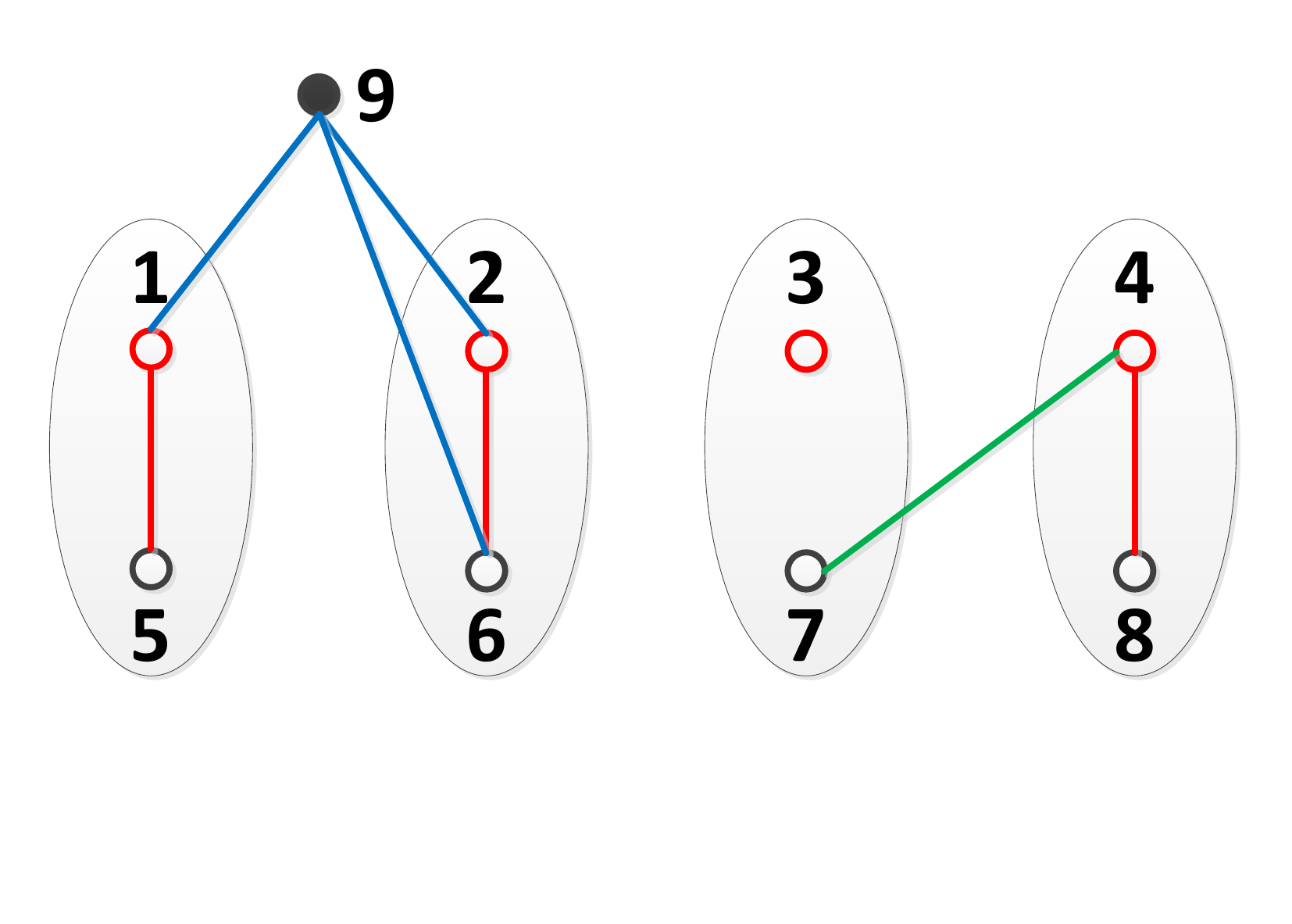}}\qquad\\
\subfigure[]{\includegraphics[width=2.34cm]{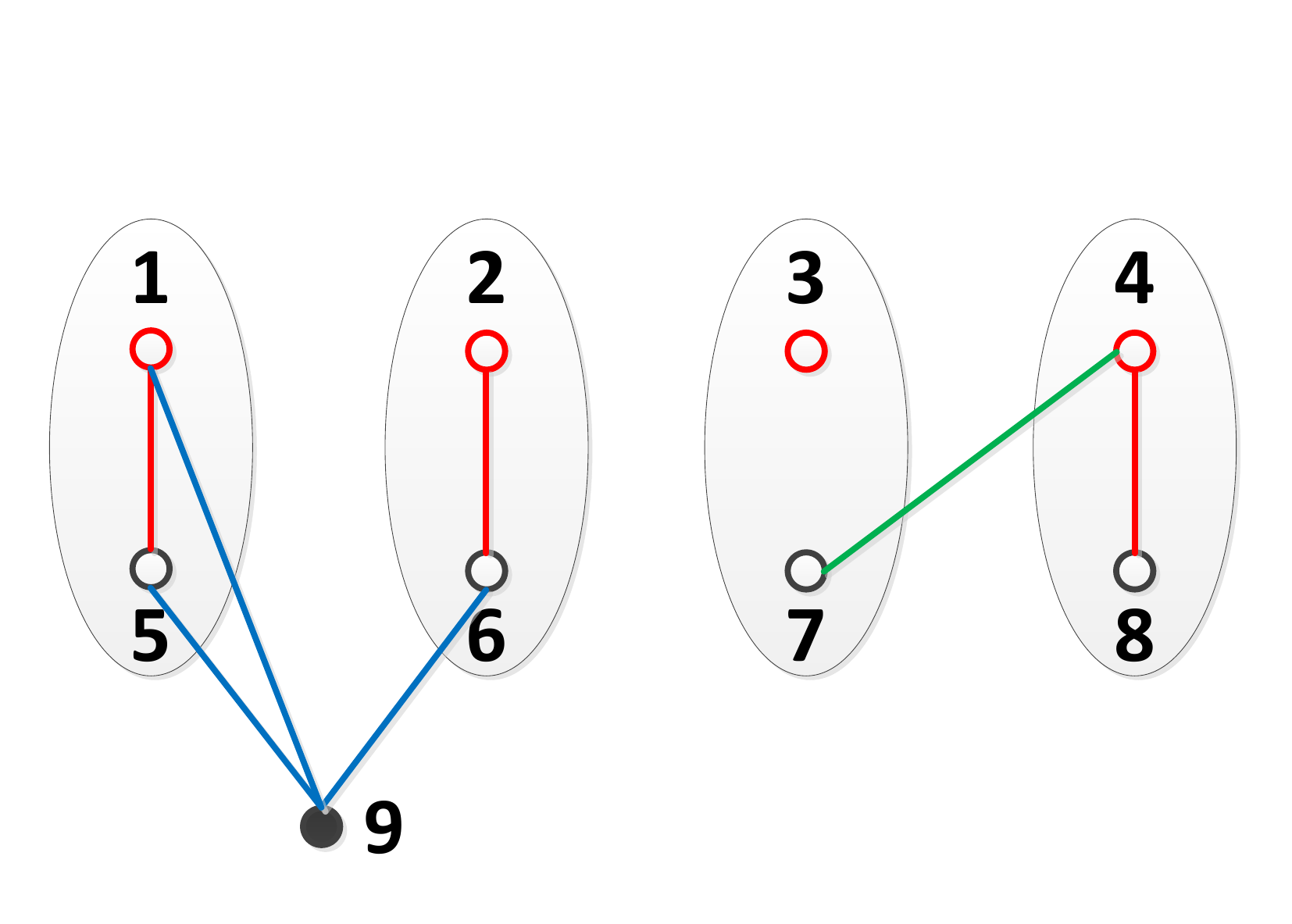}}\qquad
\subfigure[]{\includegraphics[width=2.34cm]{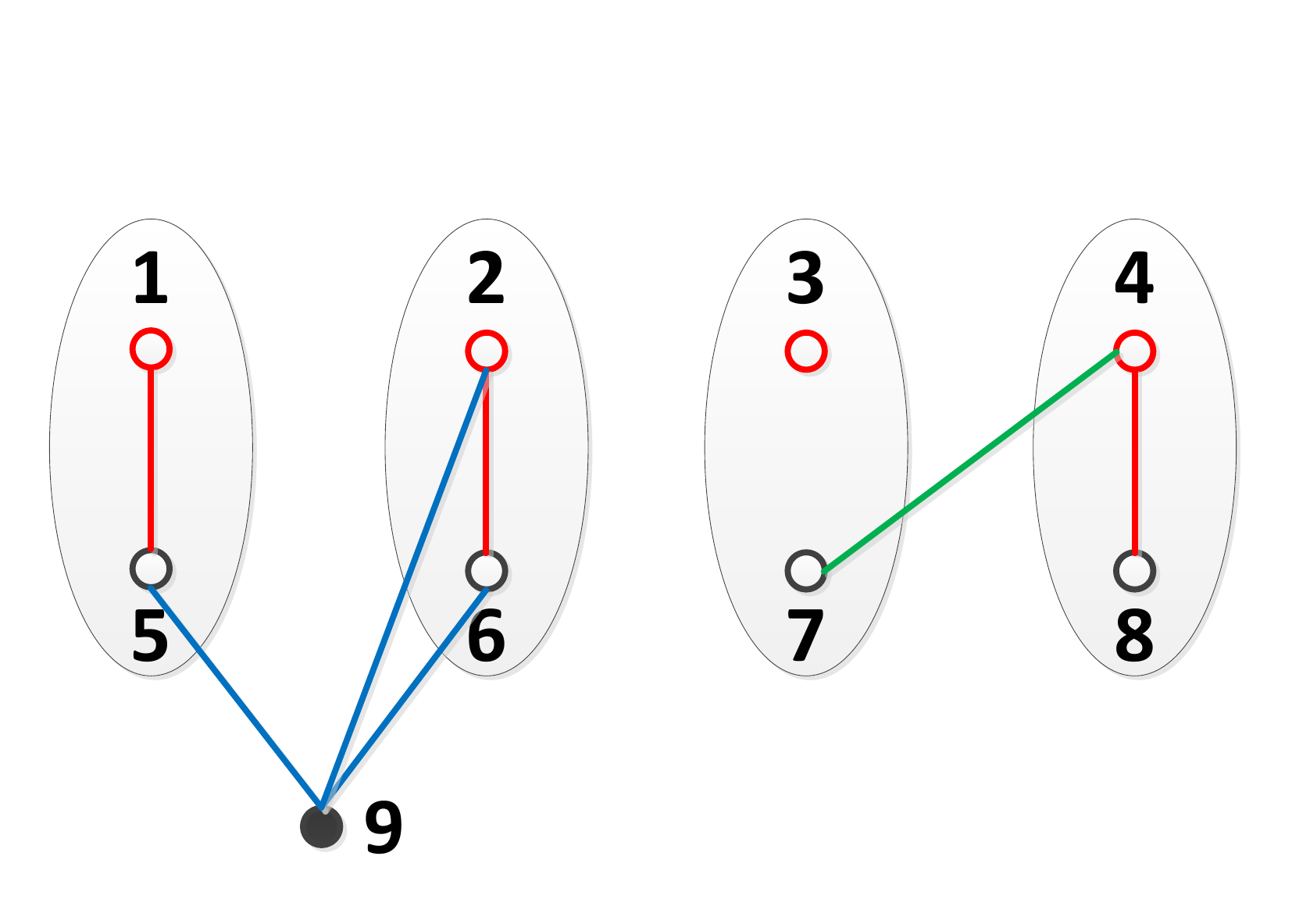}}\qquad
\caption{(a) to (d) are the figures with the third newly added edge.}\label{Step4-3con}
\end{center}
\end{figure}
\begin{figure}[H]
\begin{center}
\subfigure[]{\includegraphics[width=2.24cm]{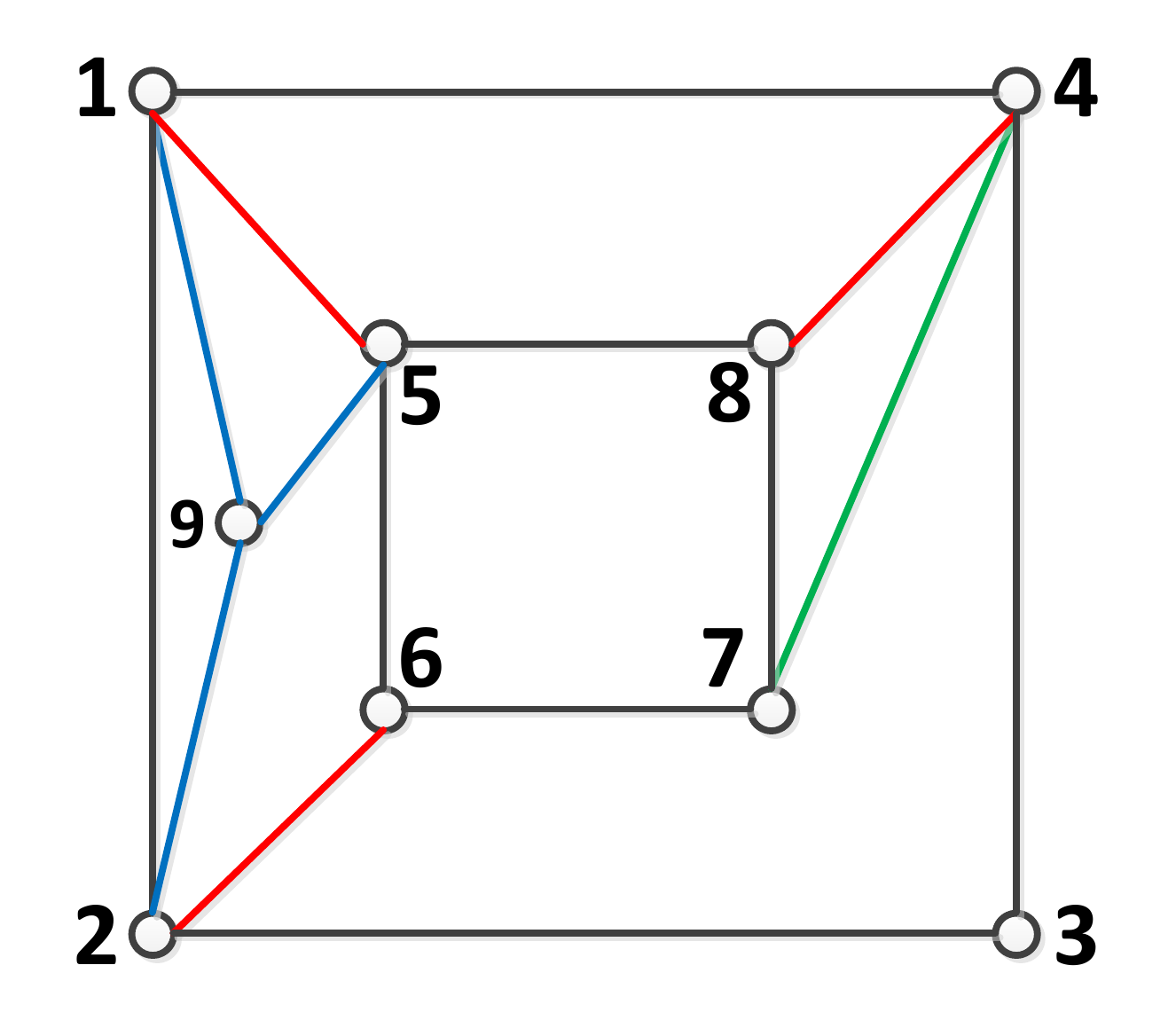}}
\subfigure[]{\includegraphics[width=2.24cm]{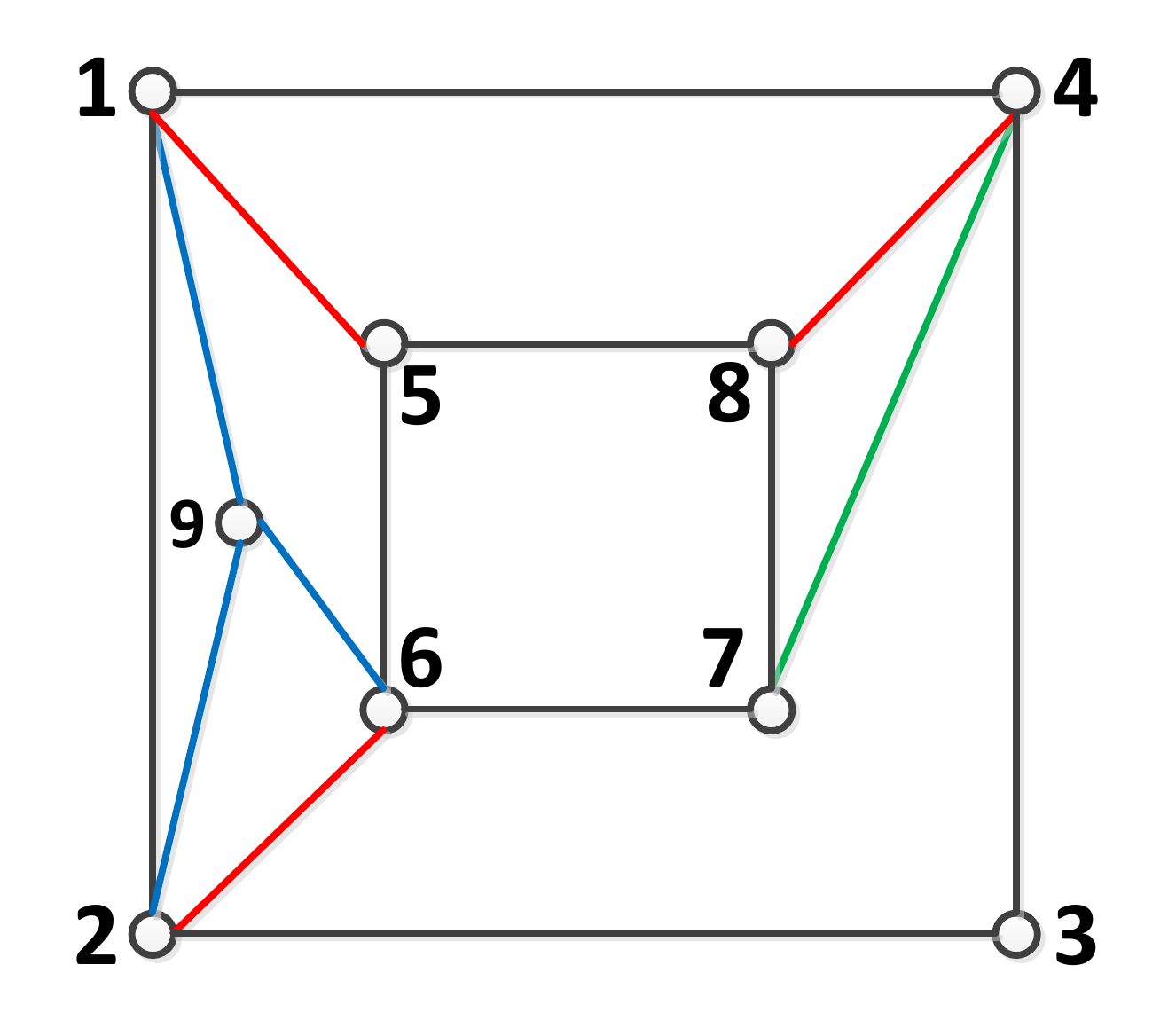}}
\subfigure[]{\includegraphics[width=2.24cm]{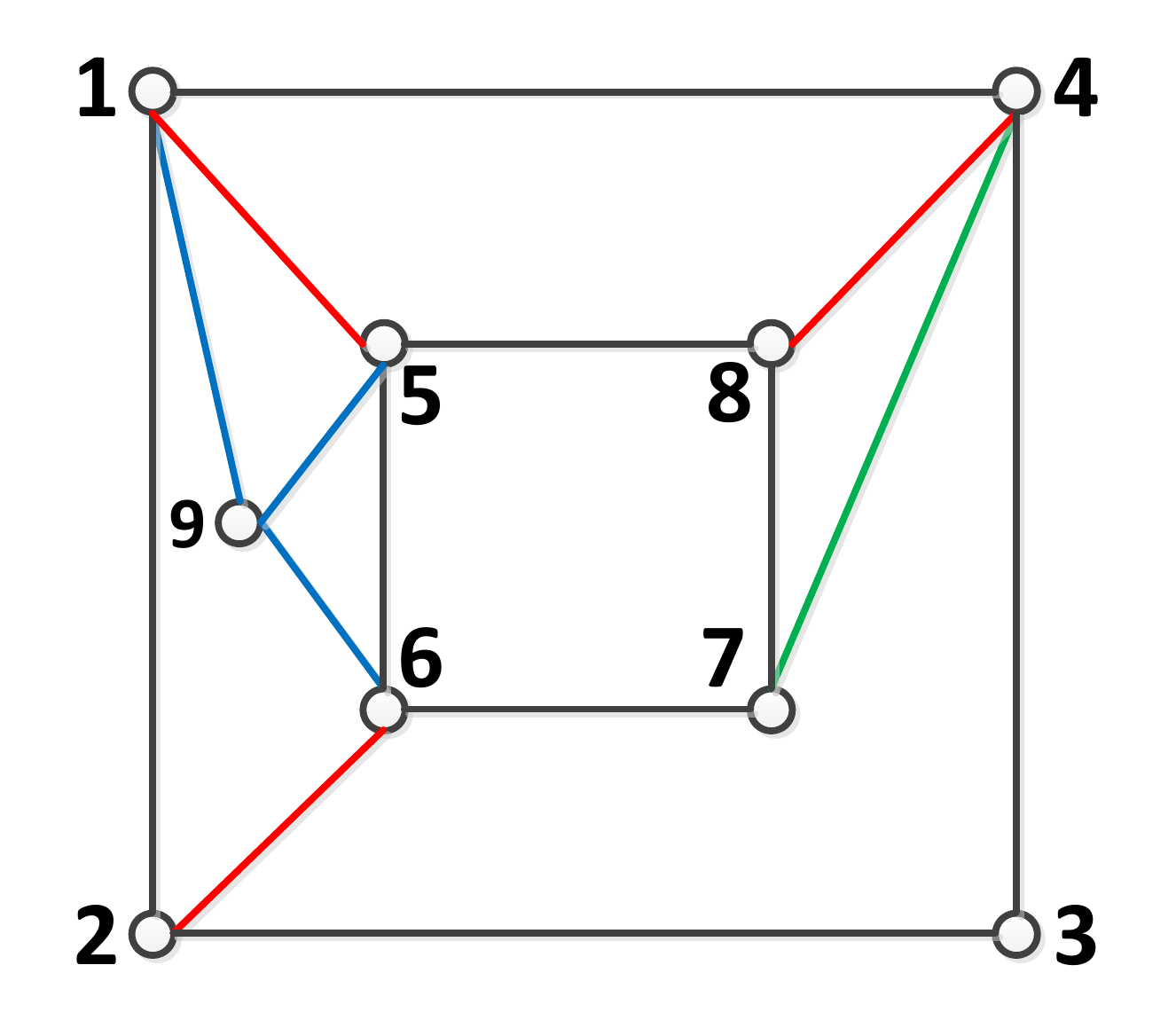}}
\subfigure[]{\includegraphics[width=2.24cm]{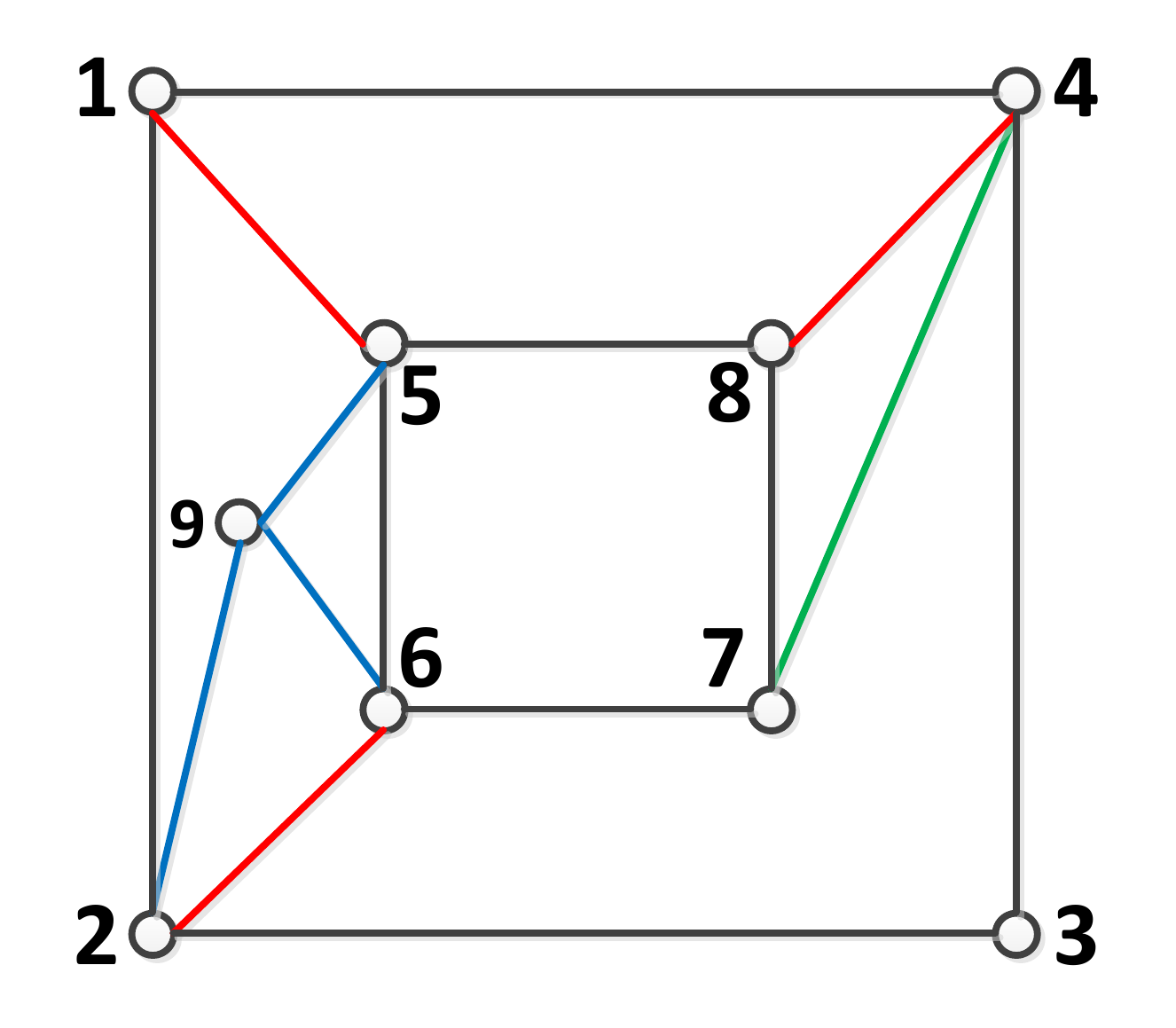}}
\caption{(a) to (d) are the topology structures corresponding, respectively, to each of the figures in Fig. \ref{Step4-3con}.}\label{Step4-3conN}
\end{center}
\end{figure}

\begin{theorem}\label{Theo2topo}
The communication graphs (a), (b), (d), (e), (f) of Fig. \ref{Step4-2conN} and (a), (b), (c) of Fig. \ref{Step4-3conN}   are perfectly controllable.
\end{theorem}
\begin{proof}
The proof is conducted only for graph (a) of Fig. \ref{Step4-2conN}. The other graphs can be proved in the same way.

For graph (a) of Fig. \ref{Step4-2conN}, calculations show that the following two aspects hold simultaneously.
\begin{itemize}
\item Eigenvalues of the Laplacian $L$ of this graph  are distinct from each other, which are, respectively, $0, 1.1834, 1.6463, 2.4581, 2.7853, 3.9468, 4.5771, 5.1780$ and $6.2250;$
\item The eigenvector associated with each of these eigenvalues has no zero entries. Note that since the graph is undirected, there is one and only one independent eigenvector corresponding to each of the aforementioned eigenvalues.
\end{itemize}
By Theorem \ref{The1Alge}, the system is controllable under any selection of leaders, where both the number and the locations of leaders are arbitrary.
\end{proof}

To create more perfectly controllable graphs, we perform the following construction process.

\textbf{Step 5}
In this step, topology structures are further designed by combining Fig. \ref{Step4-1con} and Fig. \ref{Step4-2con}. More specifically, the graphs (a) and (b) of Fig. \ref{Step5-1con} is obtained by combing (a) of Fig. \ref{Step4-2con}, respectively, with (a) and (b) of Fig. \ref{Step4-1con}. Similarly, (c) and (d) of Fig. \ref{Step5-1con} is obtained by combing (b) of Fig. \ref{Step4-2con}, respectively, with (a) and (b) of Fig. \ref{Step4-1con}. For each graph of Fig. \ref{Step4-2con}, the same procedure is repeated to get two new graphs in Fig. \ref{Step5-1con}. In this way, we have a total of 12 graphs in Fig. \ref{Step5-1con}.

\begin{figure}[H]
\begin{center}
\subfigure[]{\includegraphics[width=2.43cm]{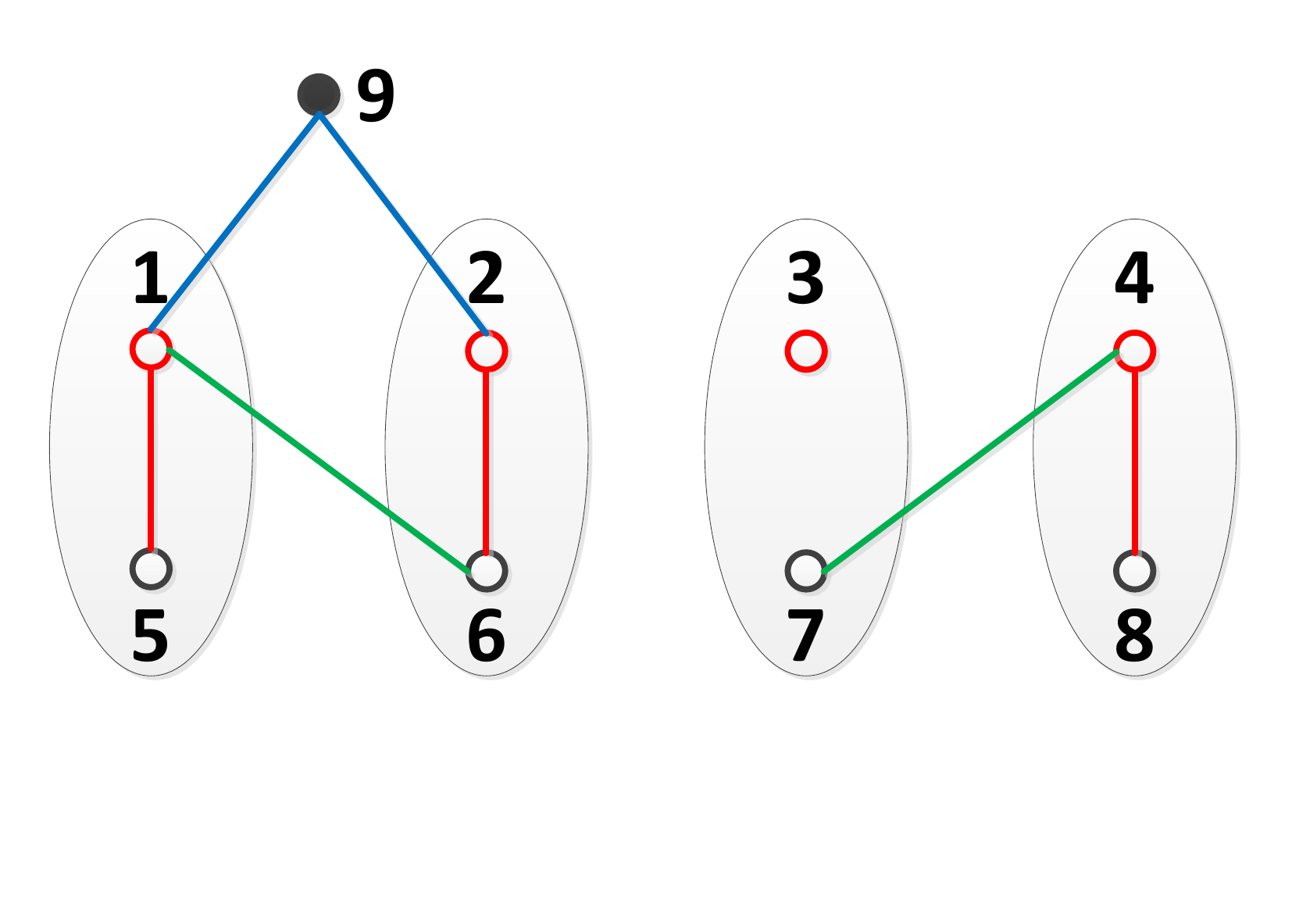}}\qquad
\subfigure[]{\includegraphics[width=2.43cm]{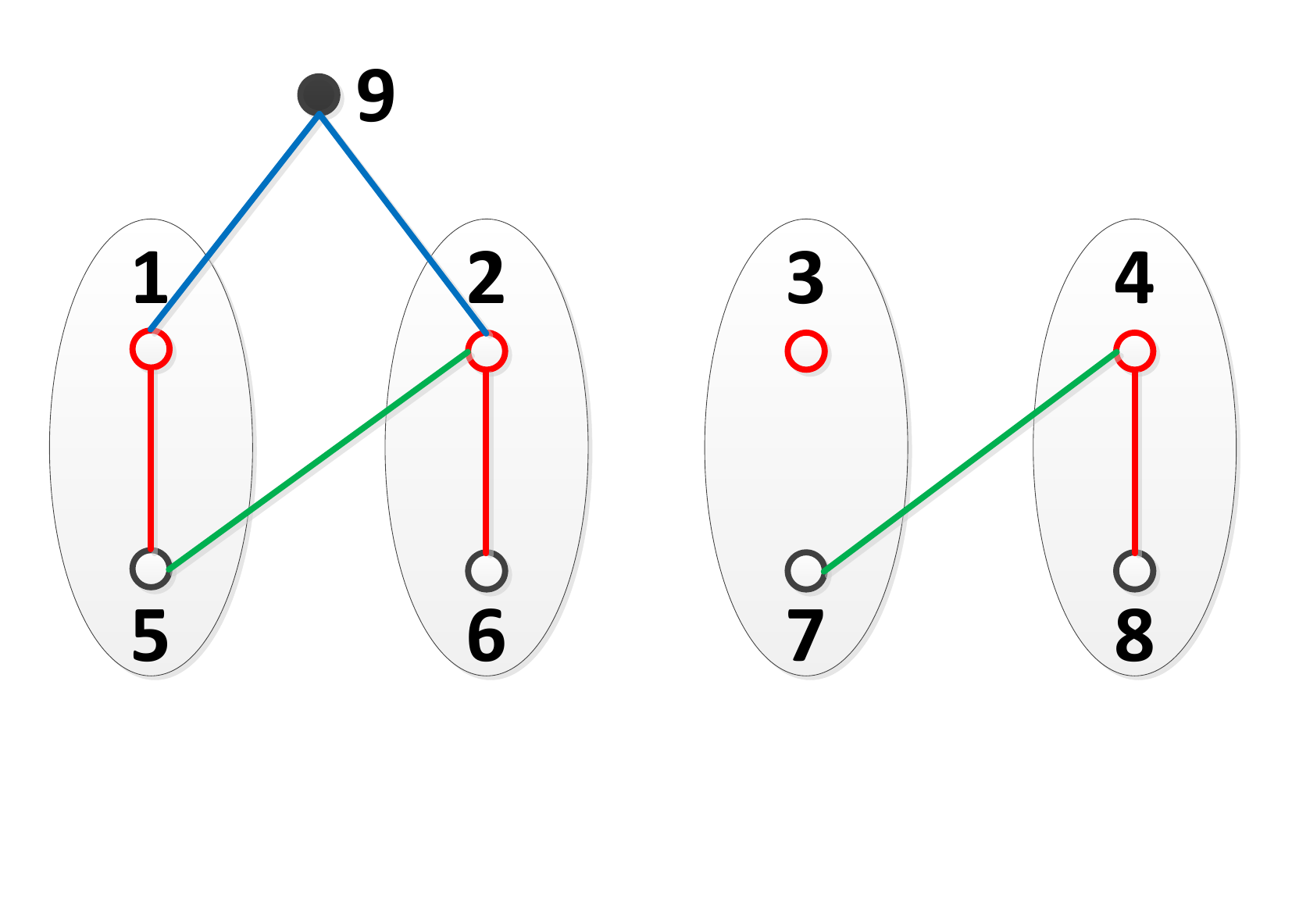}}\qquad
\subfigure[]{\includegraphics[width=2.43cm]{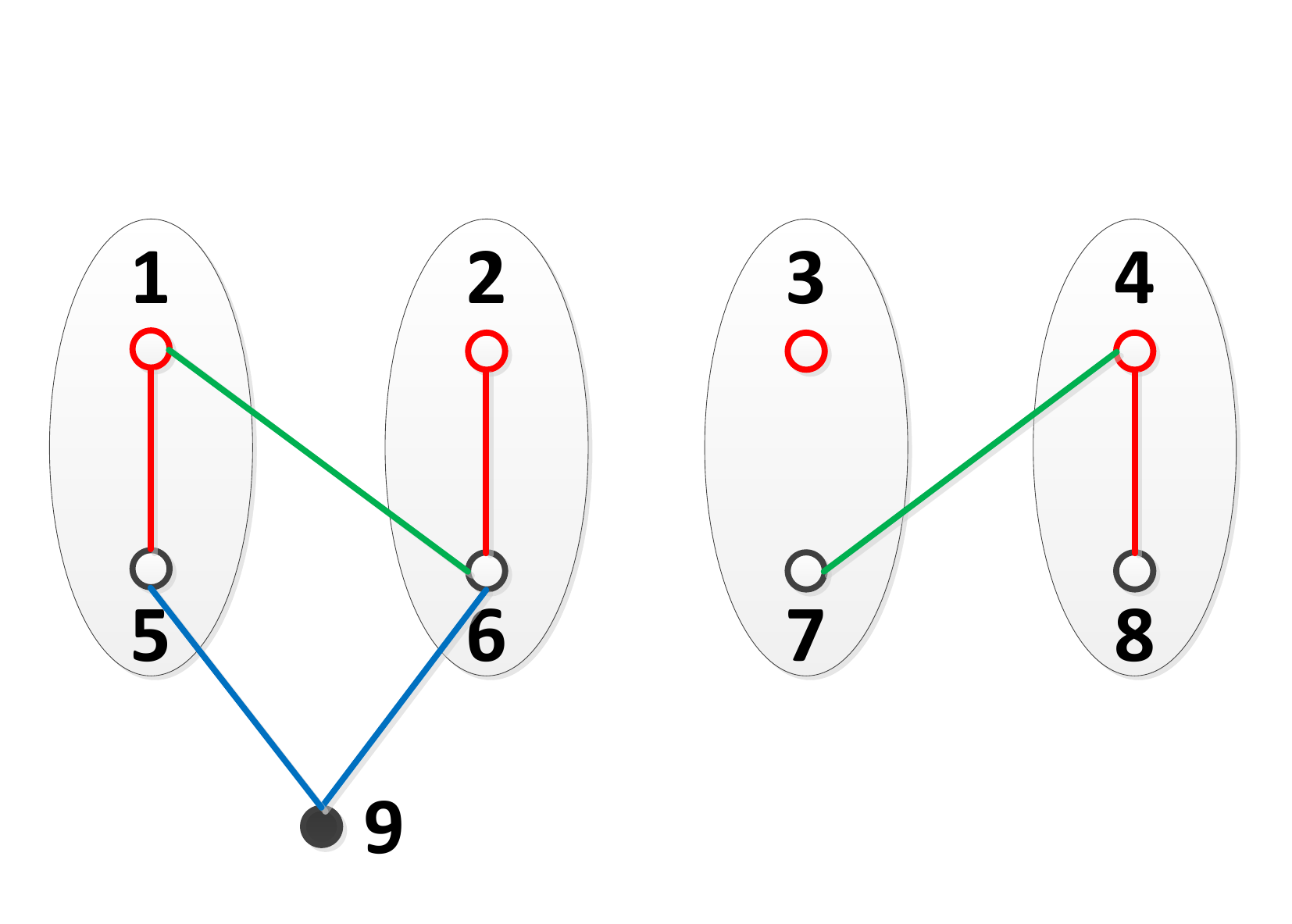}}
\subfigure[]{\includegraphics[width=2.43cm]{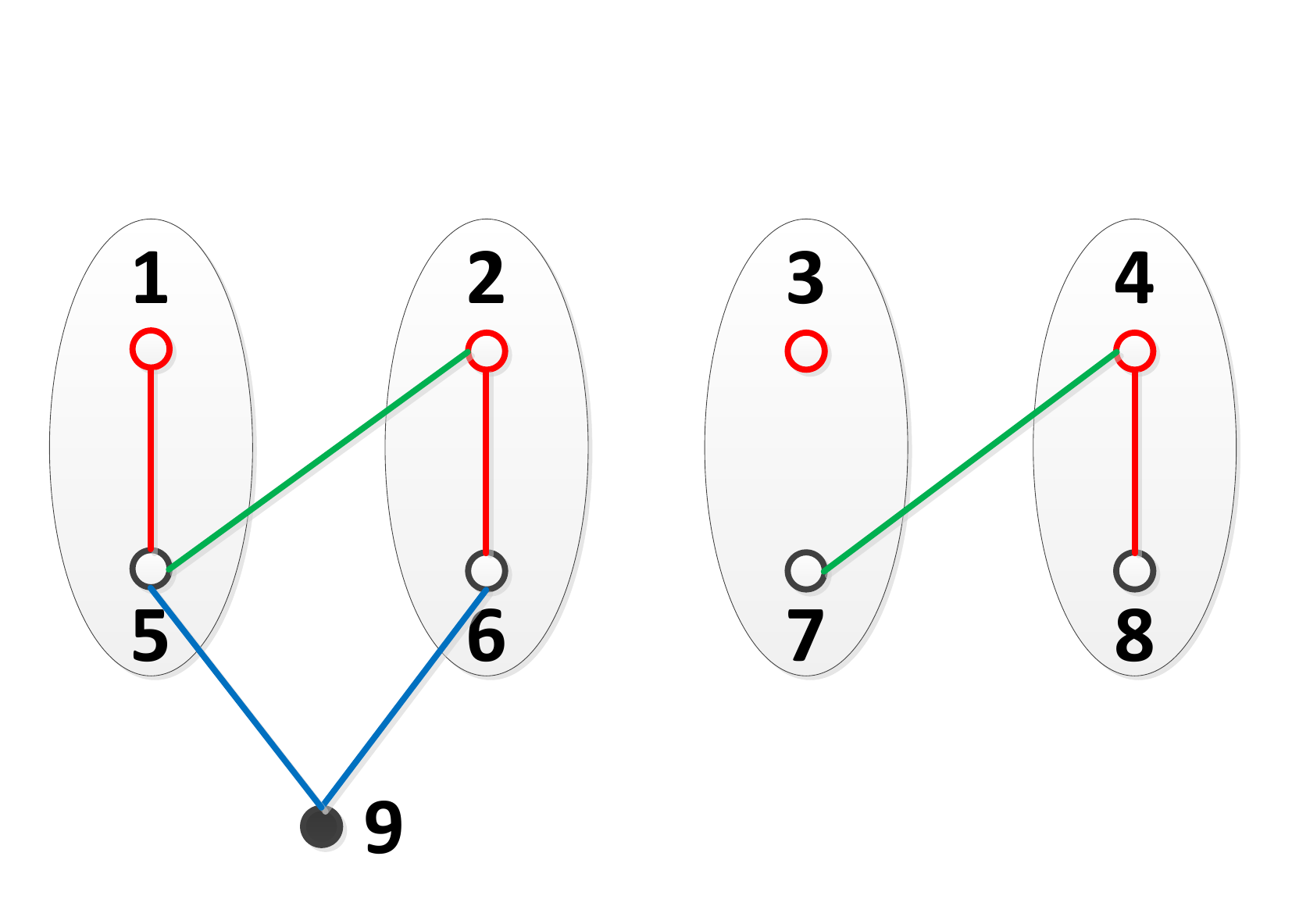}}\qquad
\subfigure[]{\includegraphics[width=2.43cm]{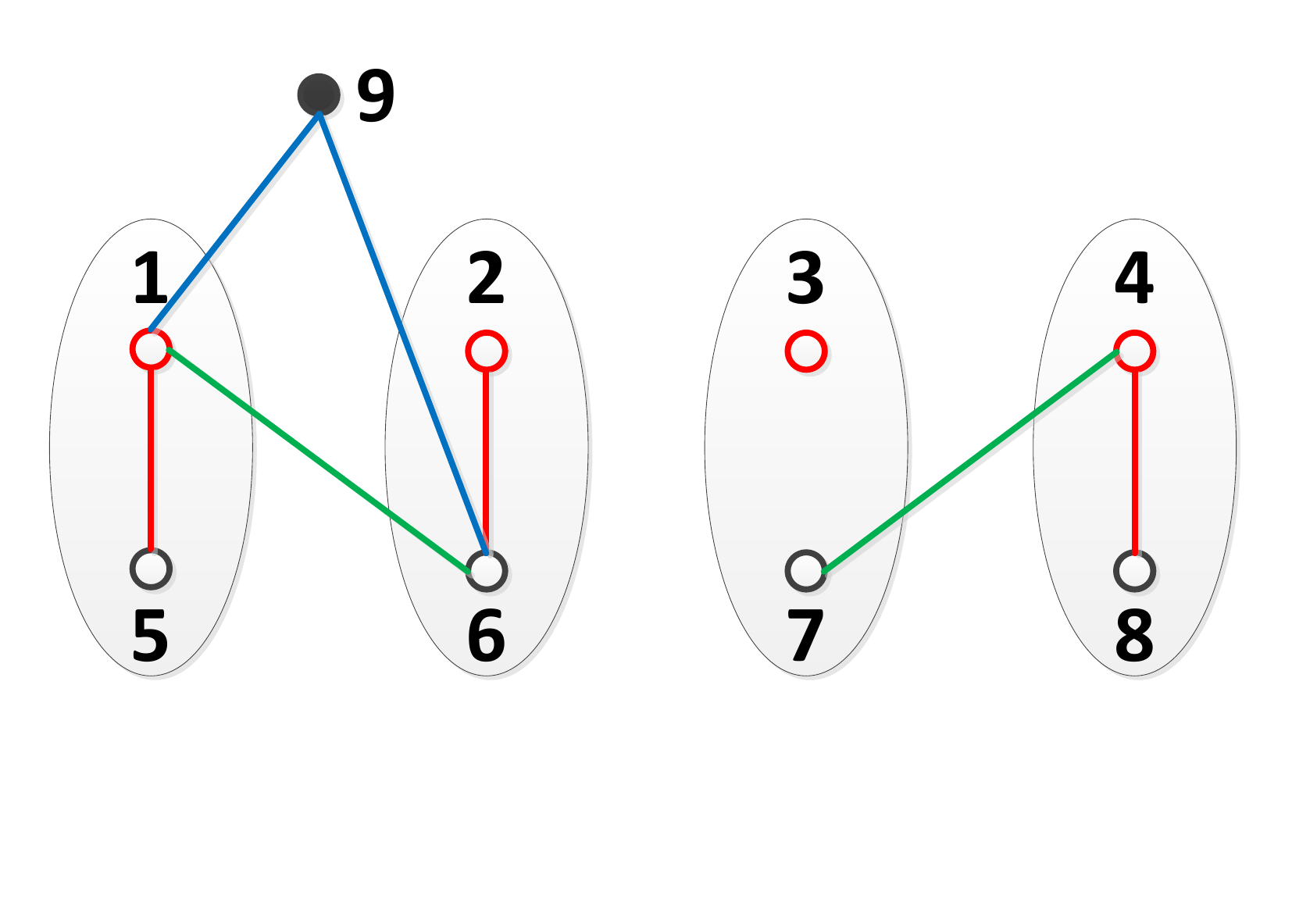}}\qquad
\subfigure[]{\includegraphics[width=2.43cm]{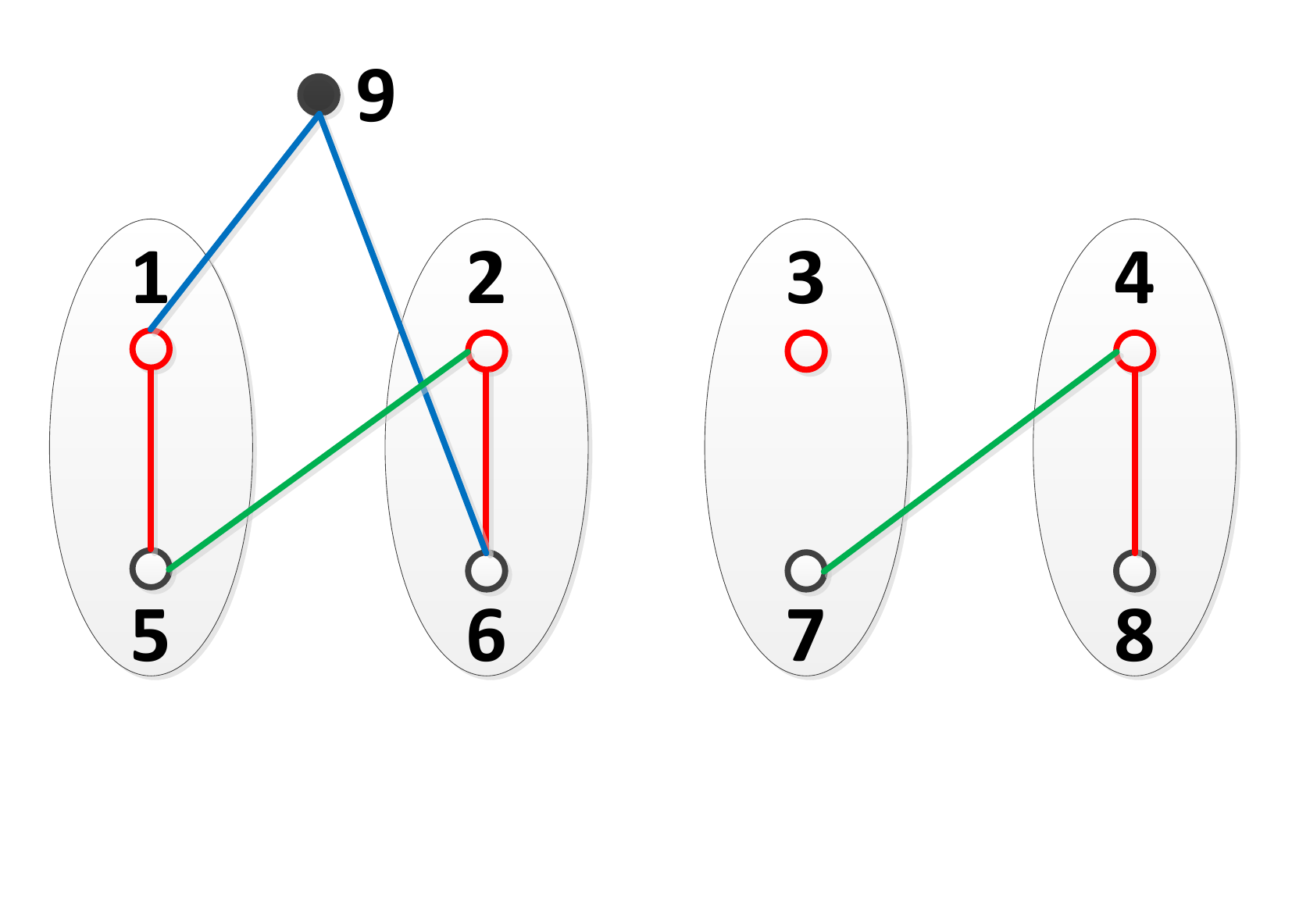}}
\subfigure[]{\includegraphics[width=2.43cm]{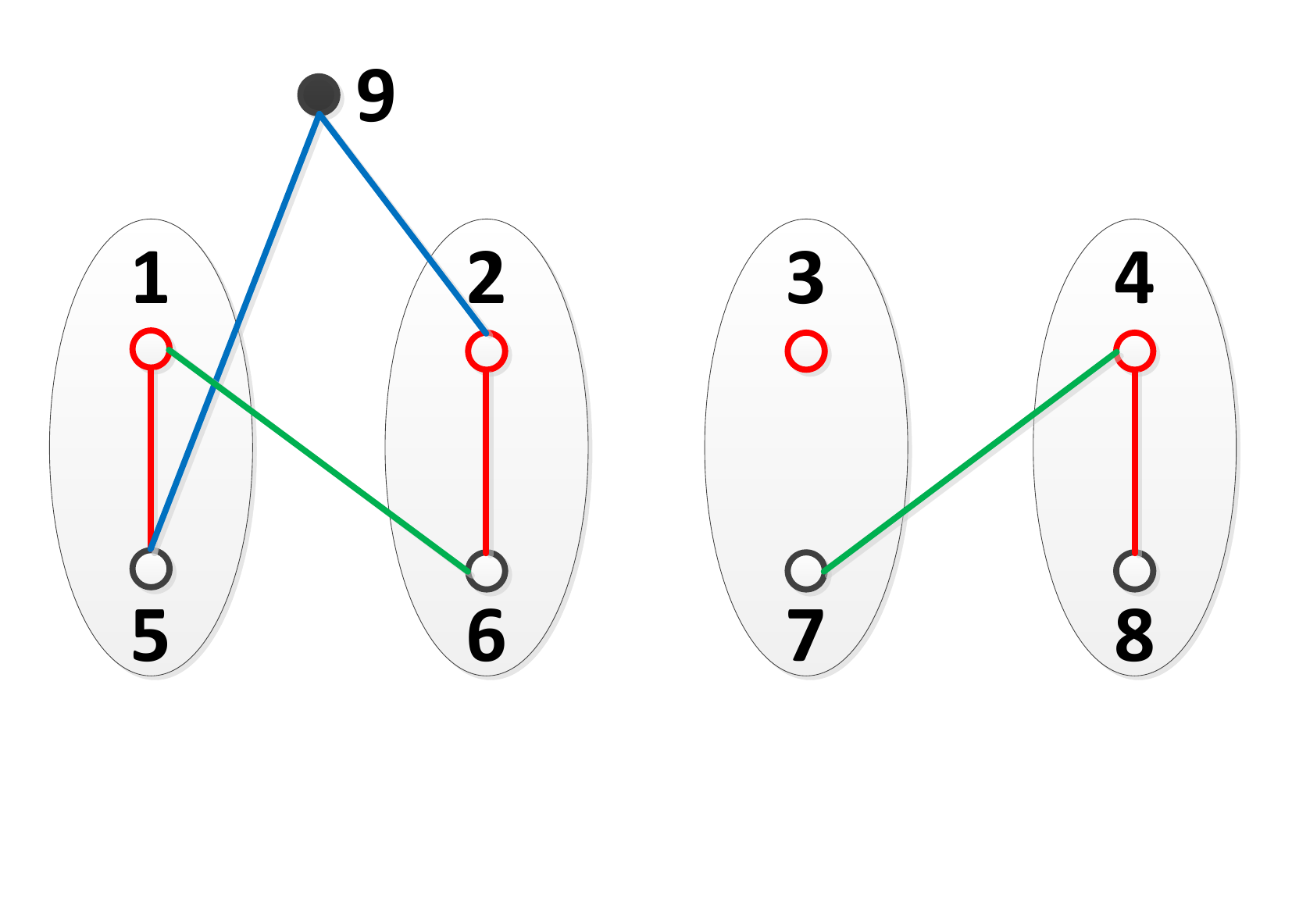}}\qquad
\subfigure[]{\includegraphics[width=2.43cm]{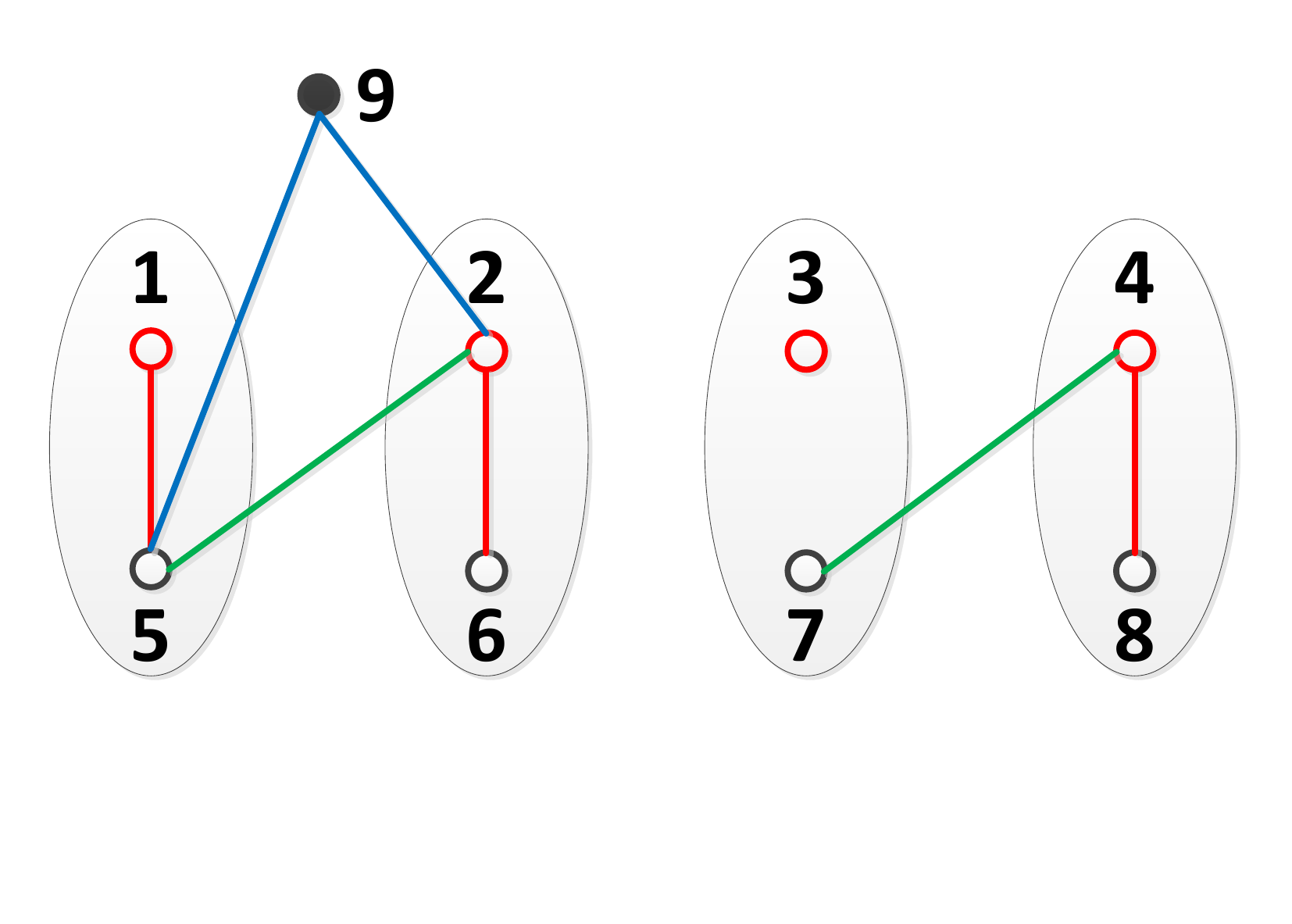}}\qquad
\subfigure[]{\includegraphics[width=2.43cm]{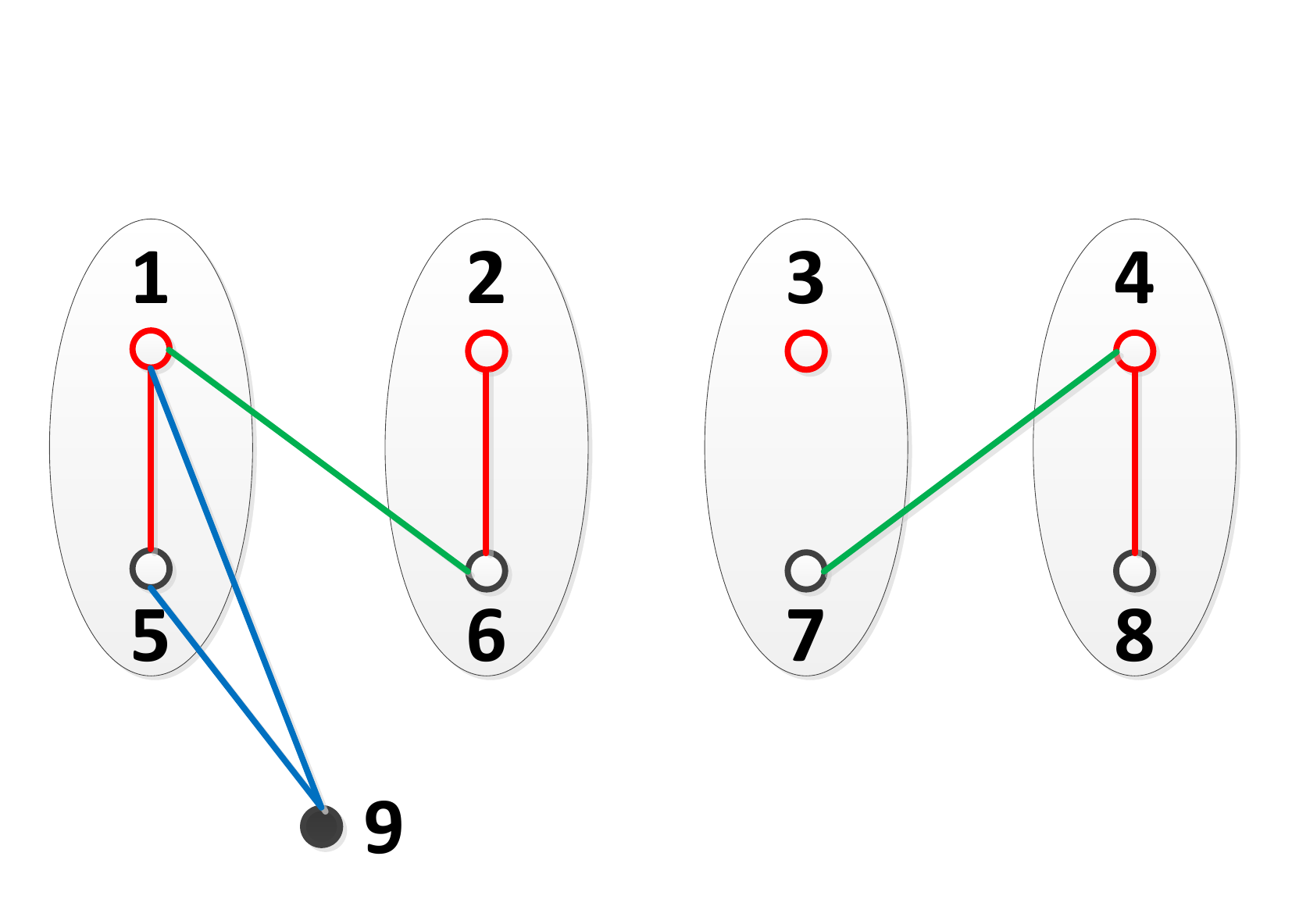}}
\subfigure[]{\includegraphics[width=2.43cm]{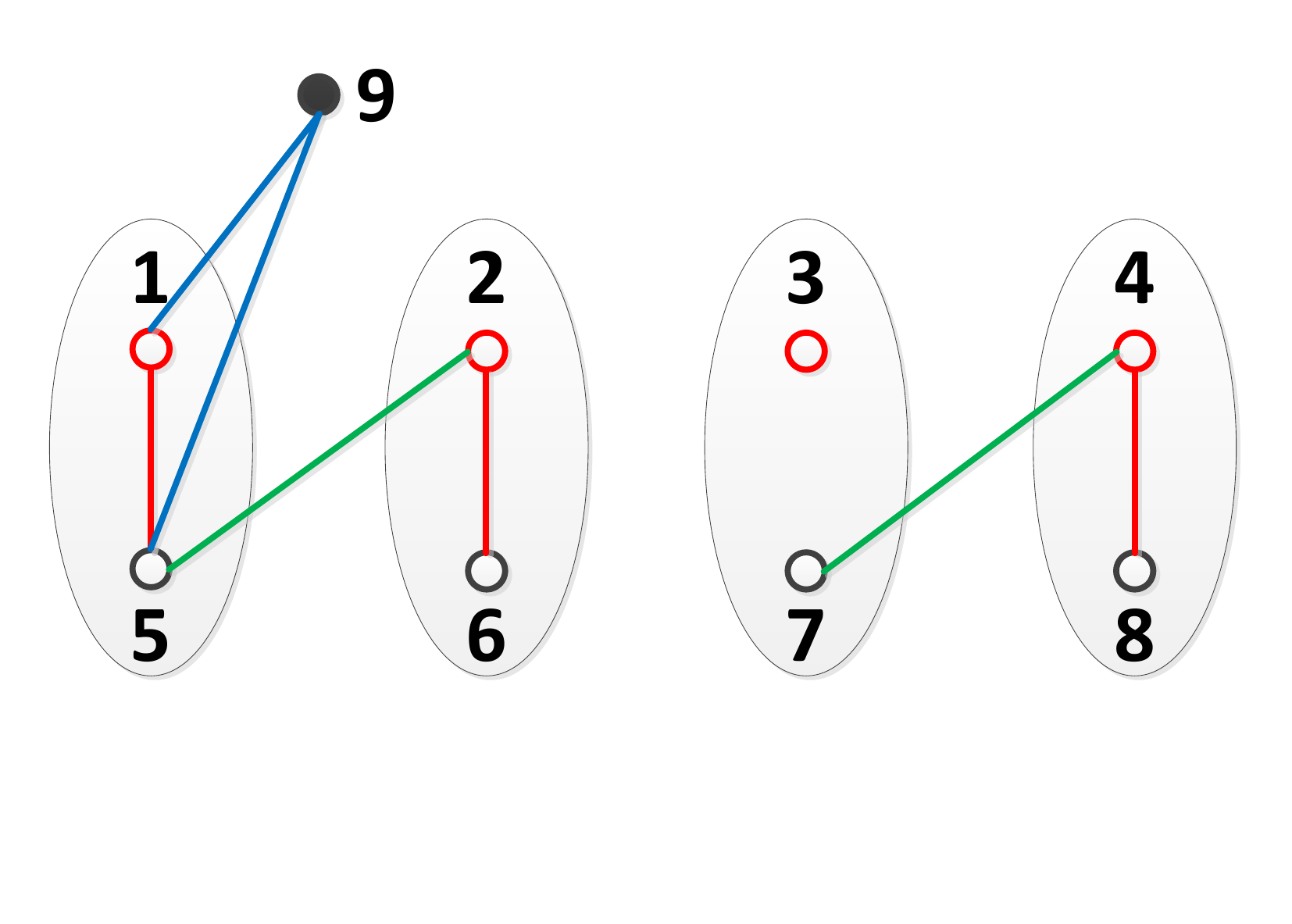}}\qquad
\subfigure[]{\includegraphics[width=2.43cm]{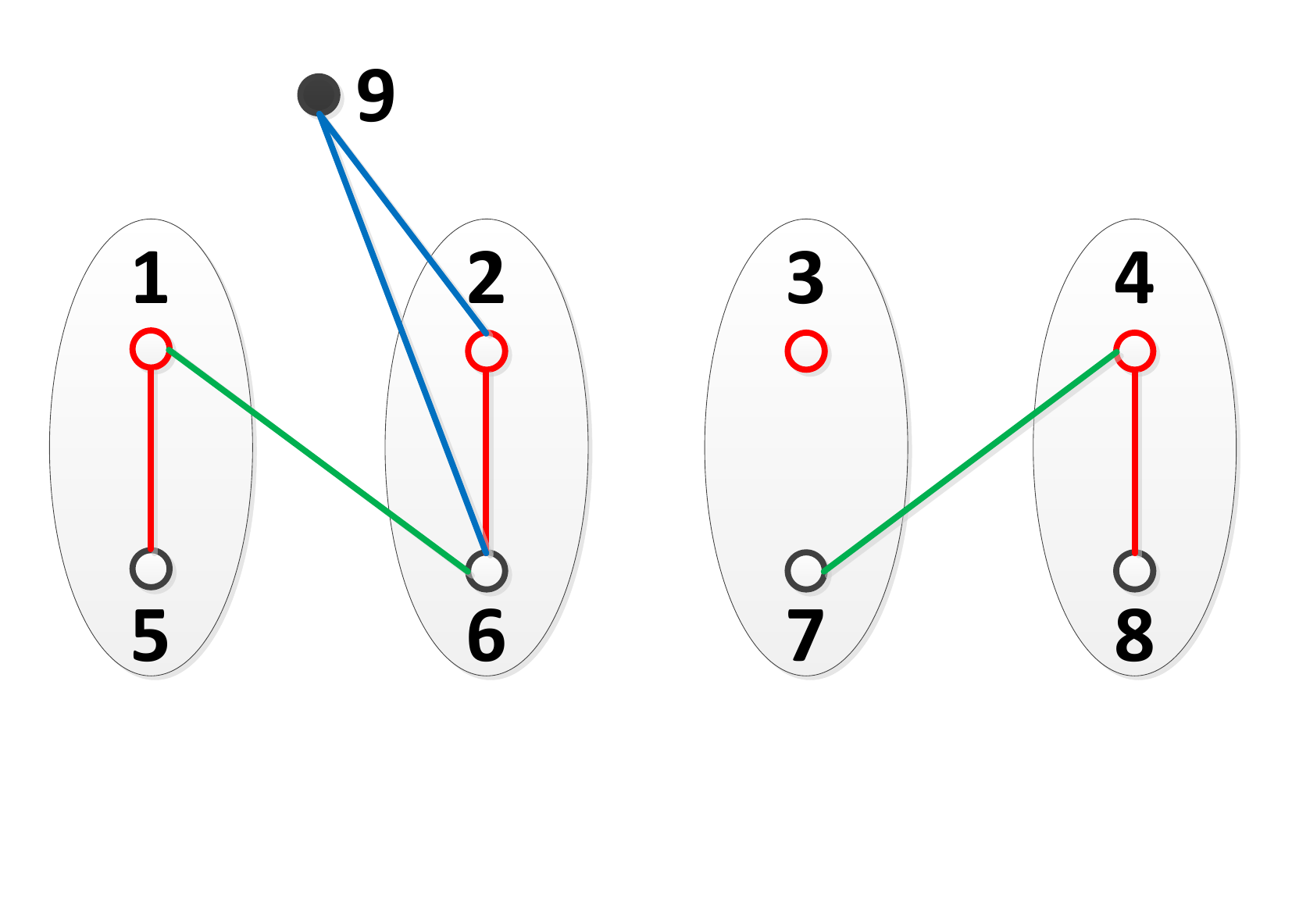}}\qquad
\subfigure[]{\includegraphics[width=2.43cm]{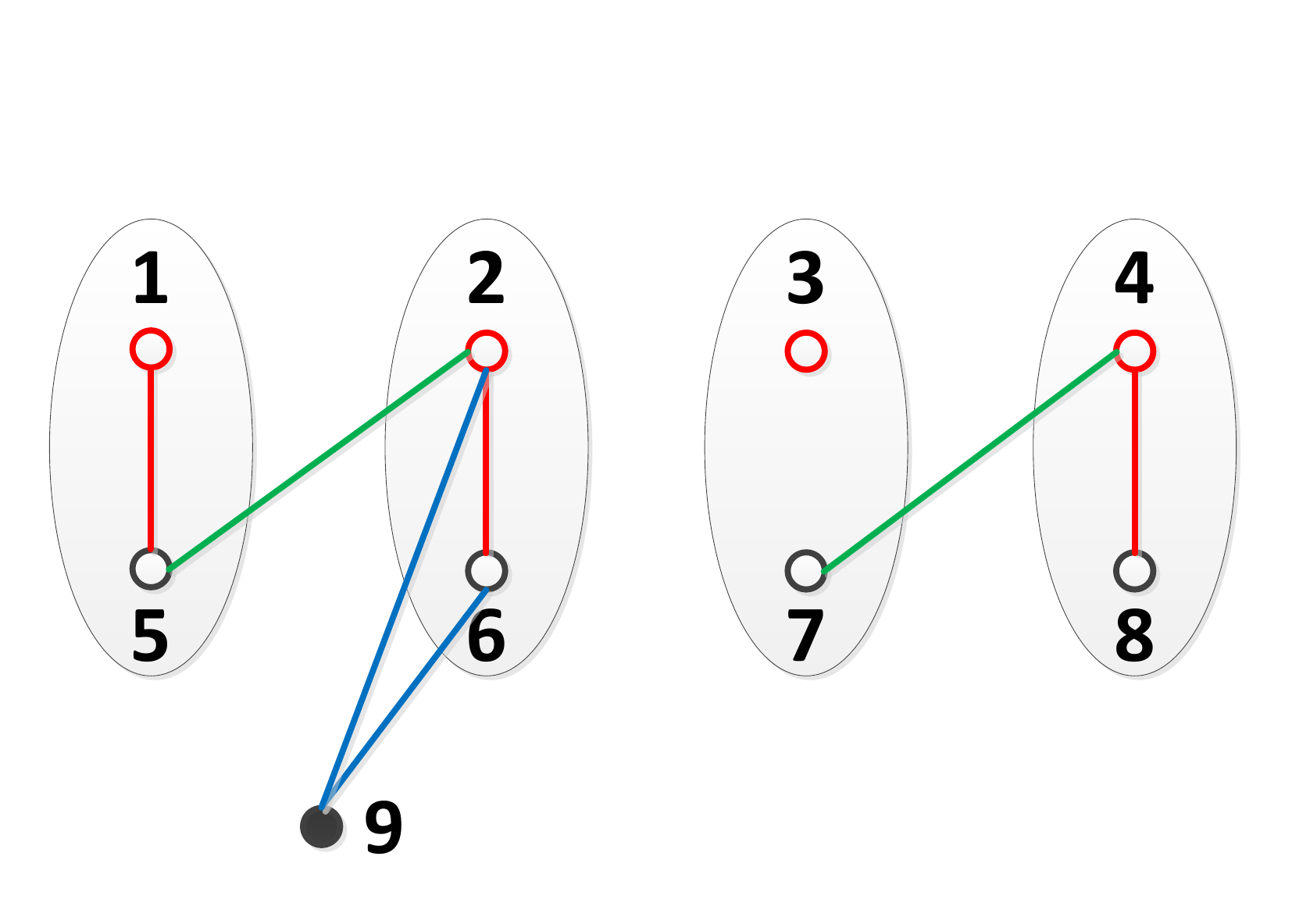}}
\caption{(a) to (l) are graphs designed by combining Figs. \ref{Step4-1con} and \ref{Step4-2con}.}\label{Step5-1con}
\end{center}
\end{figure}
\begin{figure}[H]
\begin{center}
\subfigure[]{\includegraphics[width=2.24cm]{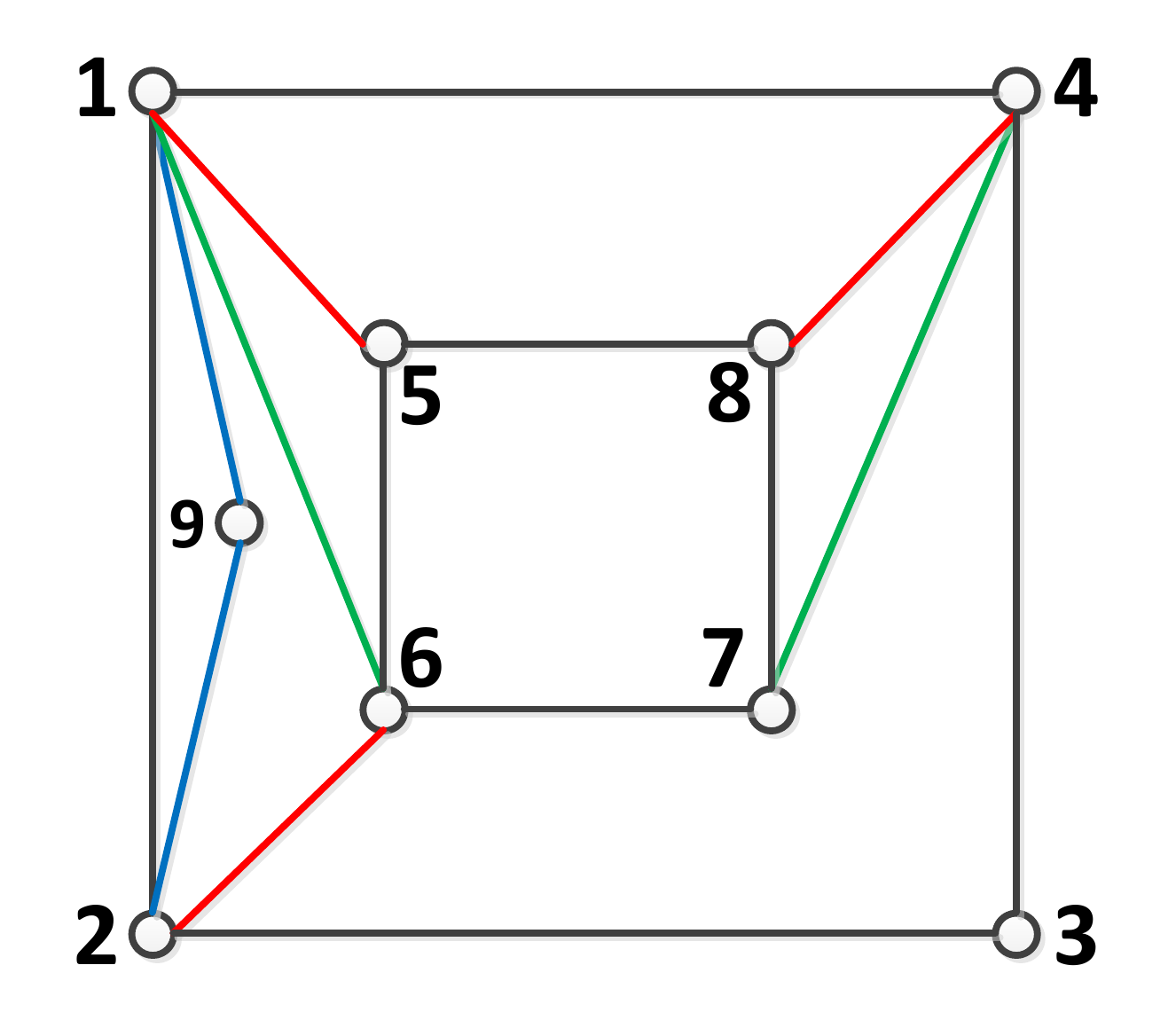}}
\subfigure[]{\includegraphics[width=2.24cm]{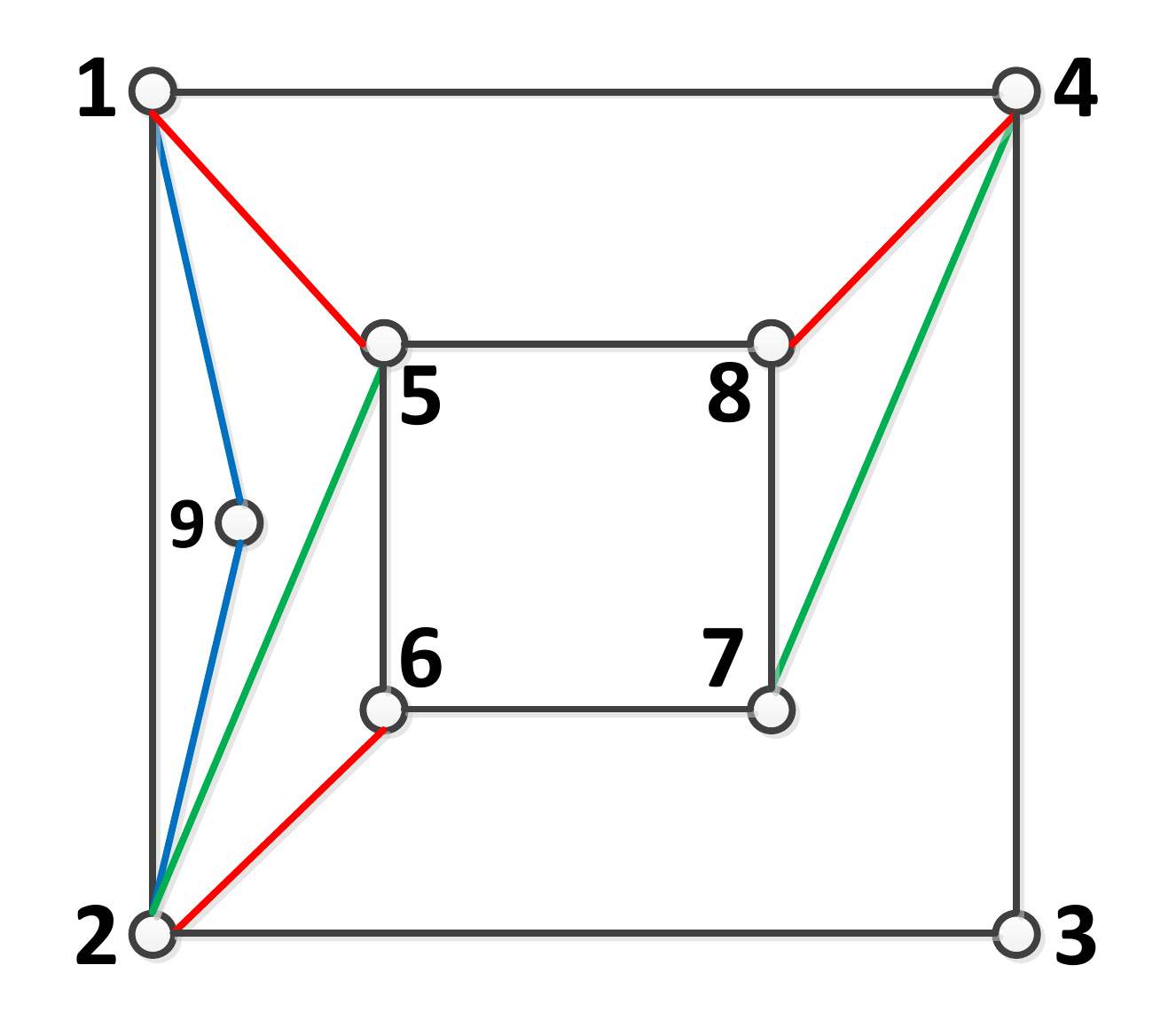}}
\subfigure[]{\includegraphics[width=2.24cm]{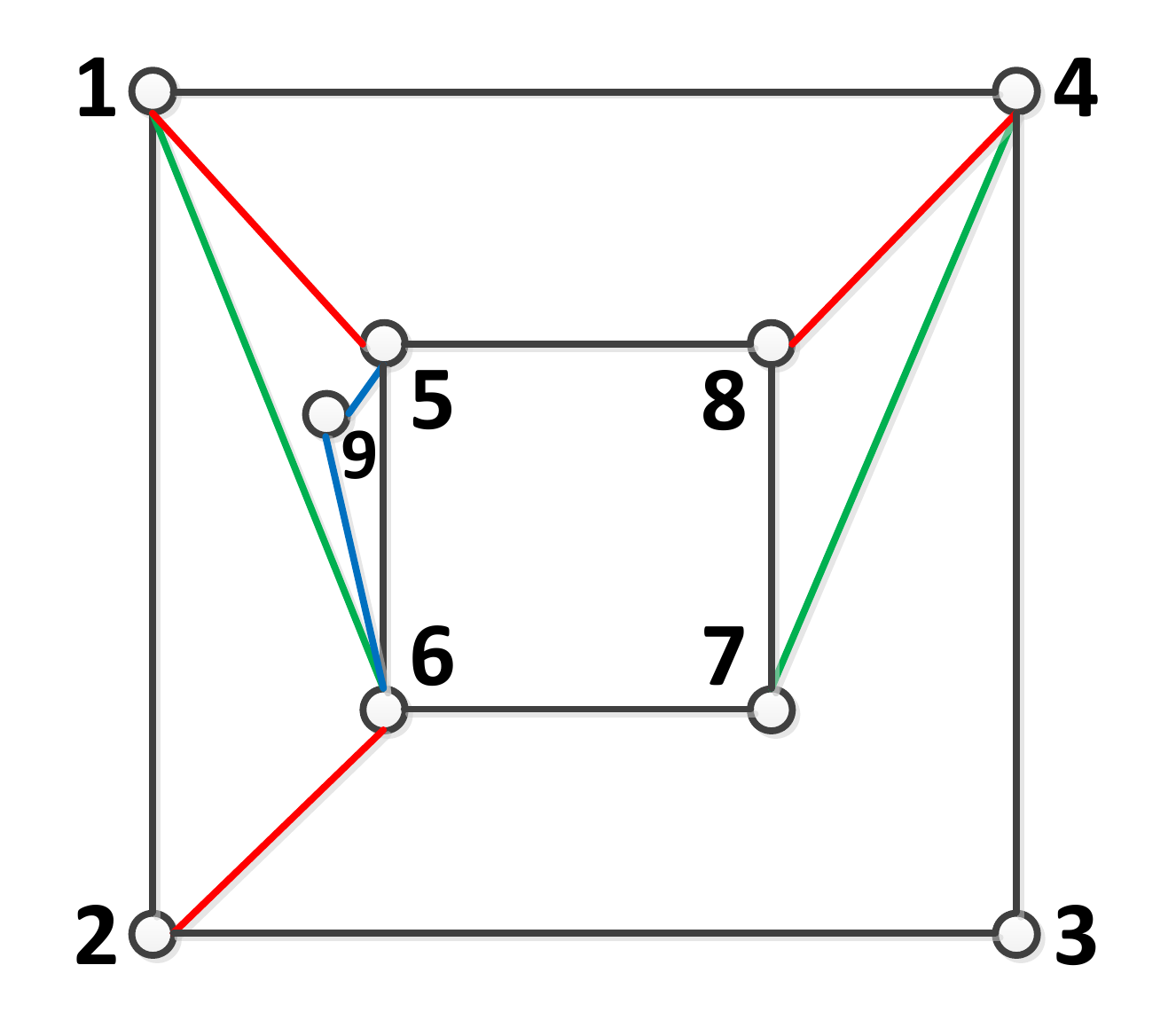}}
\subfigure[]{\includegraphics[width=2.24cm]{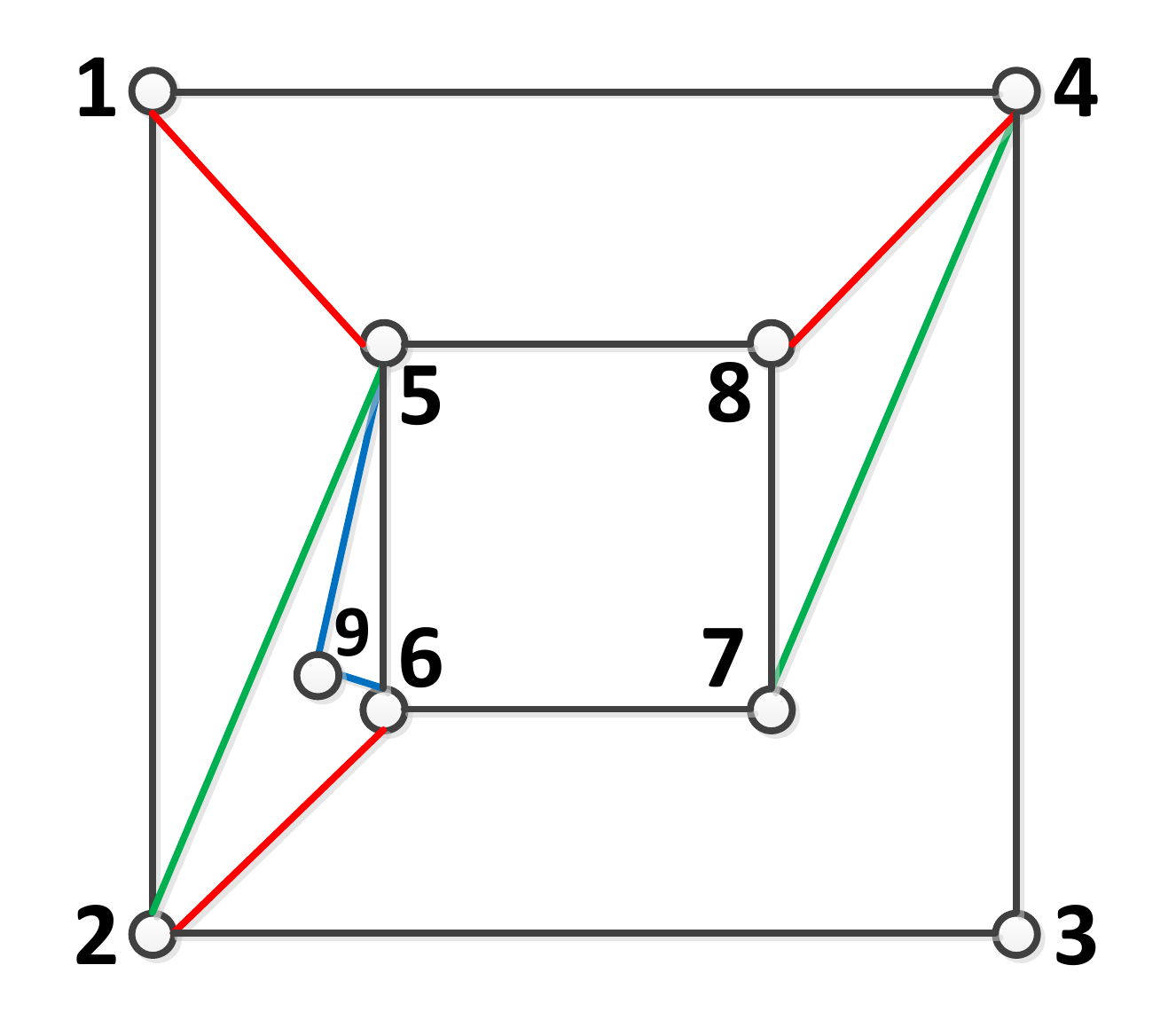}}\\
\subfigure[]{\includegraphics[width=2.24cm]{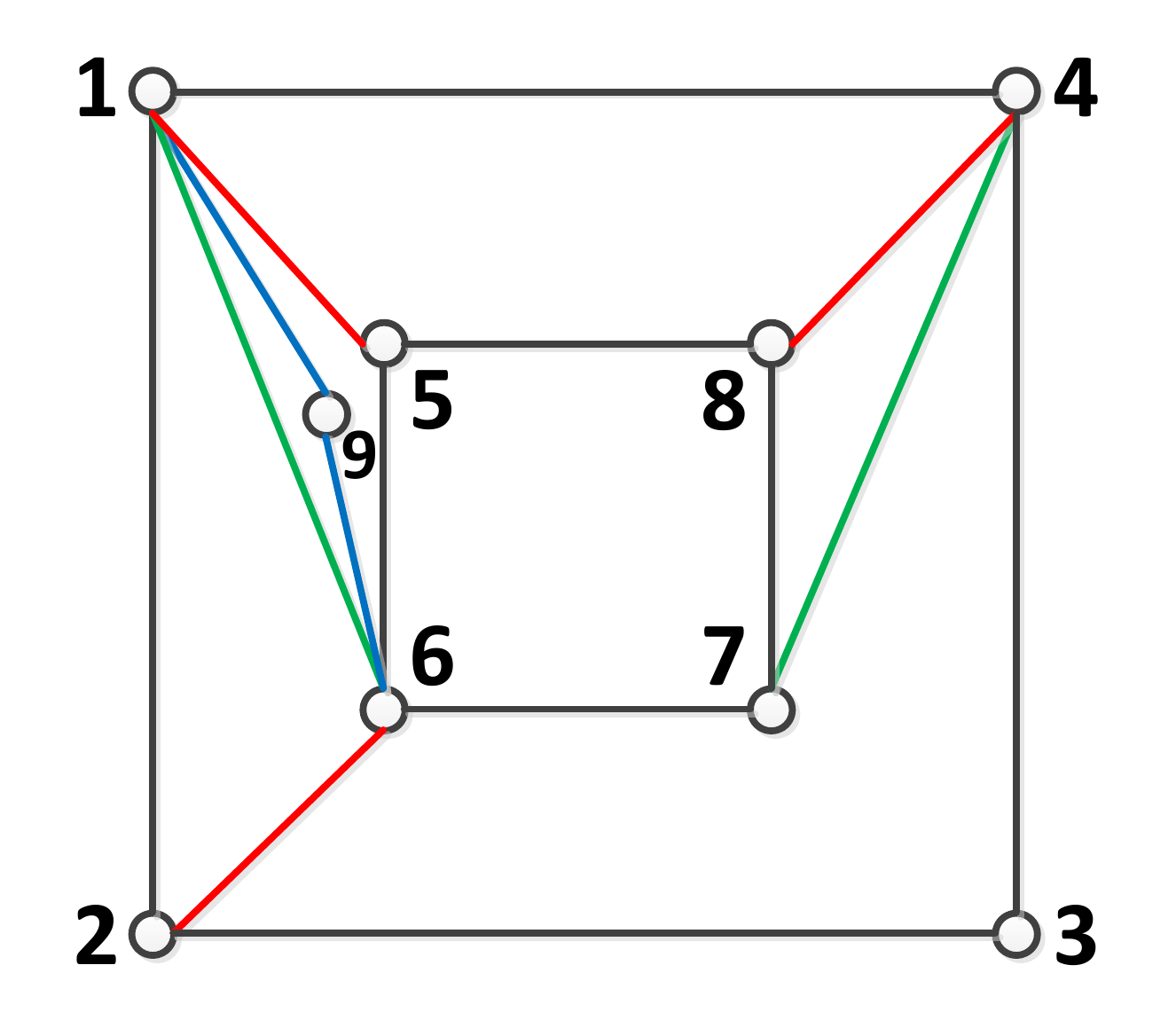}}
\subfigure[]{\includegraphics[width=2.24cm]{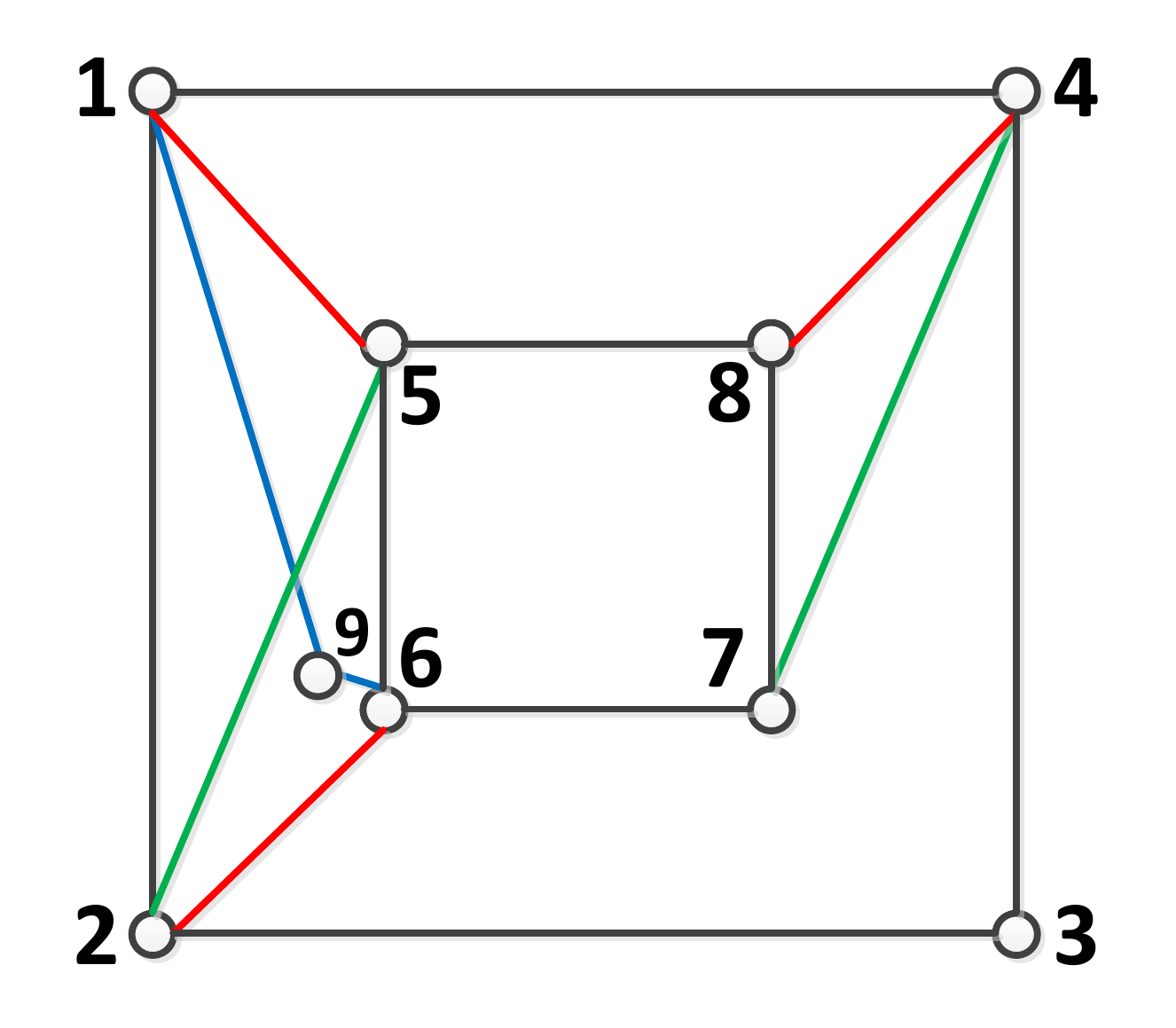}}
\subfigure[]{\includegraphics[width=2.24cm]{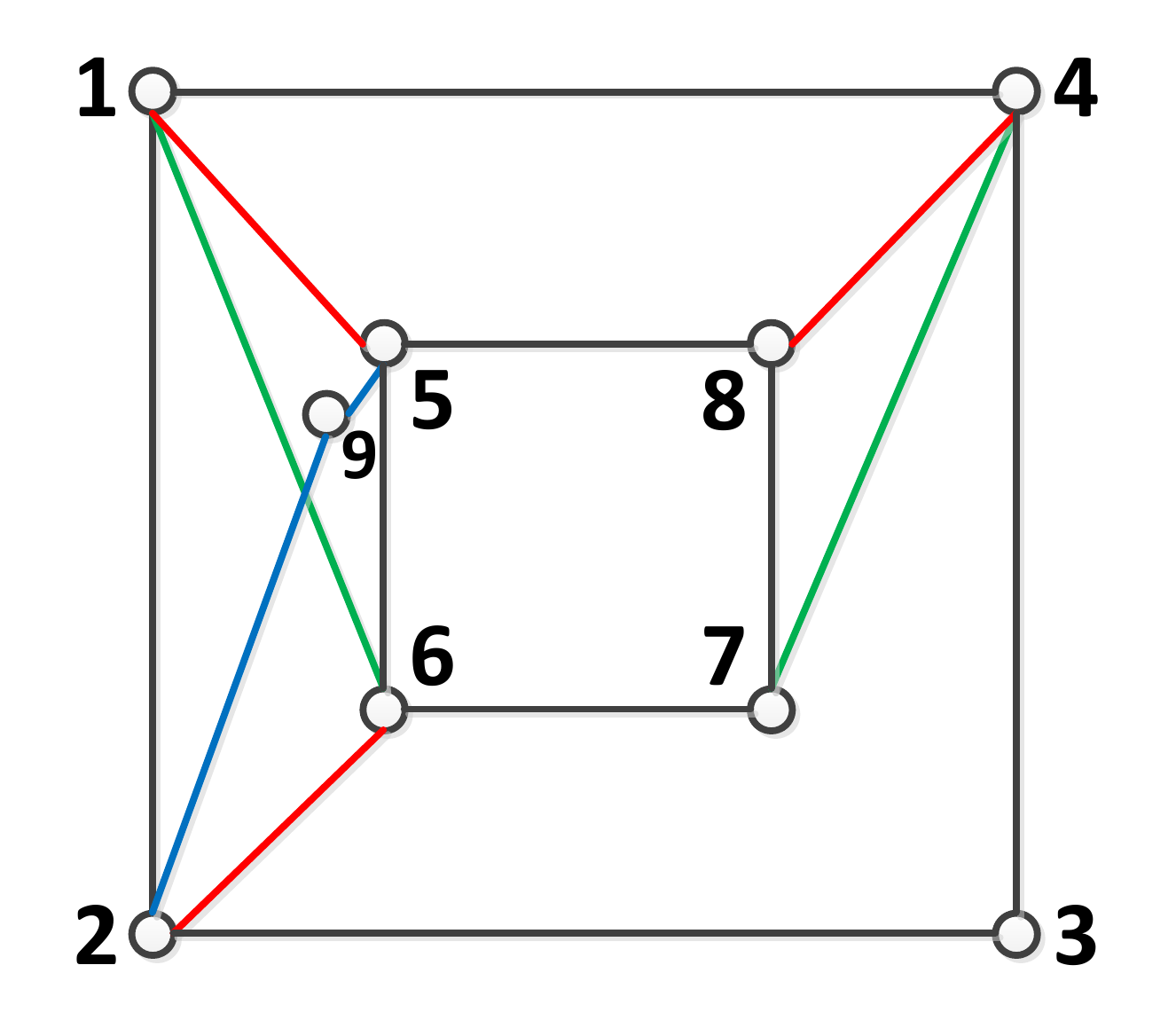}}
\subfigure[]{\includegraphics[width=2.24cm]{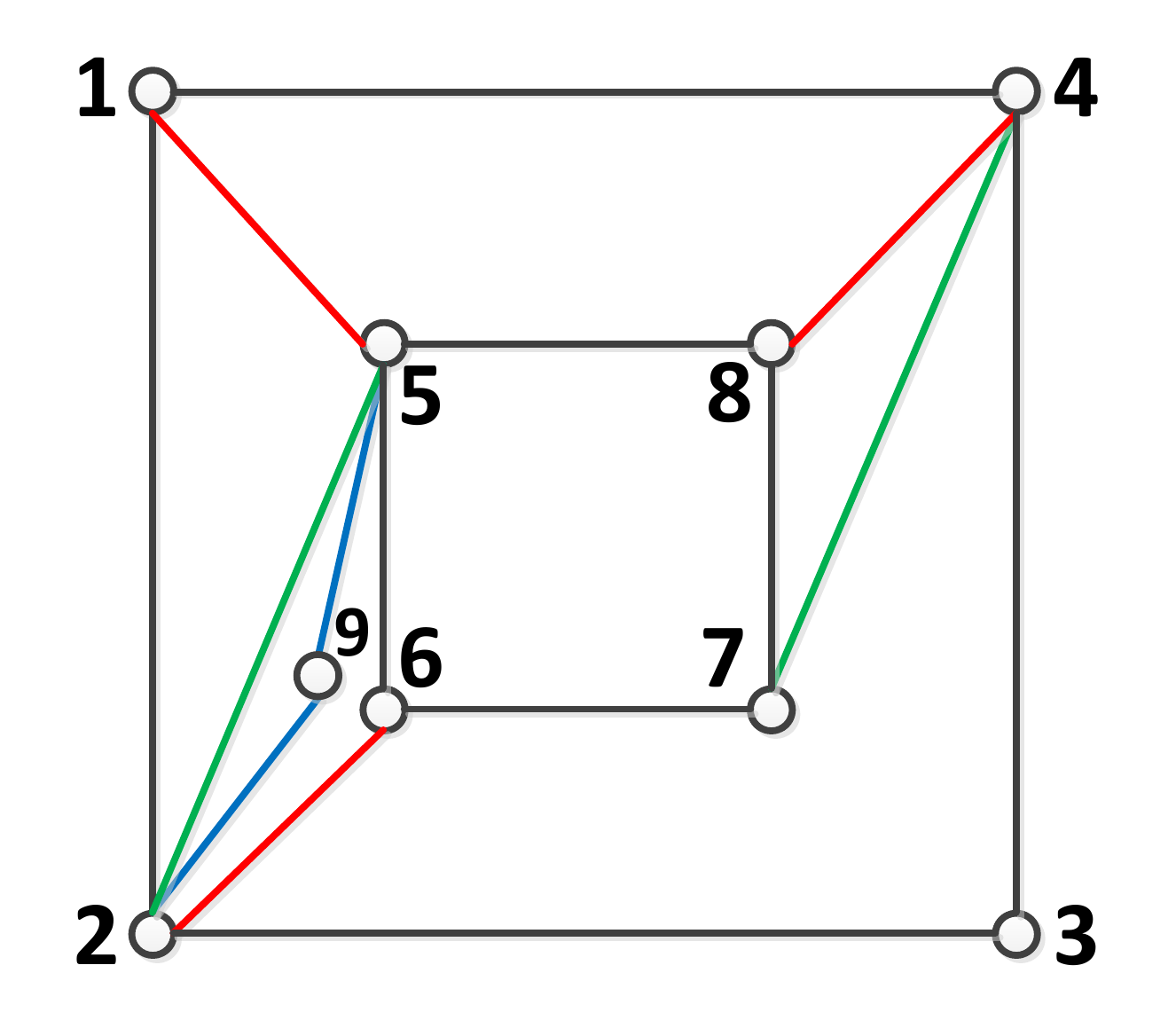}}\\
\subfigure[]{\includegraphics[width=2.24cm]{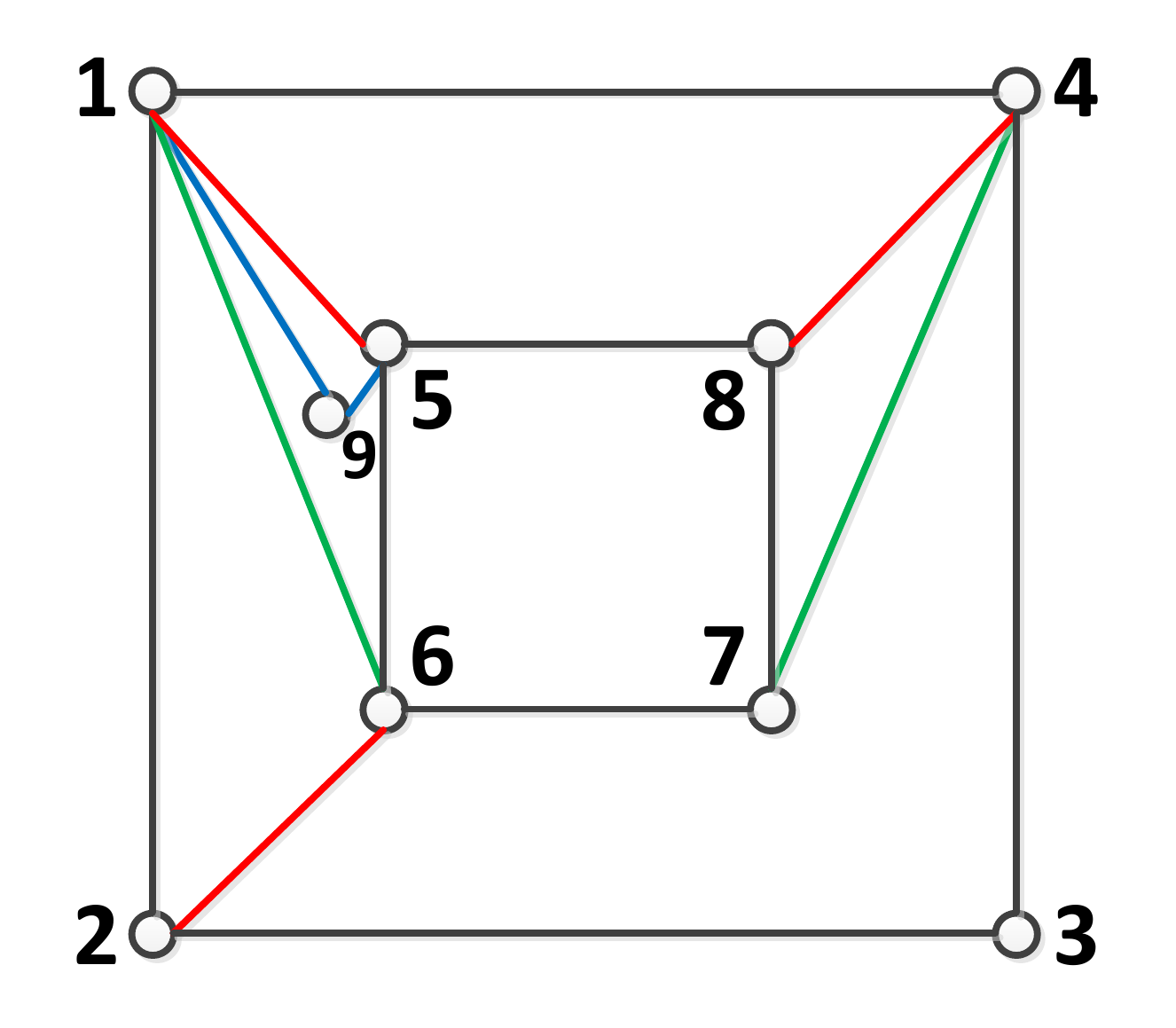}}
\subfigure[]{\includegraphics[width=2.24cm]{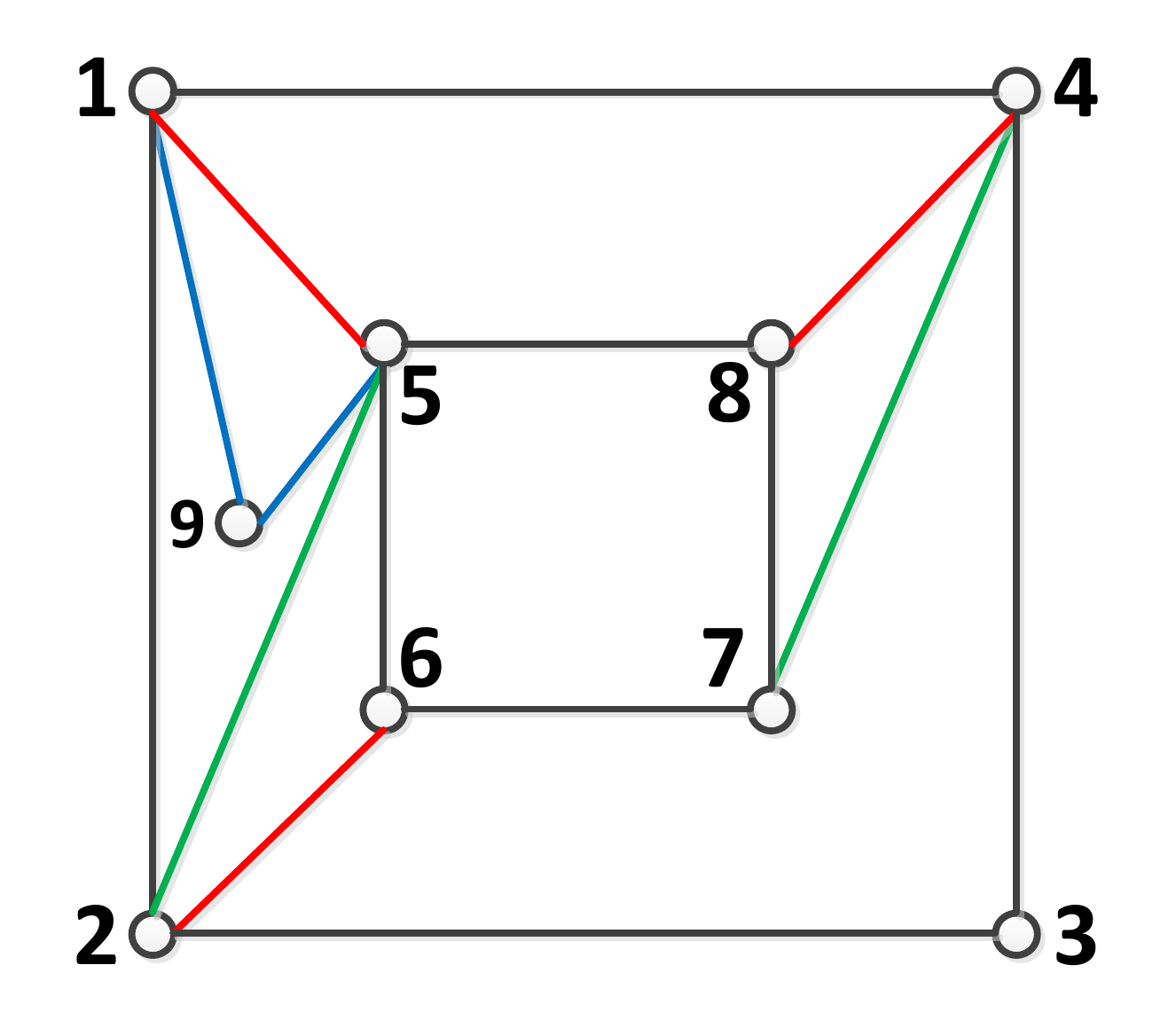}}
\subfigure[]{\includegraphics[width=2.24cm]{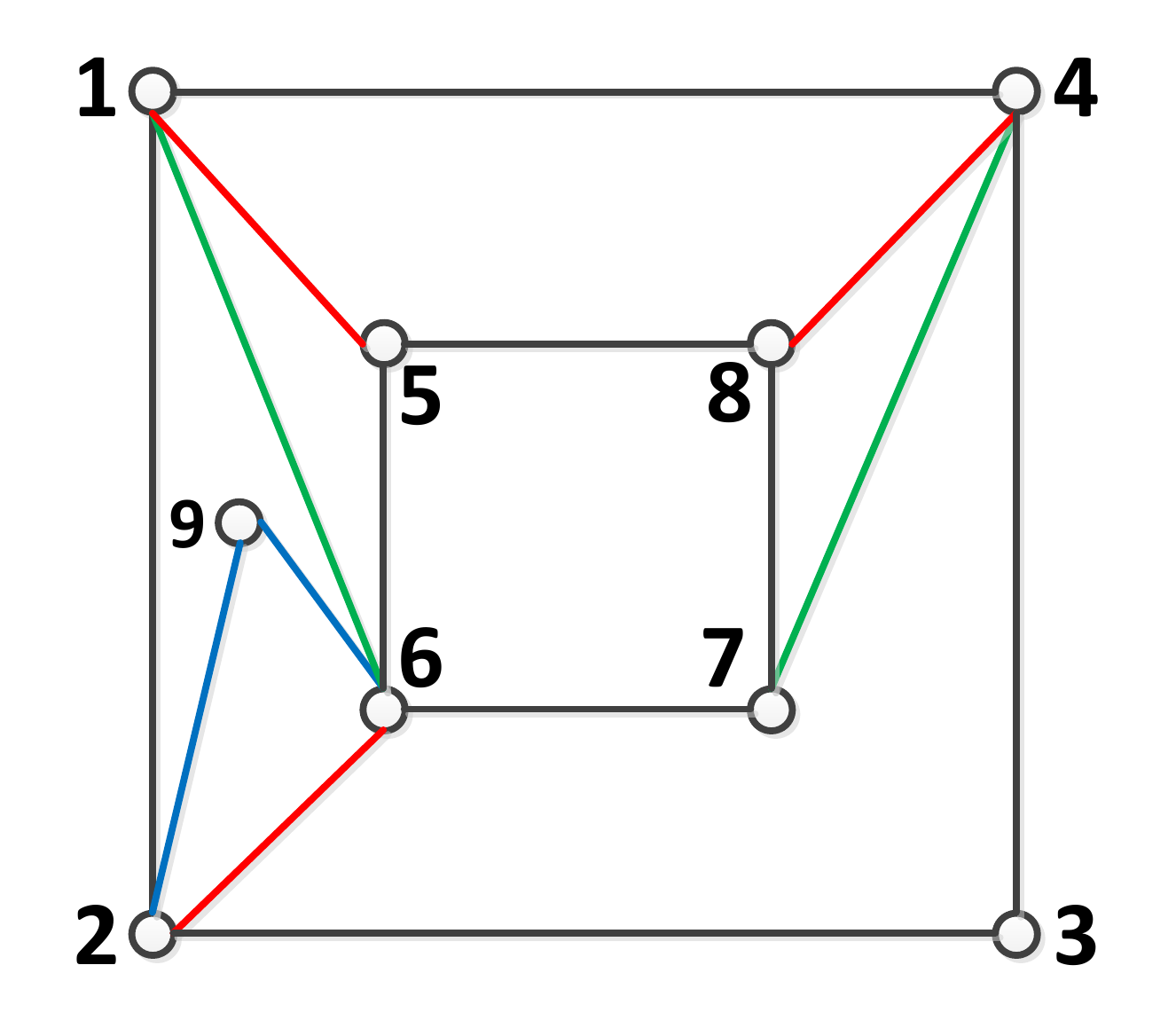}}
\subfigure[]{\includegraphics[width=2.24cm]{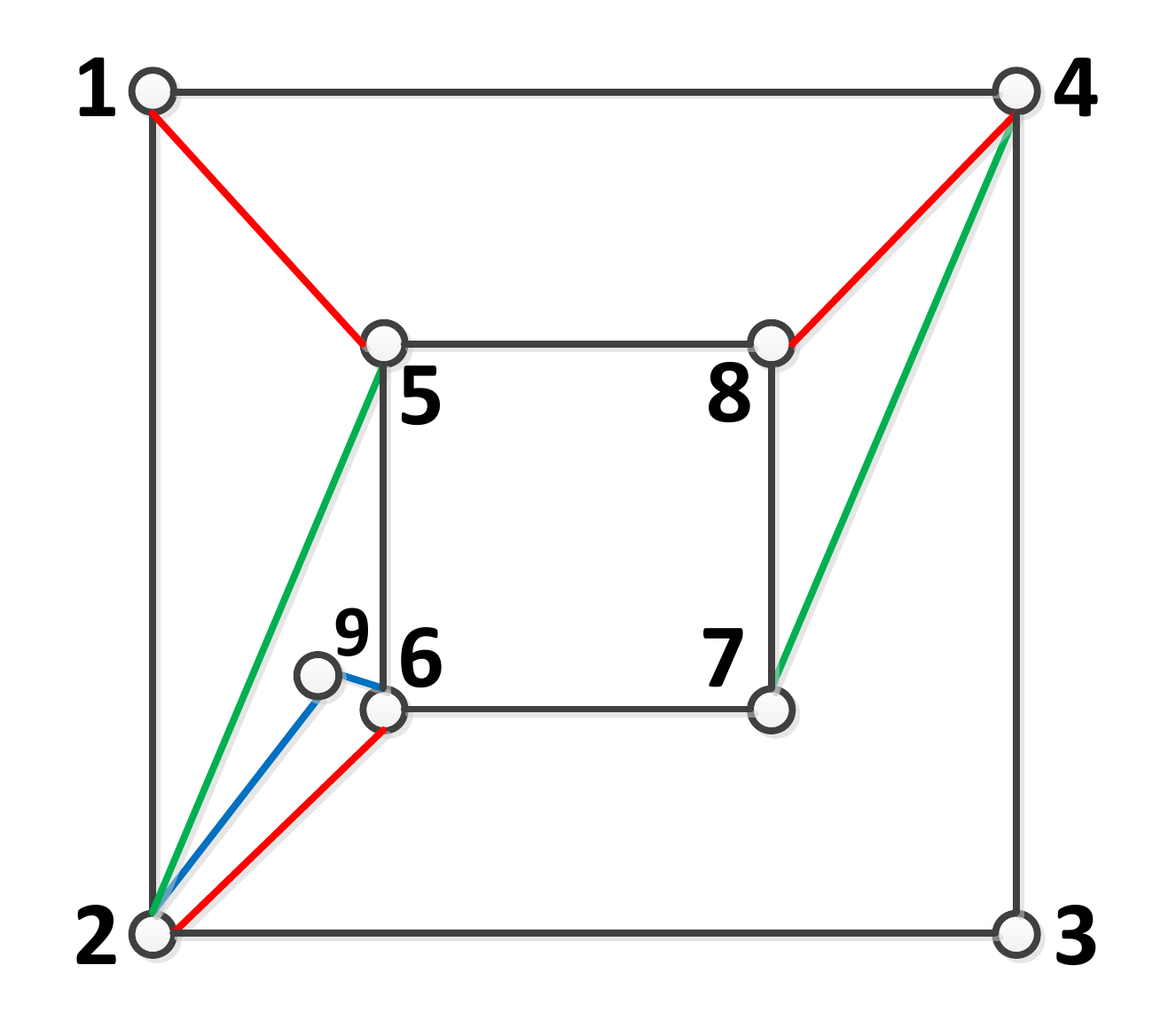}}
\caption{(a) to (l) are the topology structures corresponding, respectively, to each of the figures in Fig. \ref{Step5-1con}.}\label{Step5-1conN}
\end{center}
\end{figure}

\textbf{Step 6} Now let us focus on the Case c) of Step 4 since the other two Cases a) and b) of Step 4 have been utilized in the construction in Step 5. As to Case c) of Step 4, there are four corresponding graphs in Fig. \ref{Step4-3con}. For each of them, we only consider the two unfixed double nodes sets, between which a new edge is to be designed according to the following two rules:
\begin{itemize}
\item[i)] The two nodes to be connected to form the newly designed edge belong to different sets $\Omega_i, i=1,2;$
\item[ii)]The two nodes to be connected are not in the same double nodes set;
\item[iii)] The newly designed edge does not intersect with any other edge.
\end{itemize}
For the two unfixed double nodes sets, the newly designed edge is colored green between them in each graph of Fig. \ref{Step6-1con} and the corresponding topology structures are depicted in Fig. \ref{Step6-1conN}.

\begin{figure}[H]
\begin{center}
\subfigure[]{\includegraphics[width=2.43cm]{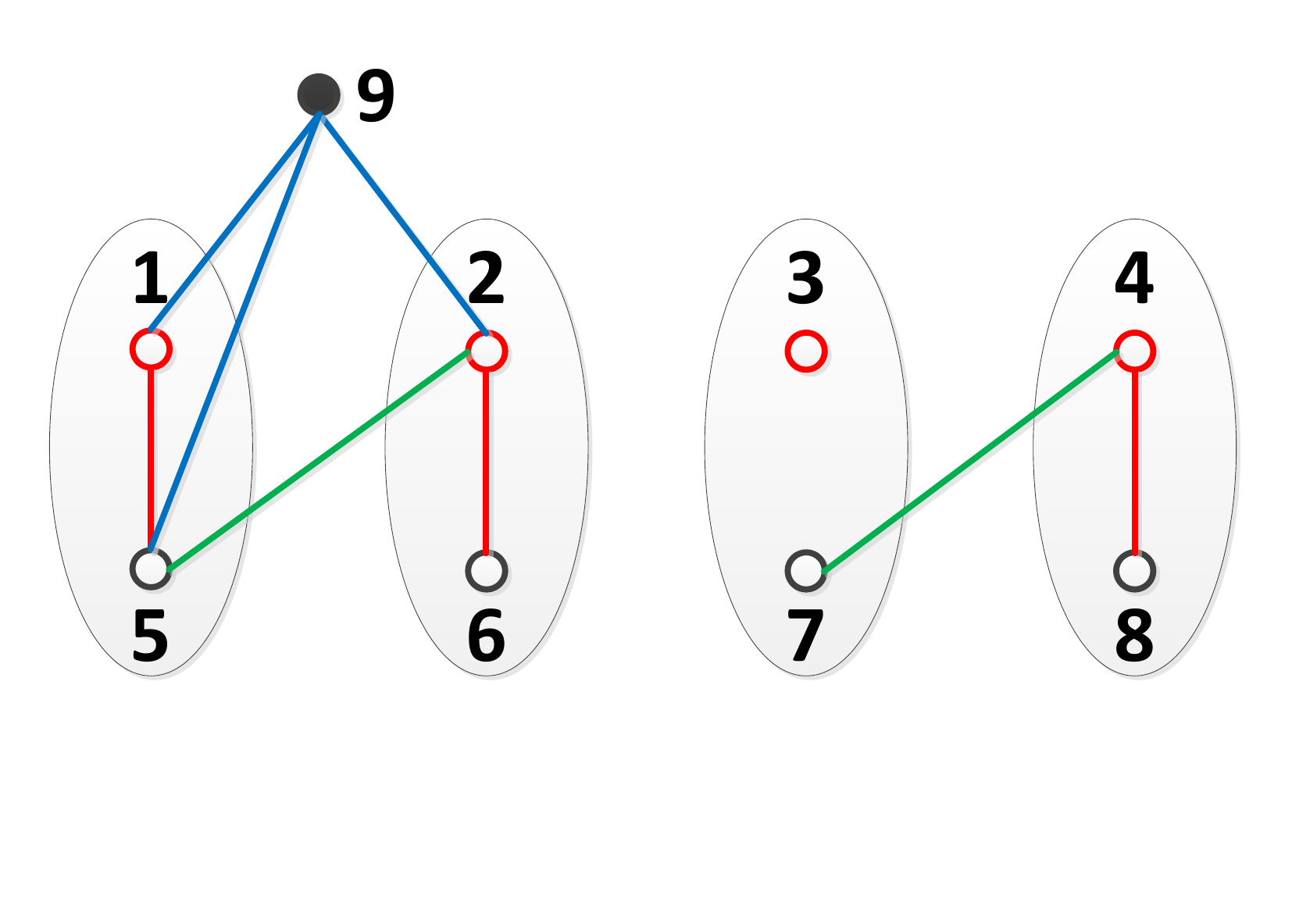}}\qquad
\subfigure[]{\includegraphics[width=2.43cm]{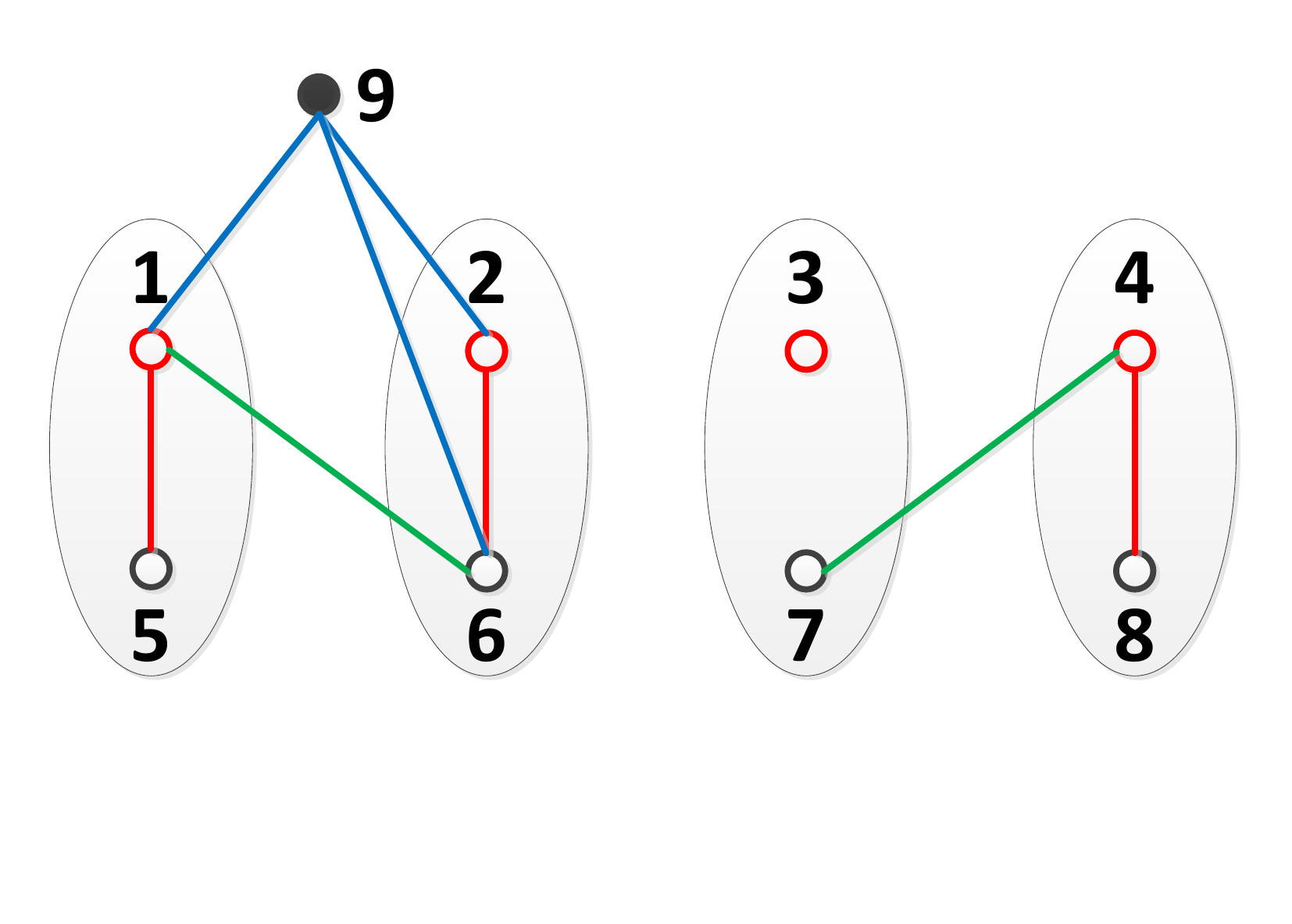}}\qquad\\
\subfigure[]{\includegraphics[width=2.43cm]{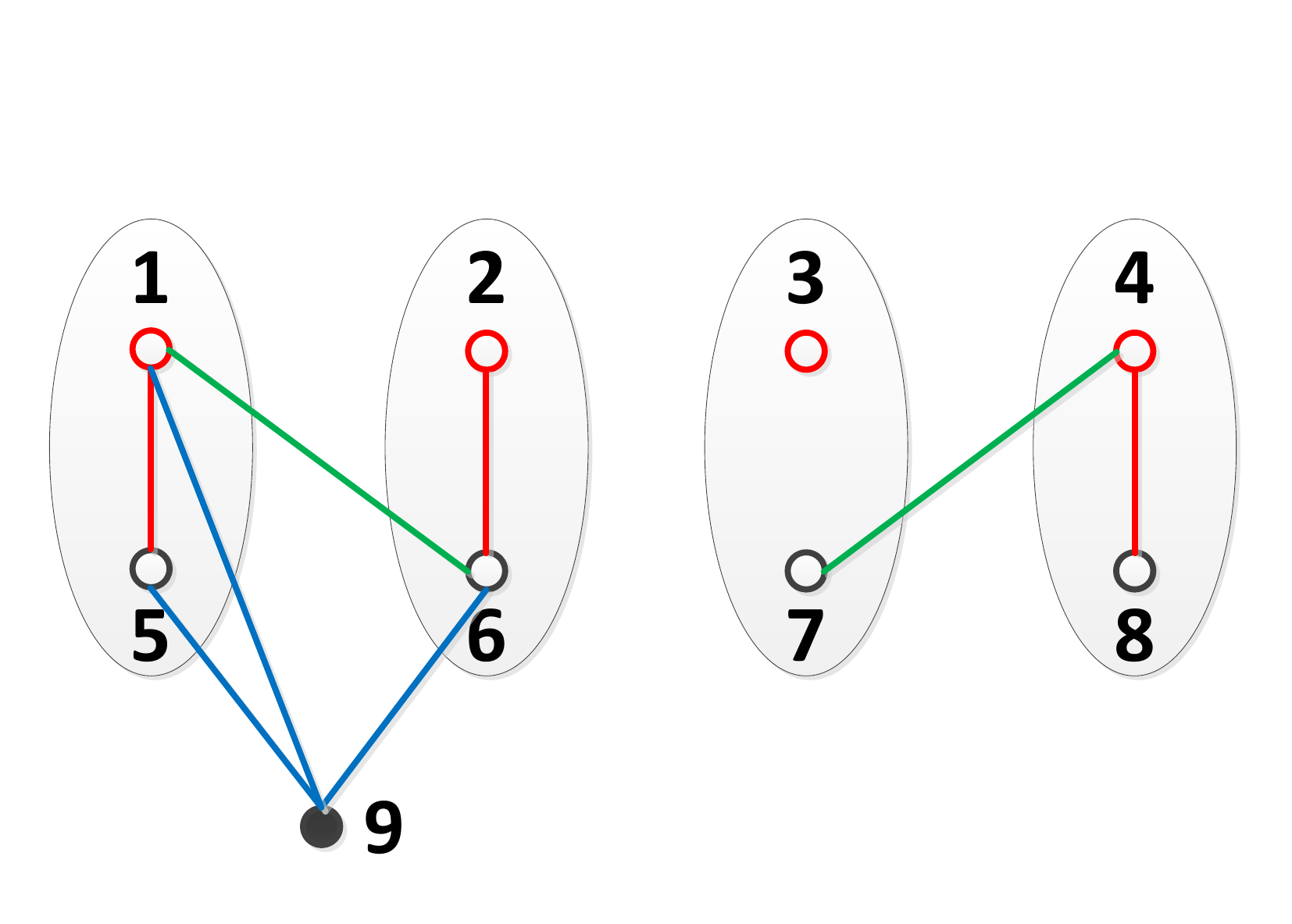}}\qquad
\subfigure[]{\includegraphics[width=2.43cm]{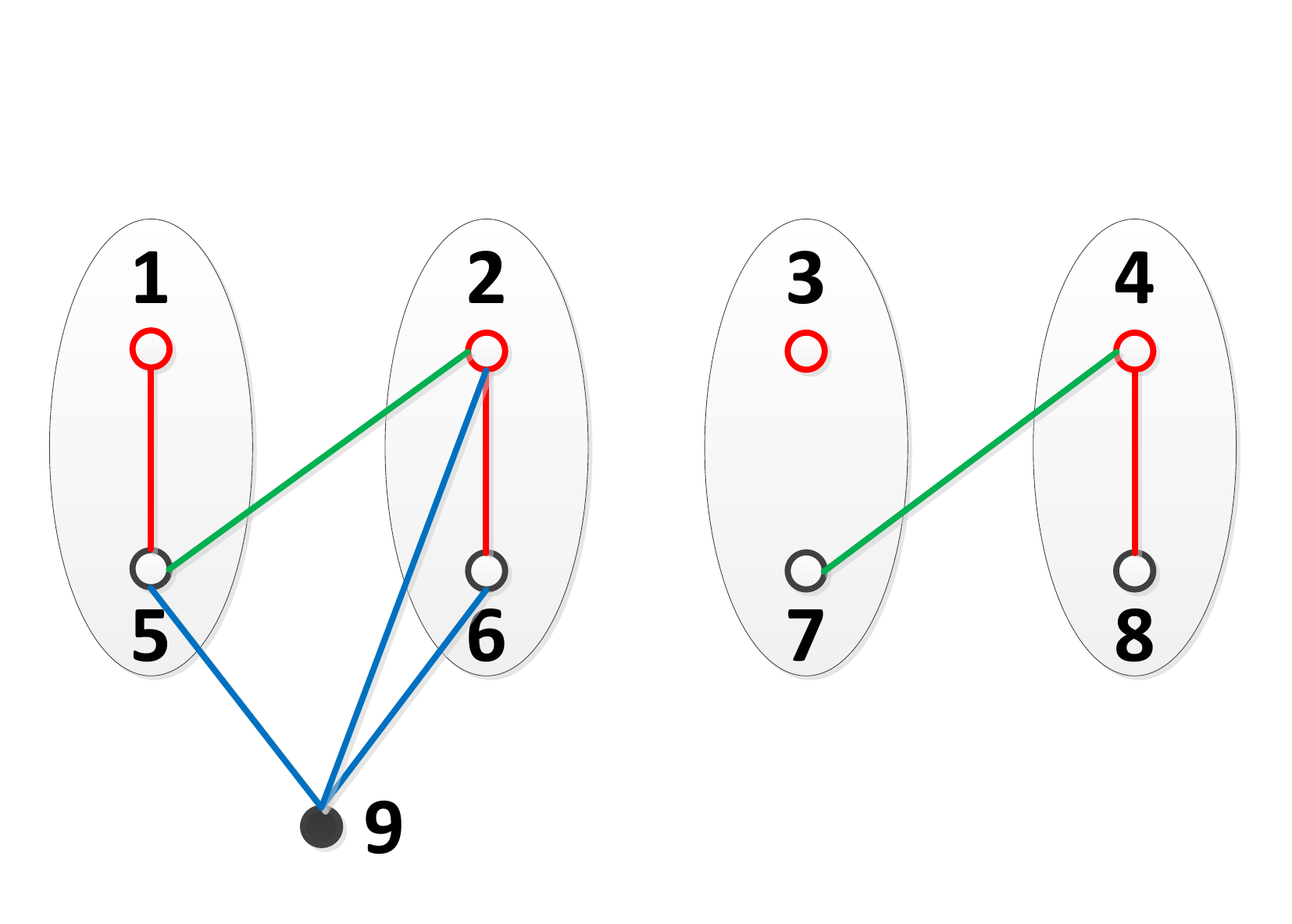}}\qquad
\caption{(a) to (d) are the figures indicating the newly designed edge.}\label{Step6-1con}
\end{center}
\end{figure}
\begin{figure}[H]
\begin{center}
\subfigure[]{\includegraphics[width=2.3cm]{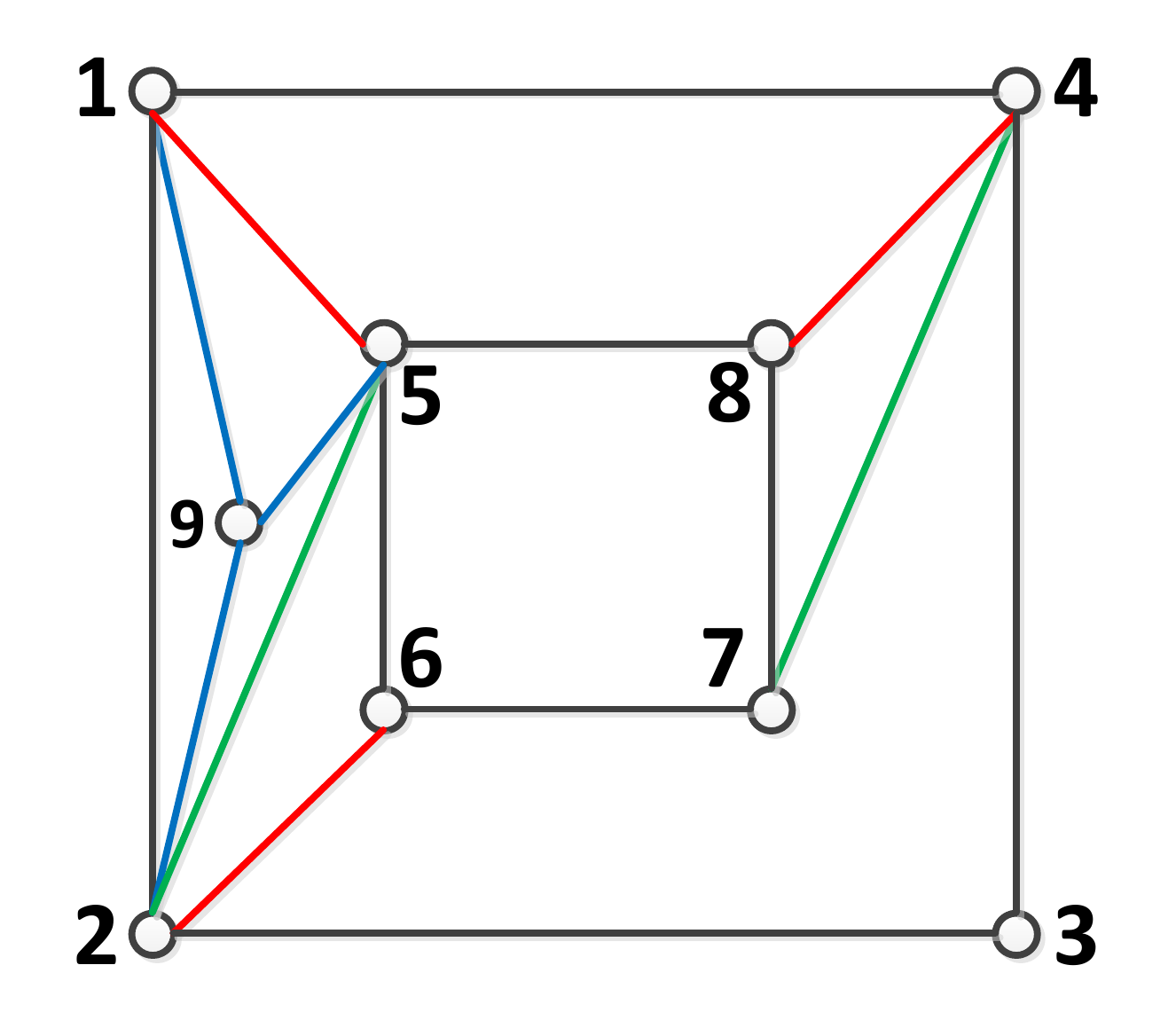}}
\subfigure[]{\includegraphics[width=2.3cm]{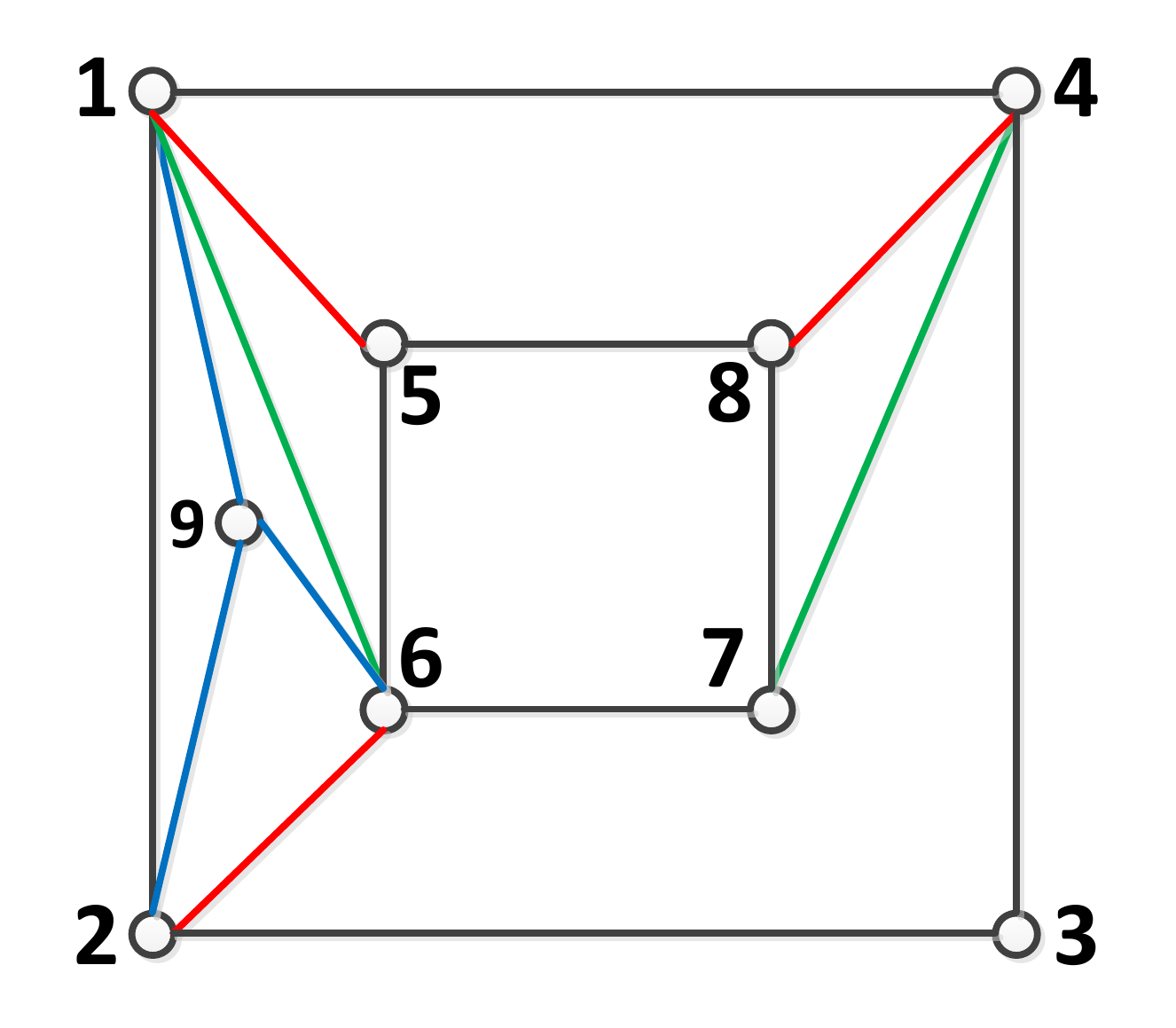}}
\subfigure[]{\includegraphics[width=2.3cm]{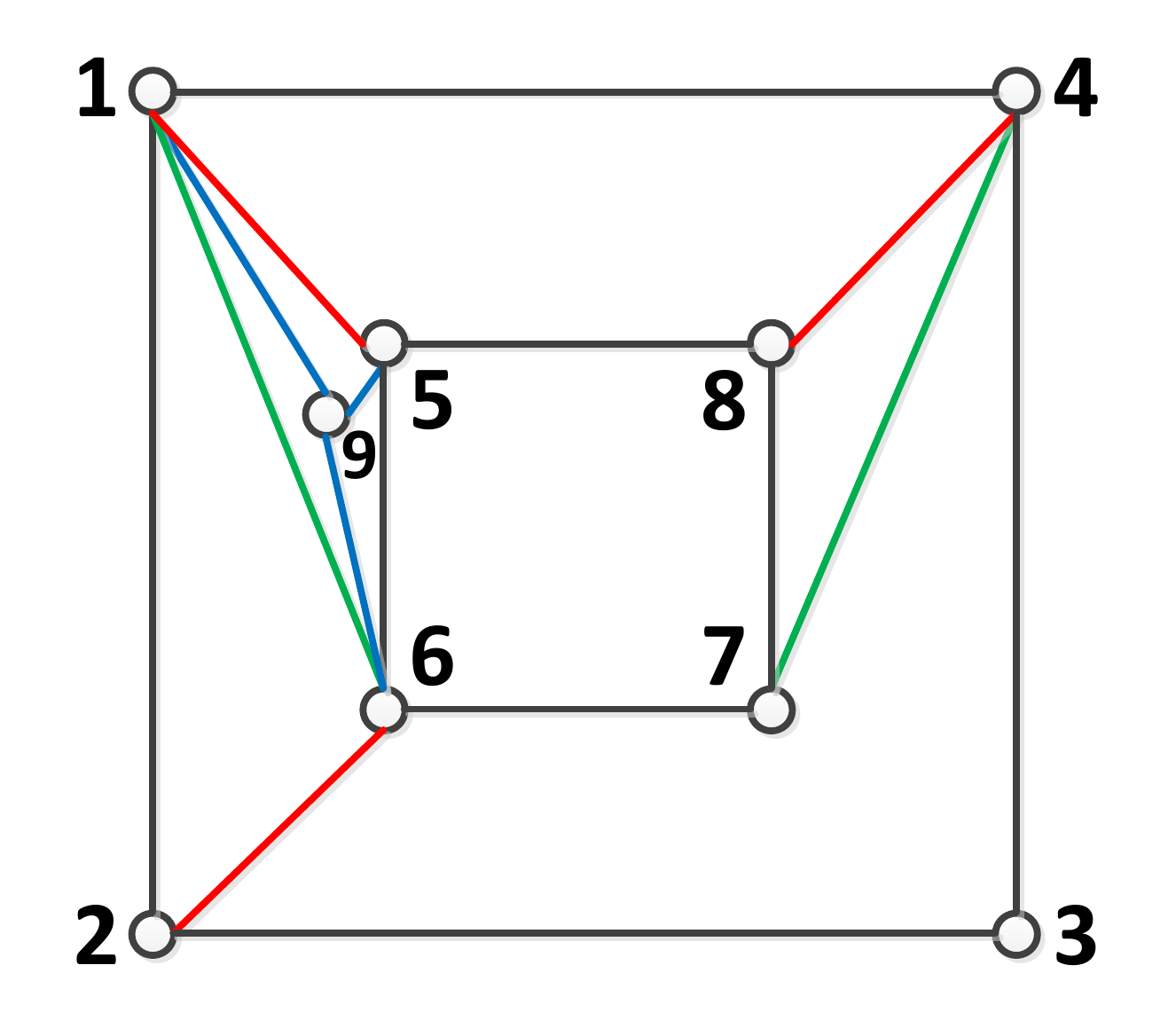}}
\subfigure[]{\includegraphics[width=2.3cm]{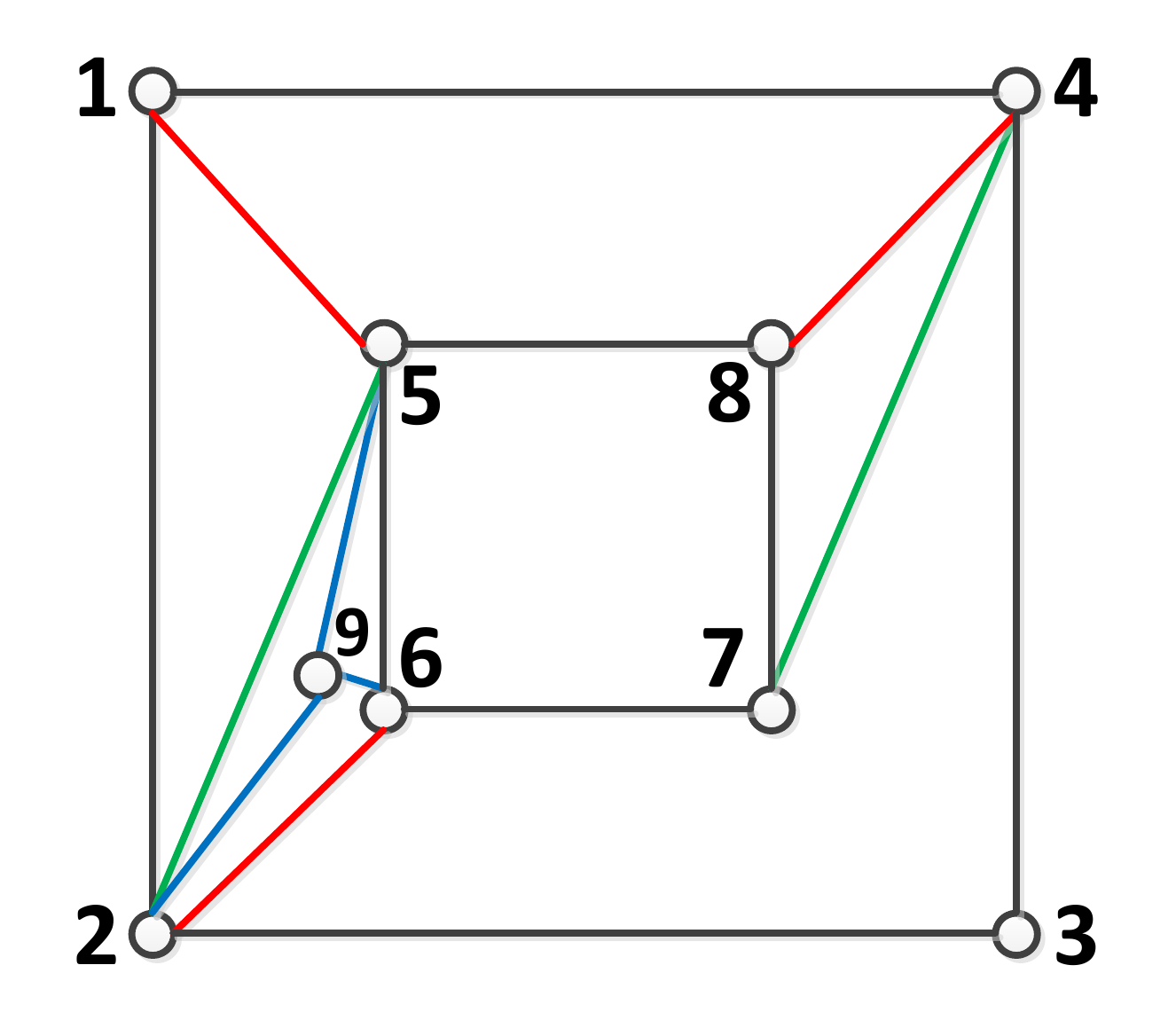}}
\caption{(a) to (d) are the topology graphs corresponding to each of the figures of Fig. \ref{Step6-1con} , respectively.}\label{Step6-1conN}
\end{center}
\end{figure}

In subsequent Theorem \ref{TheoTo3}, Steps 5 and 6 produce more perfectly controllable graphs, which can be  verified in the same way as Theorem \ref{Theo2topo}.
\begin{theorem}\label{TheoTo3}
The communication graphs (a), (b), (d), (e), (f), (g), (h), (k), (l) of Fig. \ref{Step5-1conN} and (a), (b), (c), (d) of Fig. \ref{Step6-1conN} are perfectly controllable.
\end{theorem}

\begin{remark}
Note that although (f), (g) of Fig. \ref{Step5-1conN} have intersectional edges, they are still perfectly controllable.
\end{remark}

\begin{remark}
In the preceding Steps 4 to 6, the design of perfectly controllable graphs is based on the two unfixed double nodes sets. In the next step, however, the design of perfectly controllable graphs are from the two fixed double nodes sets. Let us recall that the two fixed double nodes sets are the two oval shapes, respectively, consisting of $3, 7$ and $4, 8.$
\end{remark}

The subsequent construction of graphs in the next step begins from each of the following 26 graphs: (a) to (f) of Fig. \ref{Step4-2con}; (a) to (d) of Fig. \ref{Step4-3con}; (a) to (l) of Fig. \ref{Step5-1con}; and (a) to (d) of Fig. \ref{Step6-1con}. This construction produces a total of 208 graphs.

\begin{remark}
The above mentioned 26 graphs are designed from the two unfixed double nodes sets, which are the two oval shapes, respectively, consisting of 1, 5 and 2, 6. The next design Step 7 is to show that after manipulations, respectively, on the fixed and unfixed double node sets, more perfectly controllable graphs can be constructed by combining the manipulations from these two kinds of double nodes sets. On the other hand, the design result achieved in Step 7 implies that even if in some circumstances perfectly controllable graphs cannot be obtained by manipulations only on fixed double nodes sets, perfectly controllable graphs can still be constructed if the unfixed double nodes sets are manipulated at the same time.
\end{remark}

\textbf{Step 7}
We use graph (a) of Fig. \ref{Step5-1con} to illustrate this step. From this graph, eight new graphs are created by following subsequent items i) to iv).
\begin{itemize}
\item[i)] One new node labeled as 10 is added to the graph (a) of Fig. \ref{Step5-1con};

\item[ii)] In the two fixed double nodes sets, consider the two nodes belonging to the same $\Omega_i, i=1,2,$   say $3$ and $4$ which both belong to $\Omega_1.$ Then node $10$ is connected, respectively, to the nodes $3$ and $4$ to create two new edges $e_{10,3}, e_{10,4}.$ The corresponding graph is (a) of Fig. \ref{Step7-1con}.

\vspace{1.5mm}

For nodes $7$ and $8,$ the two new edges $e_{10,7}, e_{10,8}$ are created in the same way. The corresponding graph is (d) of Fig. \ref{Step7-1con}.

\item[iii)] In the two unfixed double nodes sets, which consist of, $1, 5$ and $2, 6,$ respectively; let us consider the two nodes belonging to the same $\Omega_i, i=1,2,$   say $1$ and $2$ which both belong to $\Omega_1.$ Then, from the graph (a) of Fig. \ref{Step7-1con}, the third new edge $e_{10,1}$ or $ e_{10,2}$ is designed by connecting node $10$ with node  $1$ or $2.$ This produces two graphs, which are (b) and (c) of Fig. \ref{Step7-1con} accordingly.

\vspace{1.5mm}

For nodes $5$ and $6,$ the two third new edges $e_{10,5}$ and $ e_{10,6}$ are designed in the same way based on the graph (d) of Fig. \ref{Step7-1con}. The corresponding two graphs are (e) and (f) of Fig. \ref{Step7-1con}.

\item[iv)] From the graph (a) of Fig. \ref{Step7-1con}, another new edge $e_{10,7}$ is designed by following the rule that the newly designed edge does not intersect with any other edge. The corresponding graph is (g) of Fig. \ref{Step7-1con}.

\vspace{1.5mm}

Similarly, from the graph (d) of Fig. \ref{Step7-1con}, a new edge $e_{10,4}$ is created which corresponds to the graph (h) of Fig. \ref{Step7-1con}.
\end{itemize}

Thus, following the above procedures i) to iv), eight graphs (a) to (h) of Fig. \ref{Step7-1con} are produced with respect to the graph (a) of Fig. \ref{Step5-1con}.

In Fig. \ref{Step7-1con}, graphs (i) to (n) are generated randomly from some of the aforementioned 22 graphs. The generation obeys items in Step 7. For example, graph (i) of Fig. \ref{Step7-1con} is produced from (a) of Fig. \ref{Step4-2con} by following items i) and ii) of Step 7.

The topology structures corresponding to each of the graphs in Fig. \ref{Step7-1con} are depicted in Fig. \ref{Step7-1conN}.

\begin{figure}[H]
\begin{center}
\subfigure[]{\includegraphics[width=2.1cm]{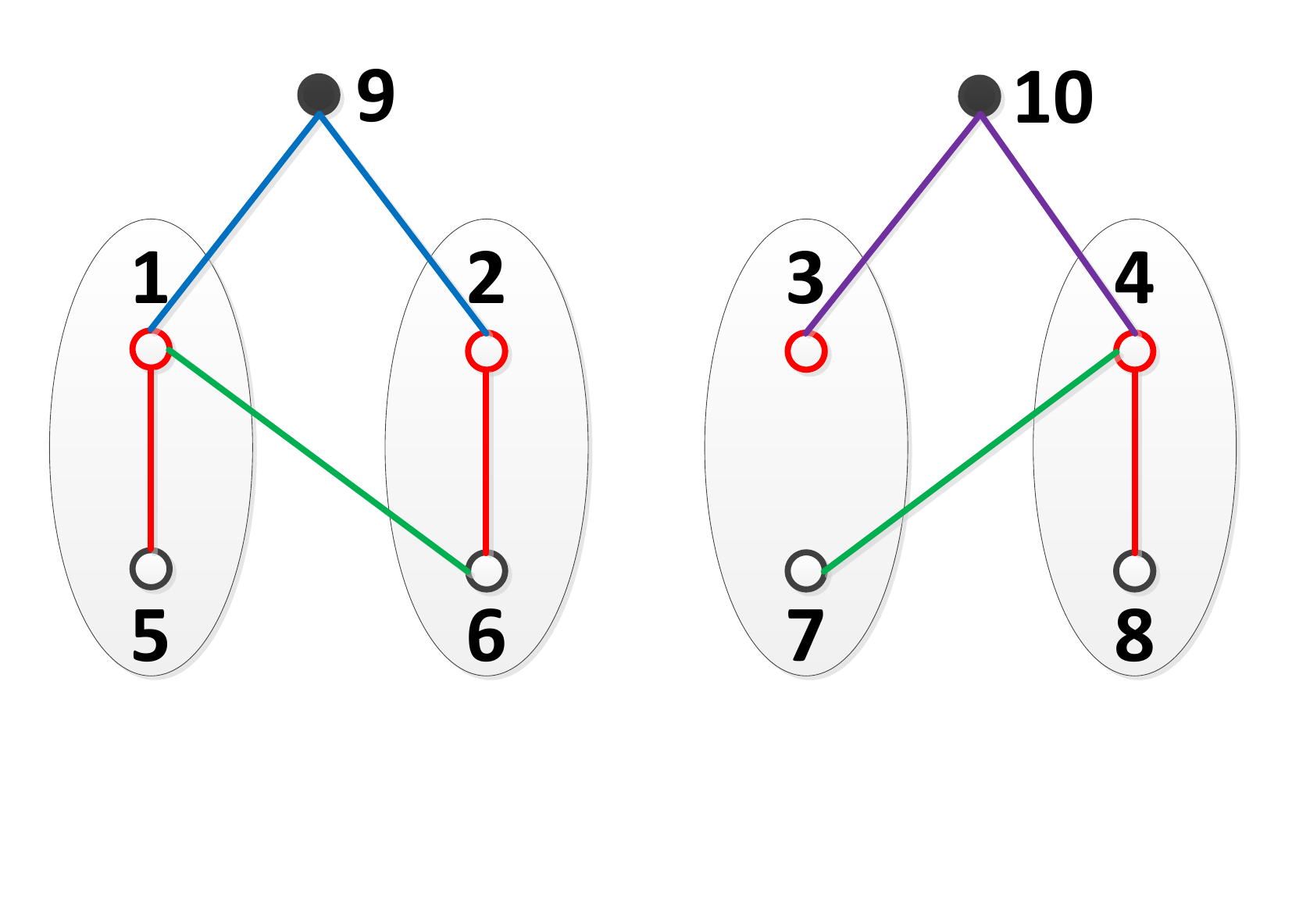}}\quad
\subfigure[]{\includegraphics[width=2.1cm]{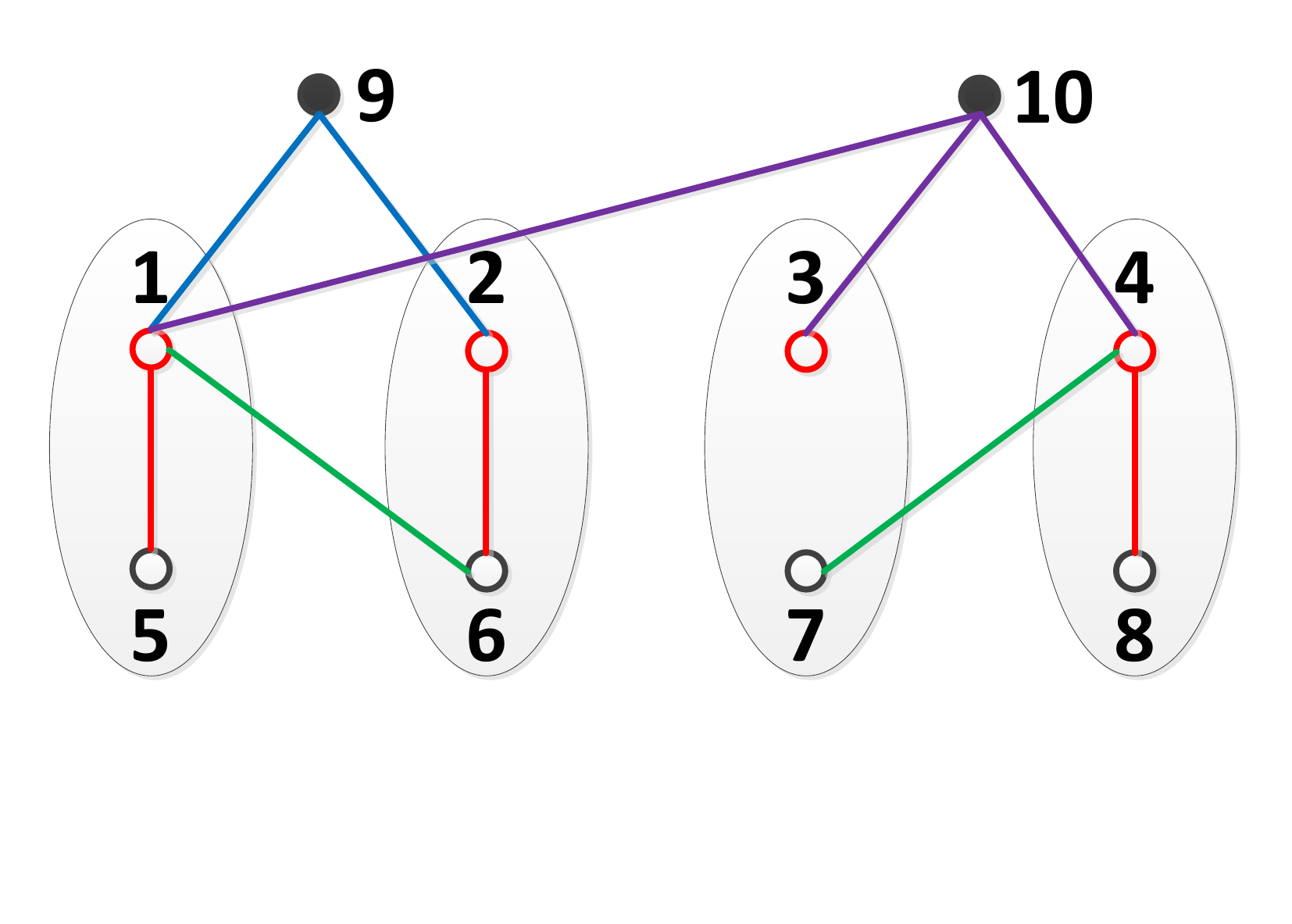}}\quad
\subfigure[]{\includegraphics[width=2.1cm]{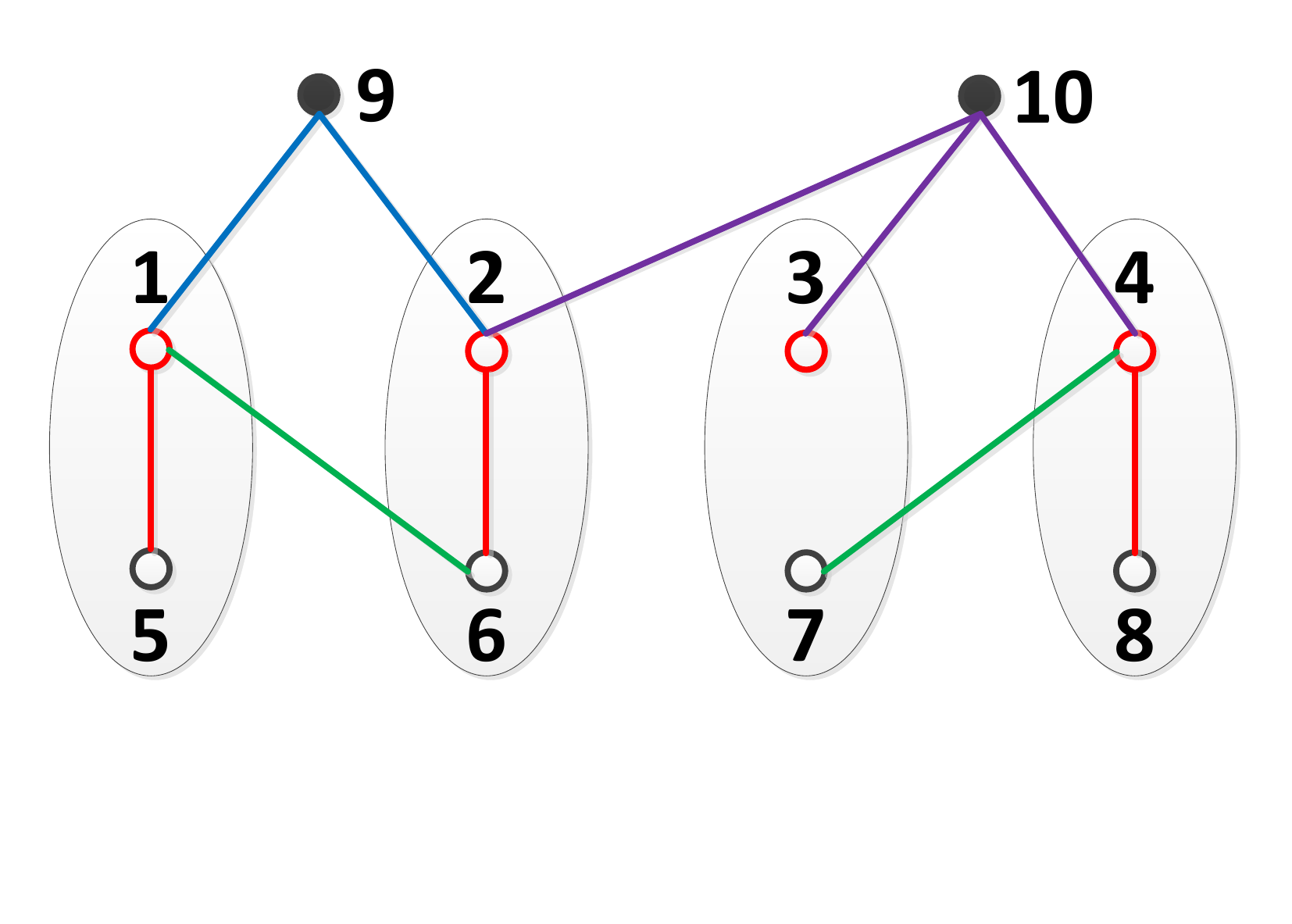}}
\subfigure[]{\includegraphics[width=2.1cm]{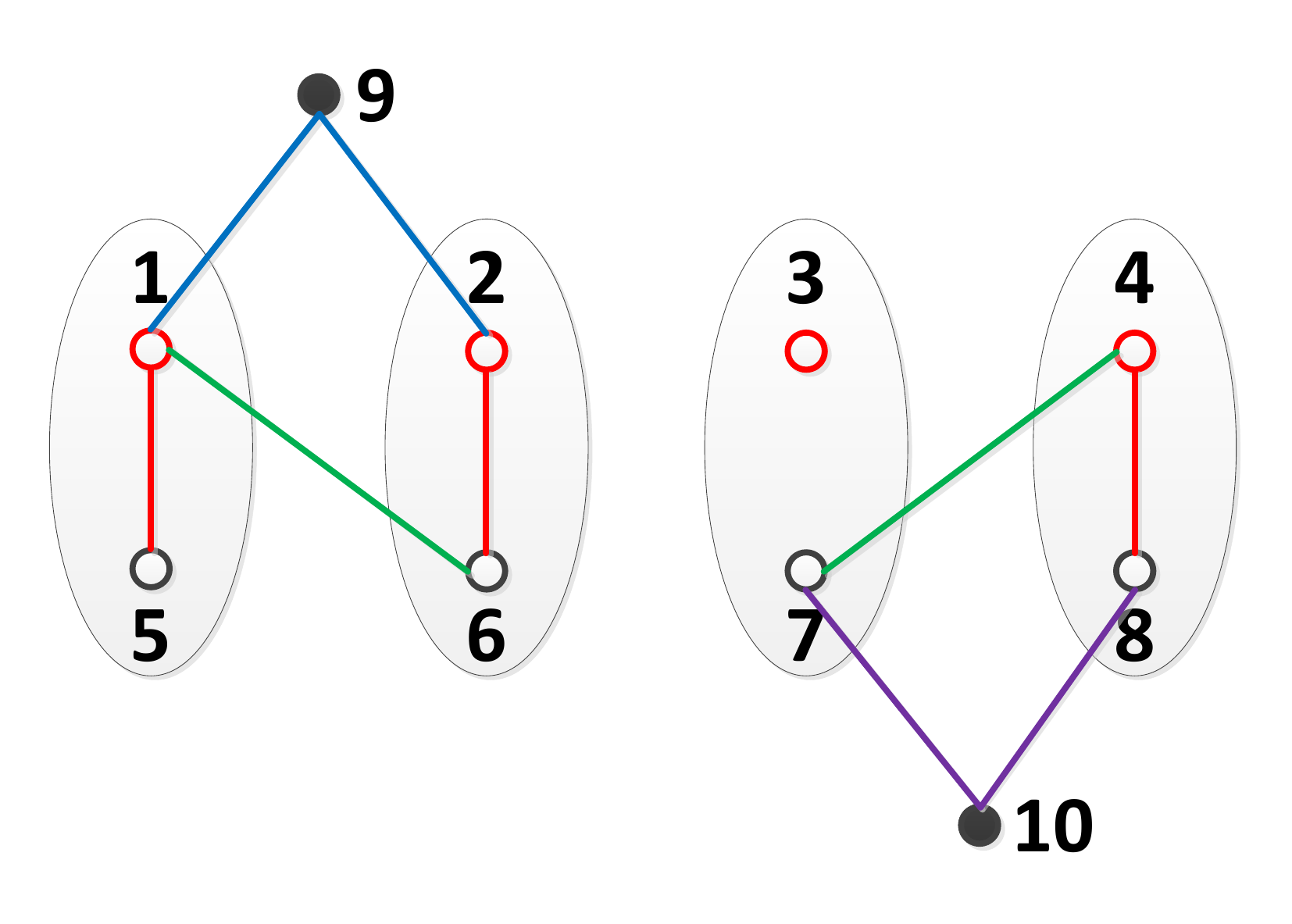}}\quad
\subfigure[]{\includegraphics[width=2.1cm]{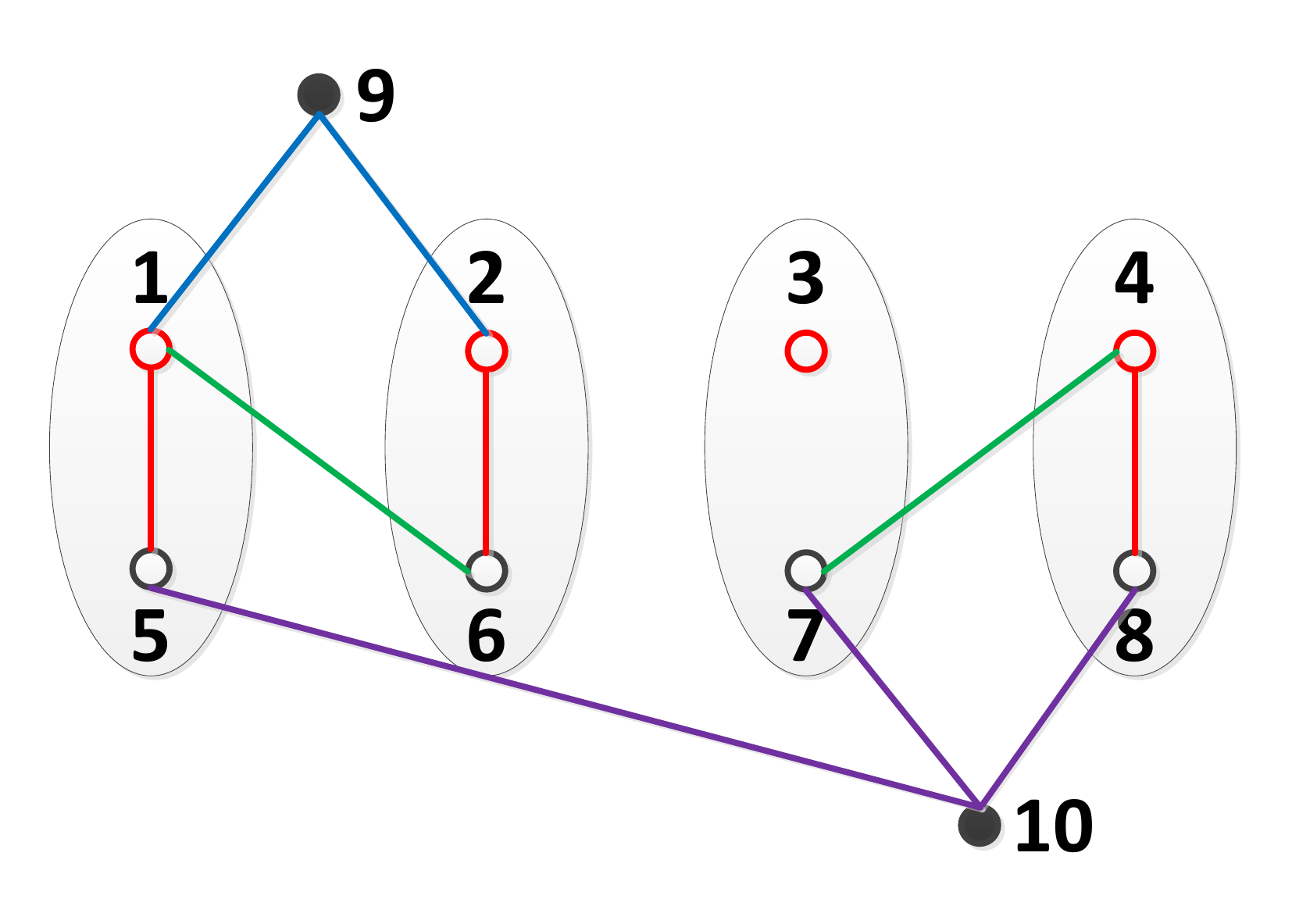}}\quad
\subfigure[]{\includegraphics[width=2.1cm]{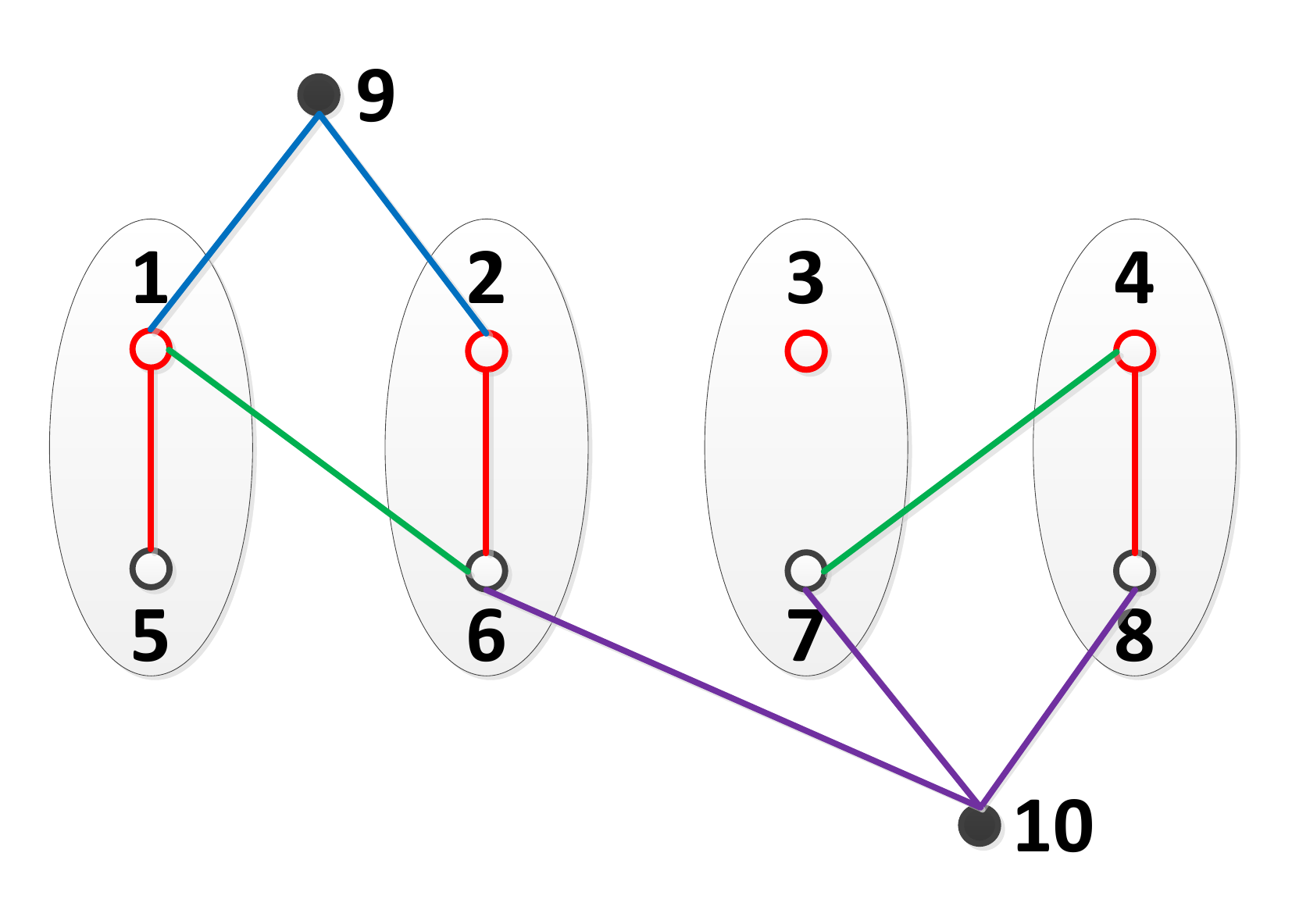}}
\subfigure[]{\includegraphics[width=2.1cm]{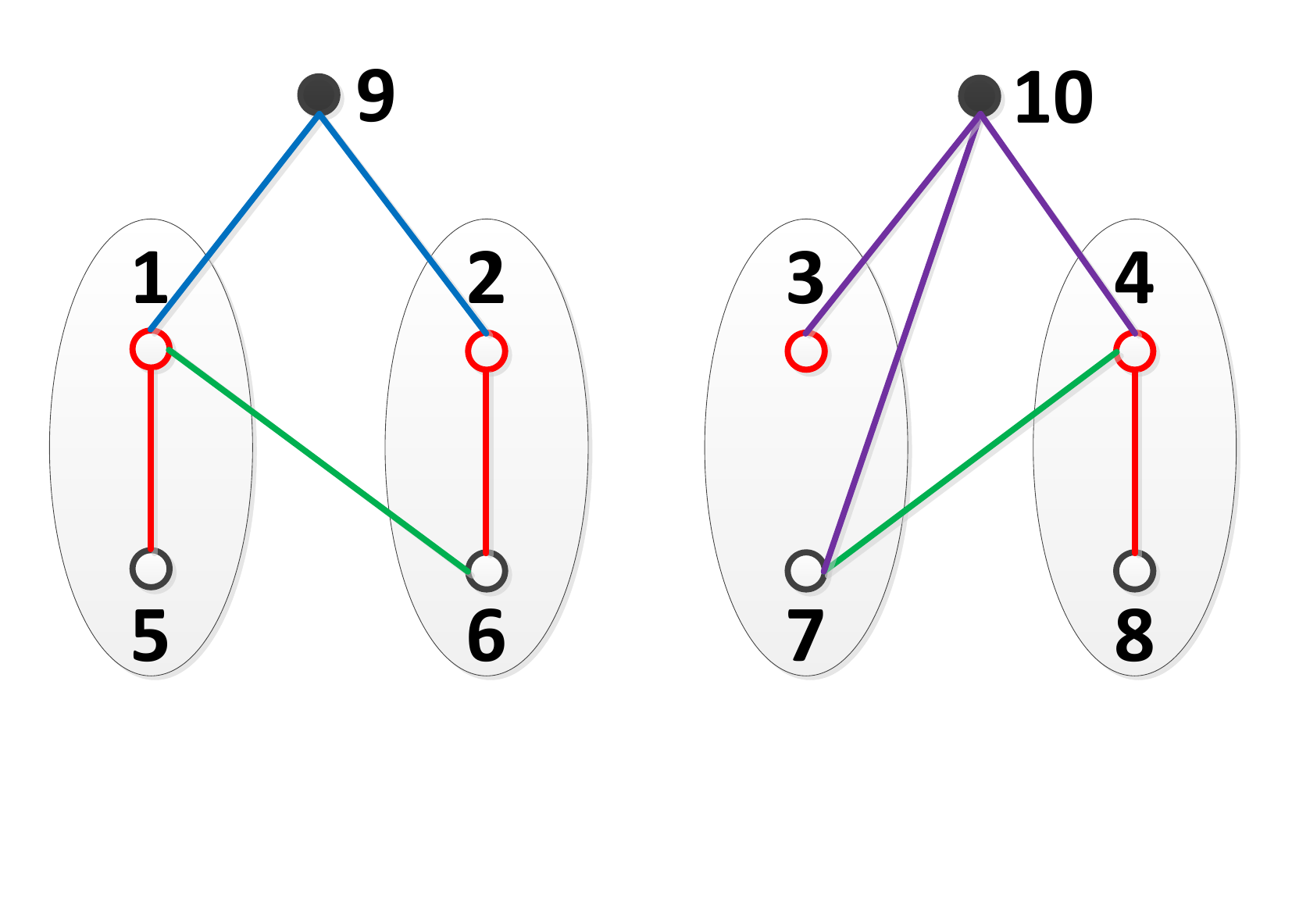}}\quad
\subfigure[]{\includegraphics[width=2.1cm]{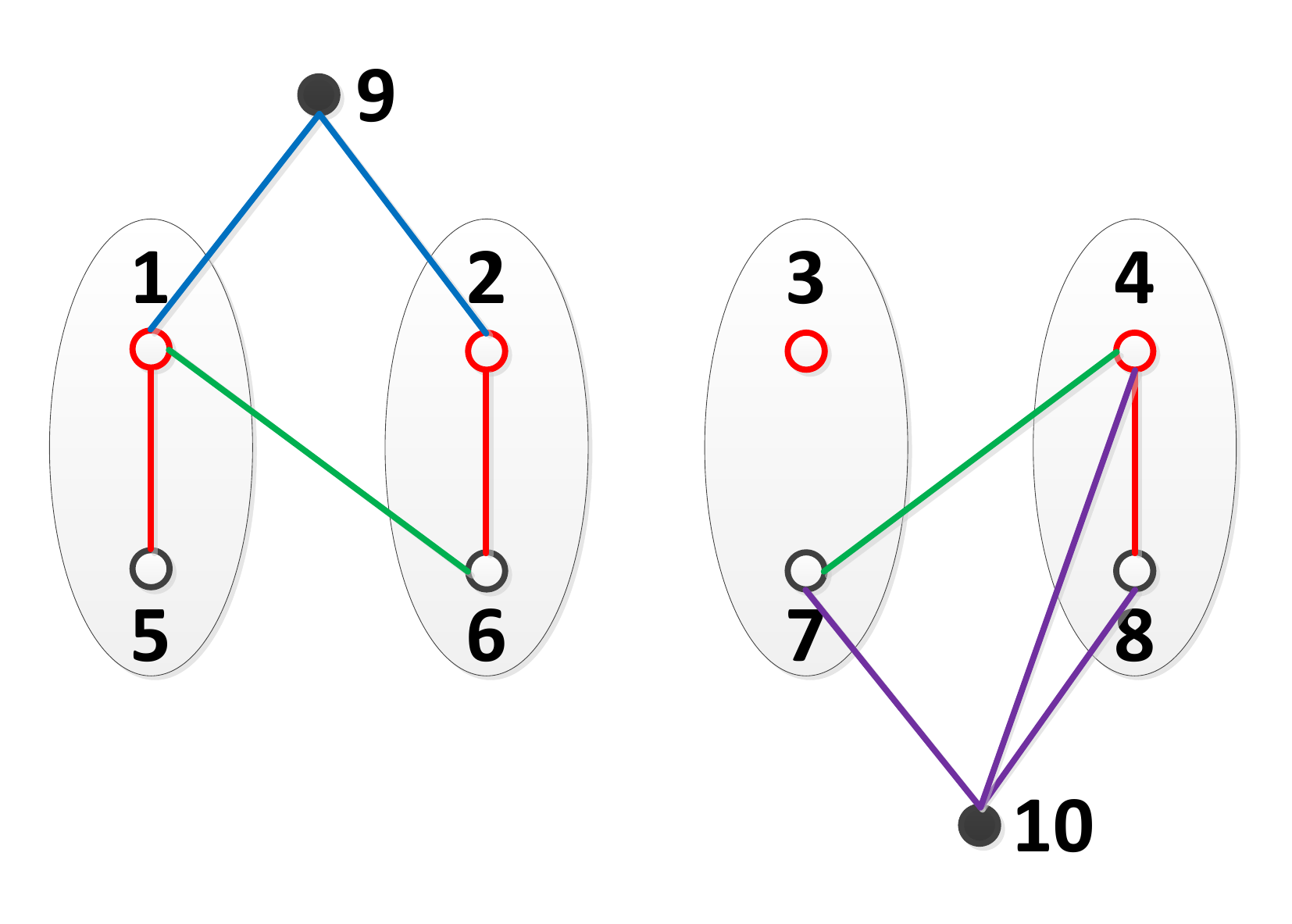}}\quad\\
\subfigure[]{\includegraphics[width=2.1cm]{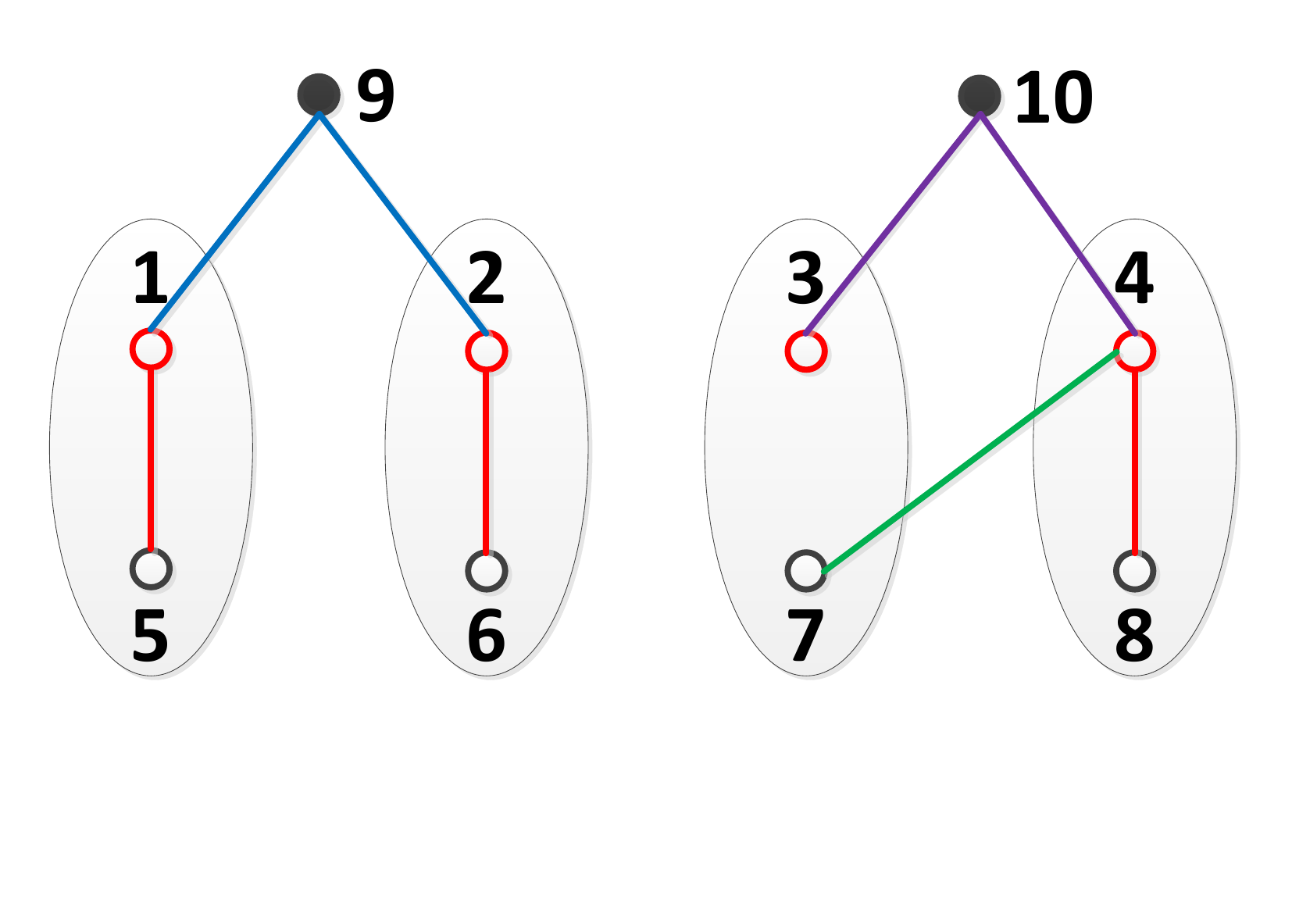}}\quad
\subfigure[]{\includegraphics[width=2.1cm]{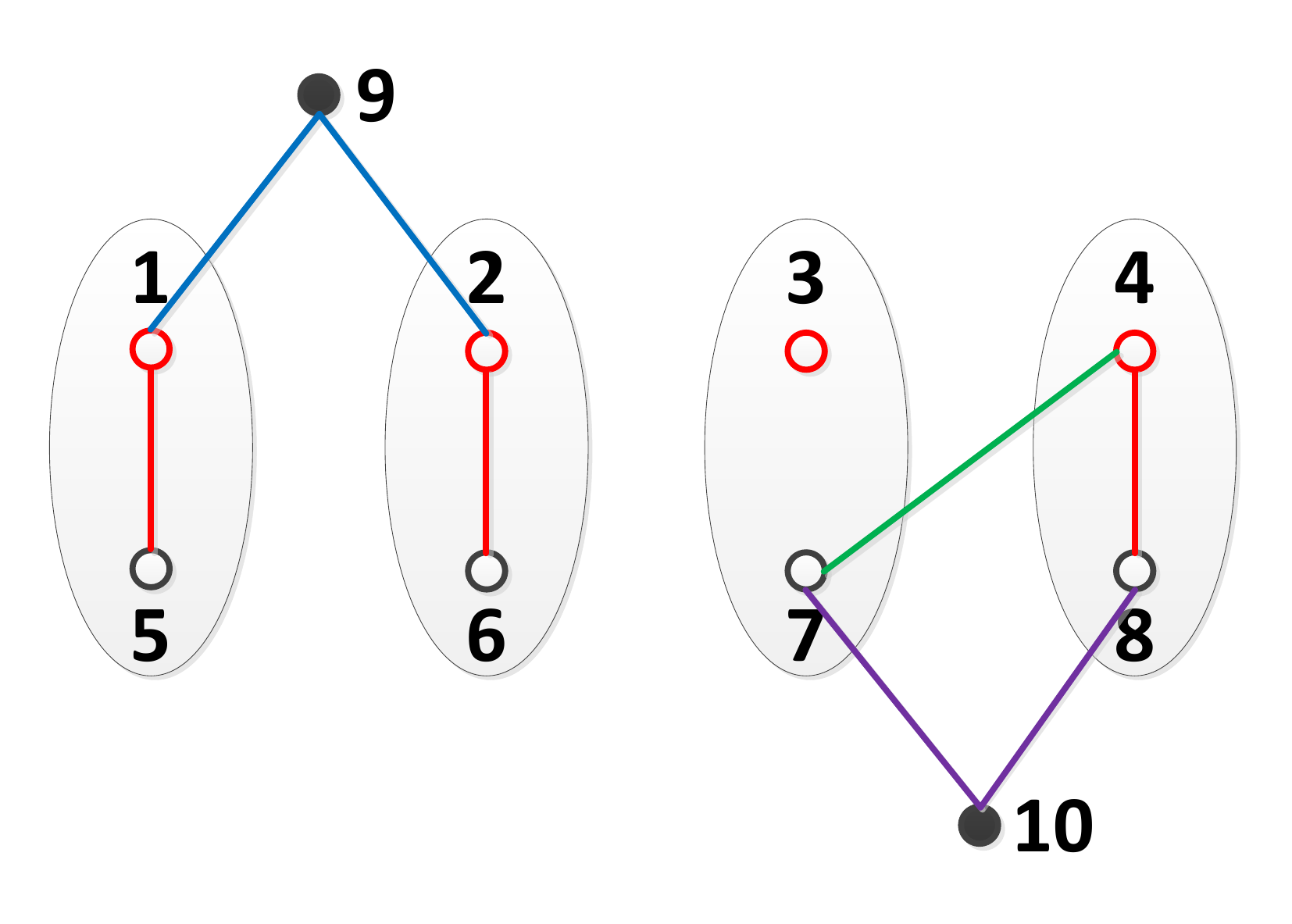}}\quad
\subfigure[]{\includegraphics[width=2.1cm]{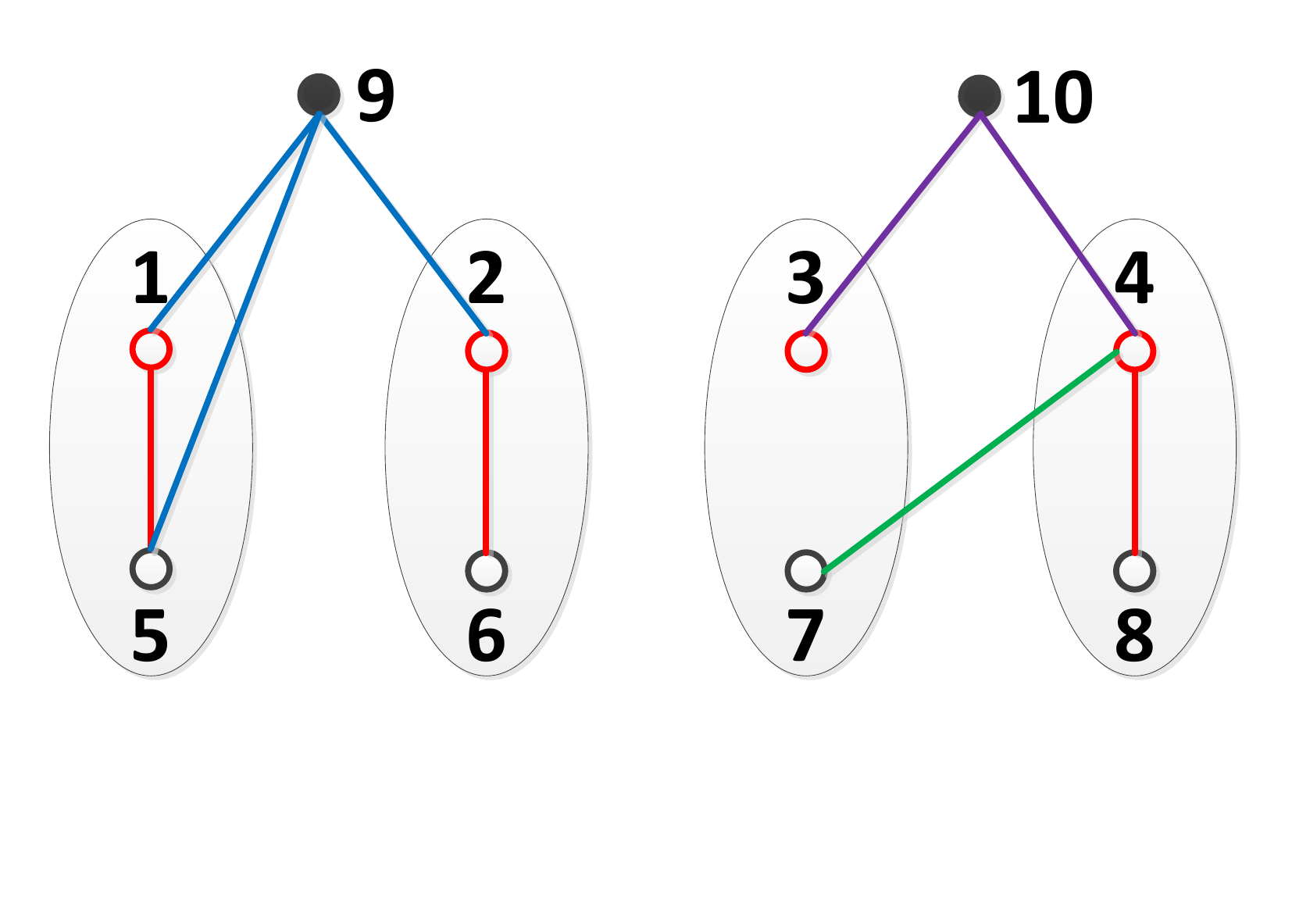}}
\subfigure[]{\includegraphics[width=2.1cm]{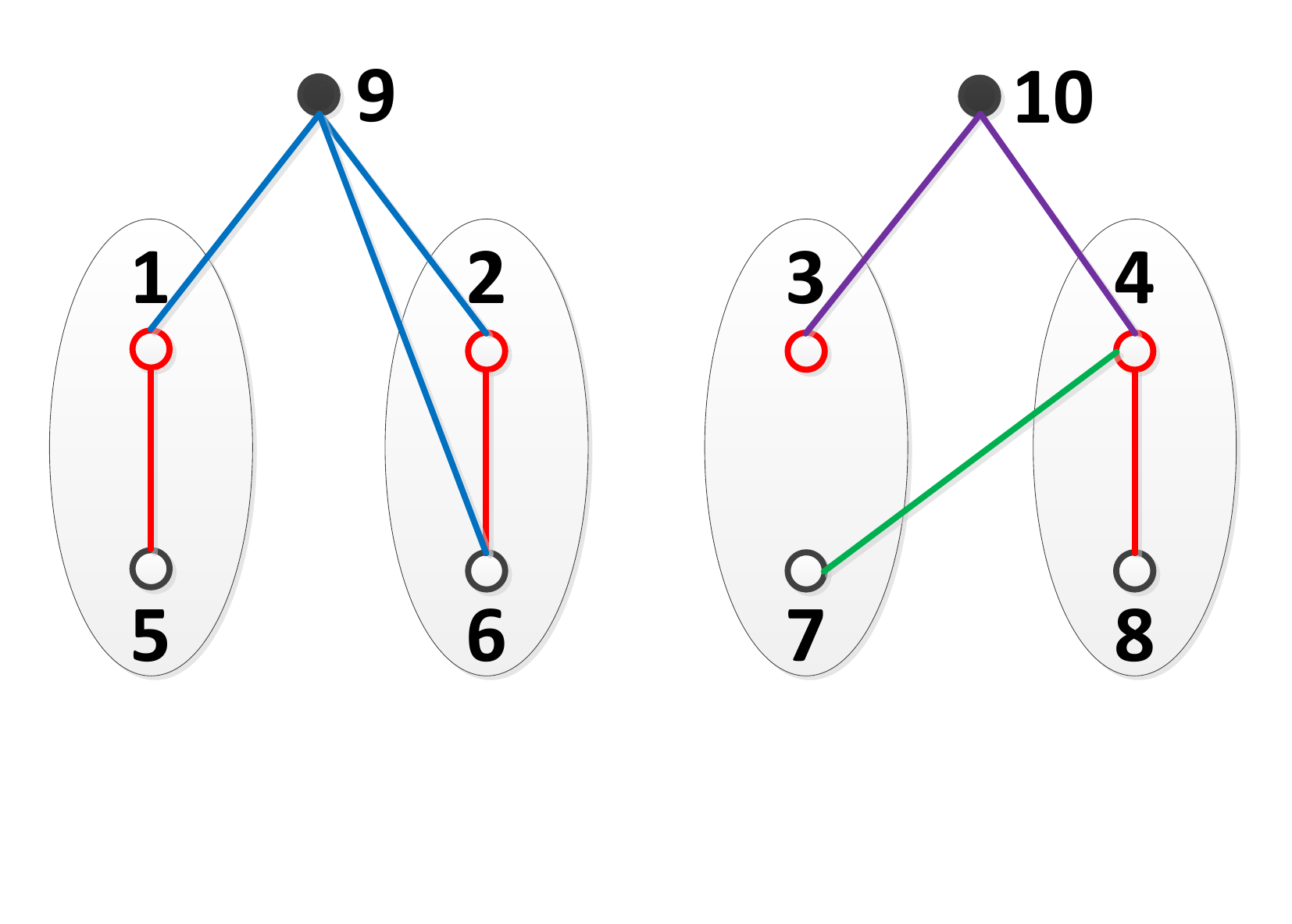}}\quad
\subfigure[]{\includegraphics[width=2.1cm]{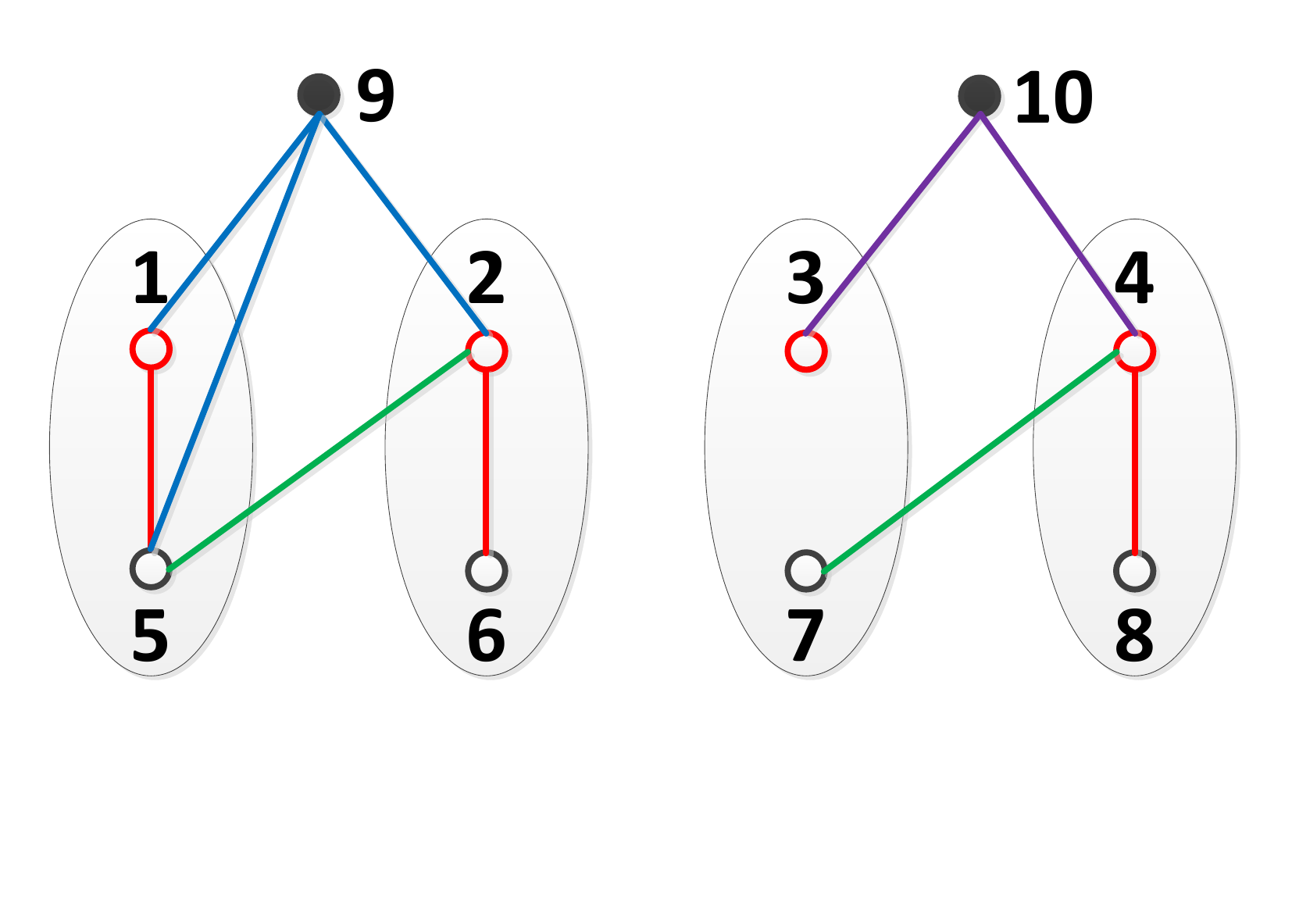}}\quad
\subfigure[]{\includegraphics[width=2.1cm]{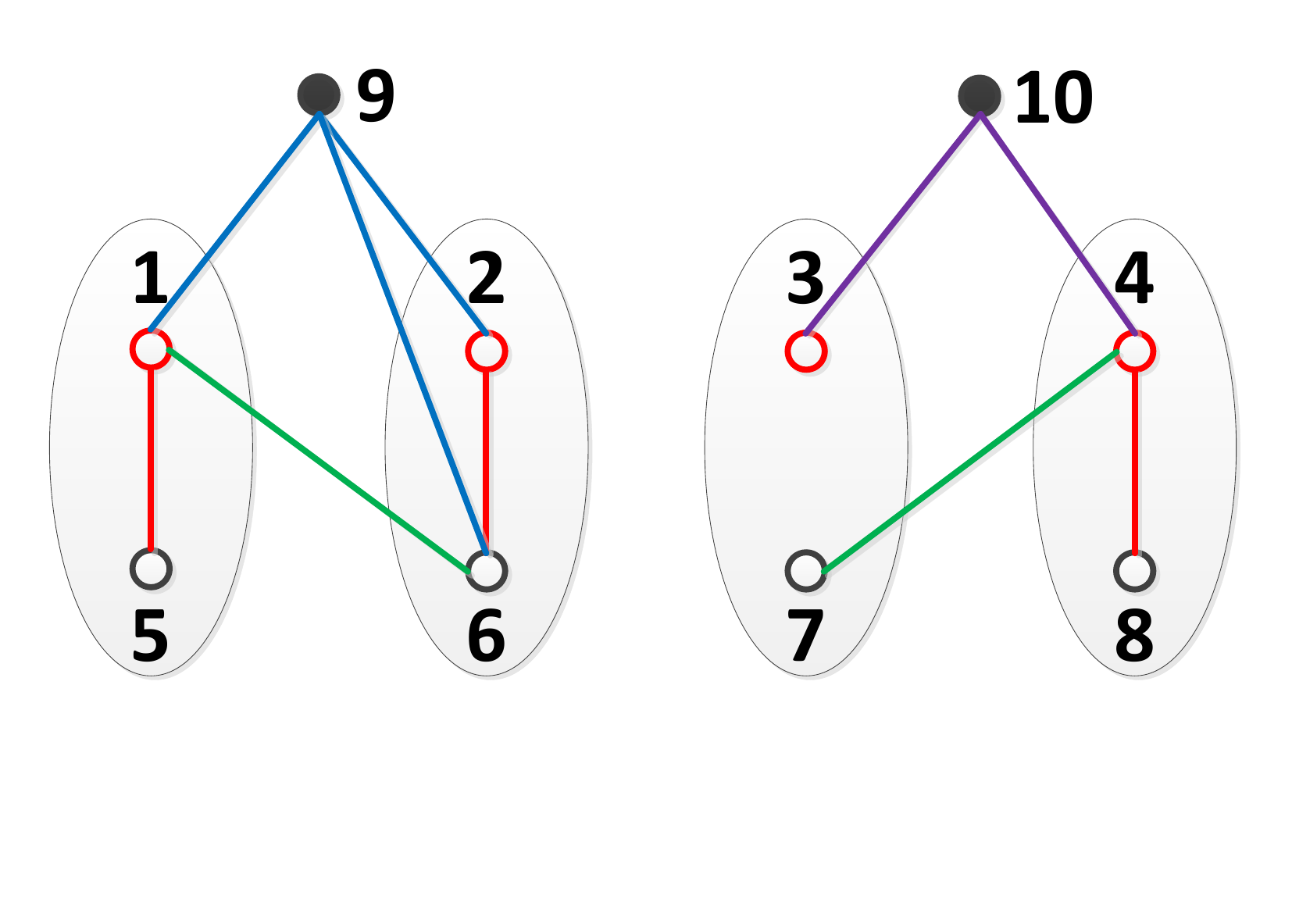}}
\caption{(a) to (n) are graphs designed from some of the aforementioned 22 graphs by following Step 7.}\label{Step7-1con}
\end{center}
\end{figure}

\begin{figure}[H]
\begin{center}
\subfigure[]{\includegraphics[width=2.1cm]{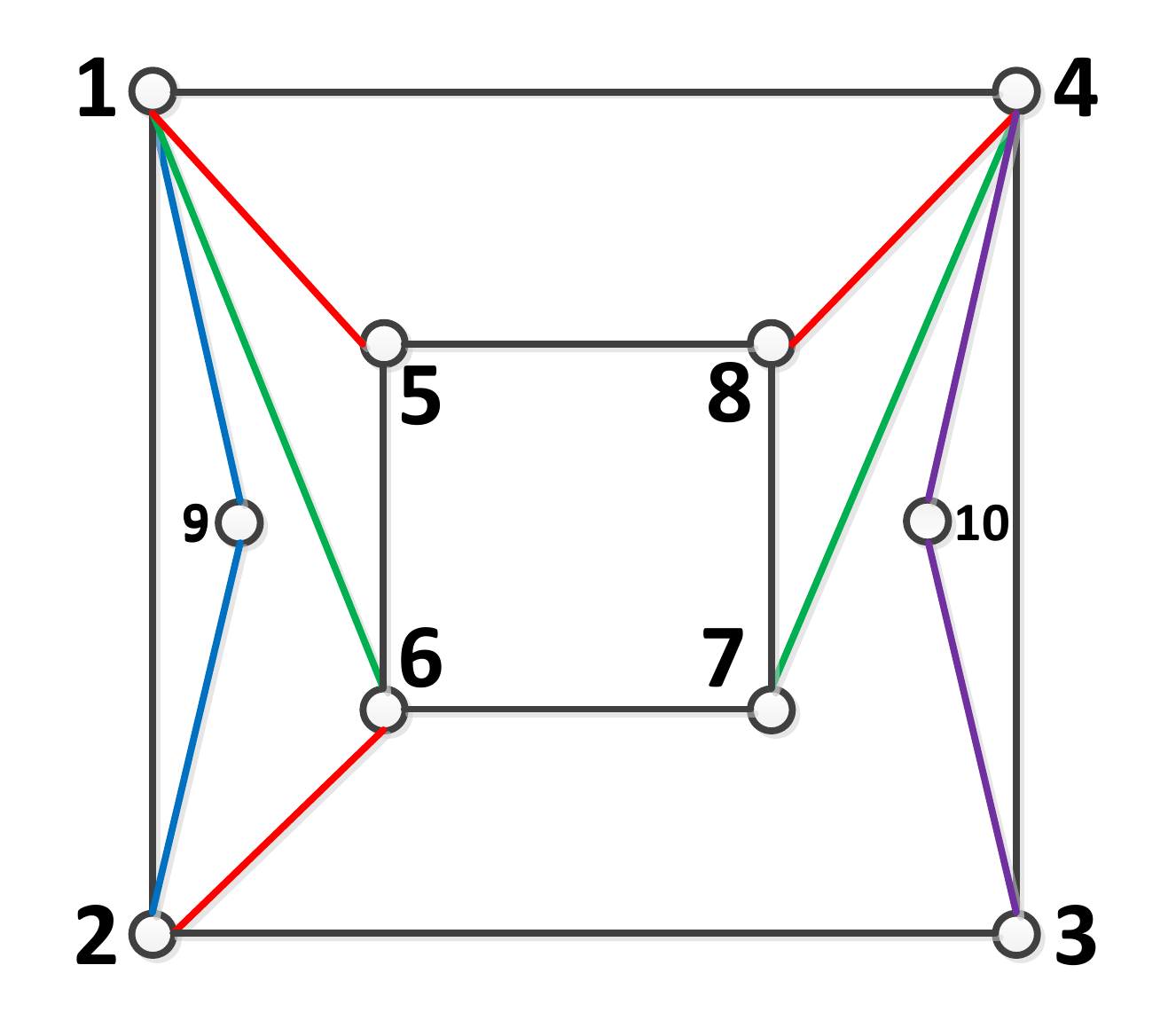}}
\subfigure[]{\includegraphics[width=2.1cm]{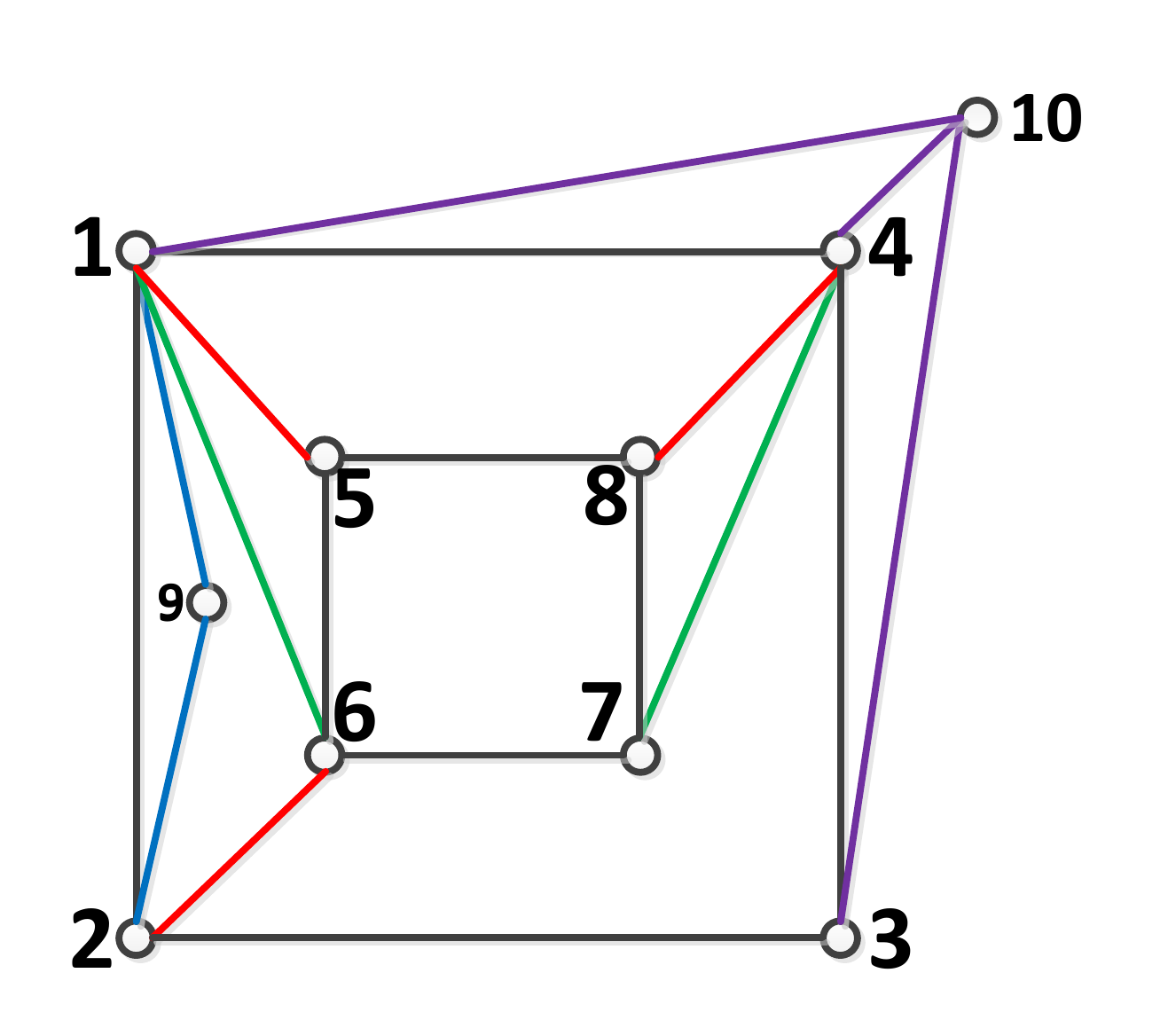}}
\subfigure[]{\includegraphics[width=2.1cm]{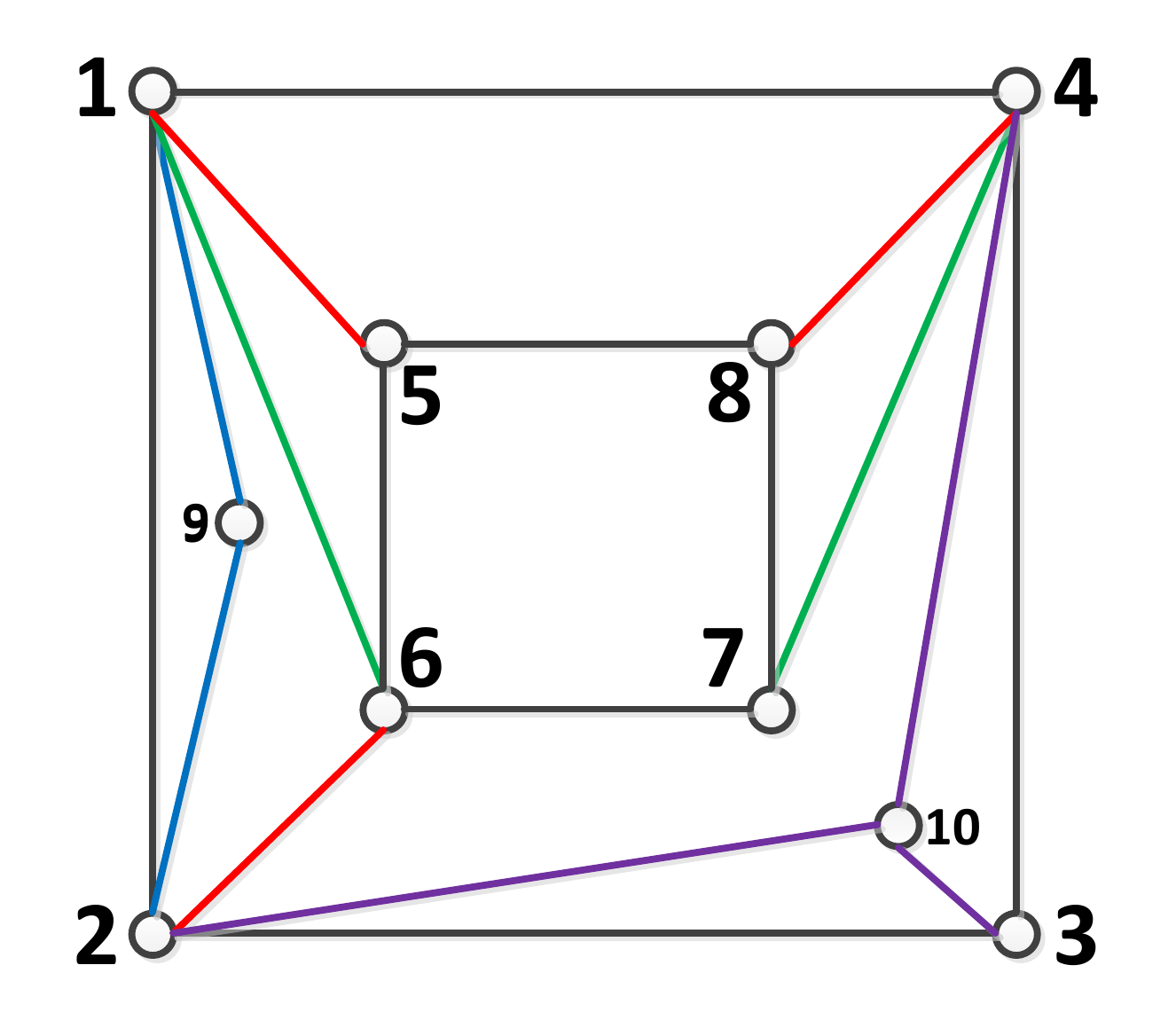}}
\subfigure[]{\includegraphics[width=2.1cm]{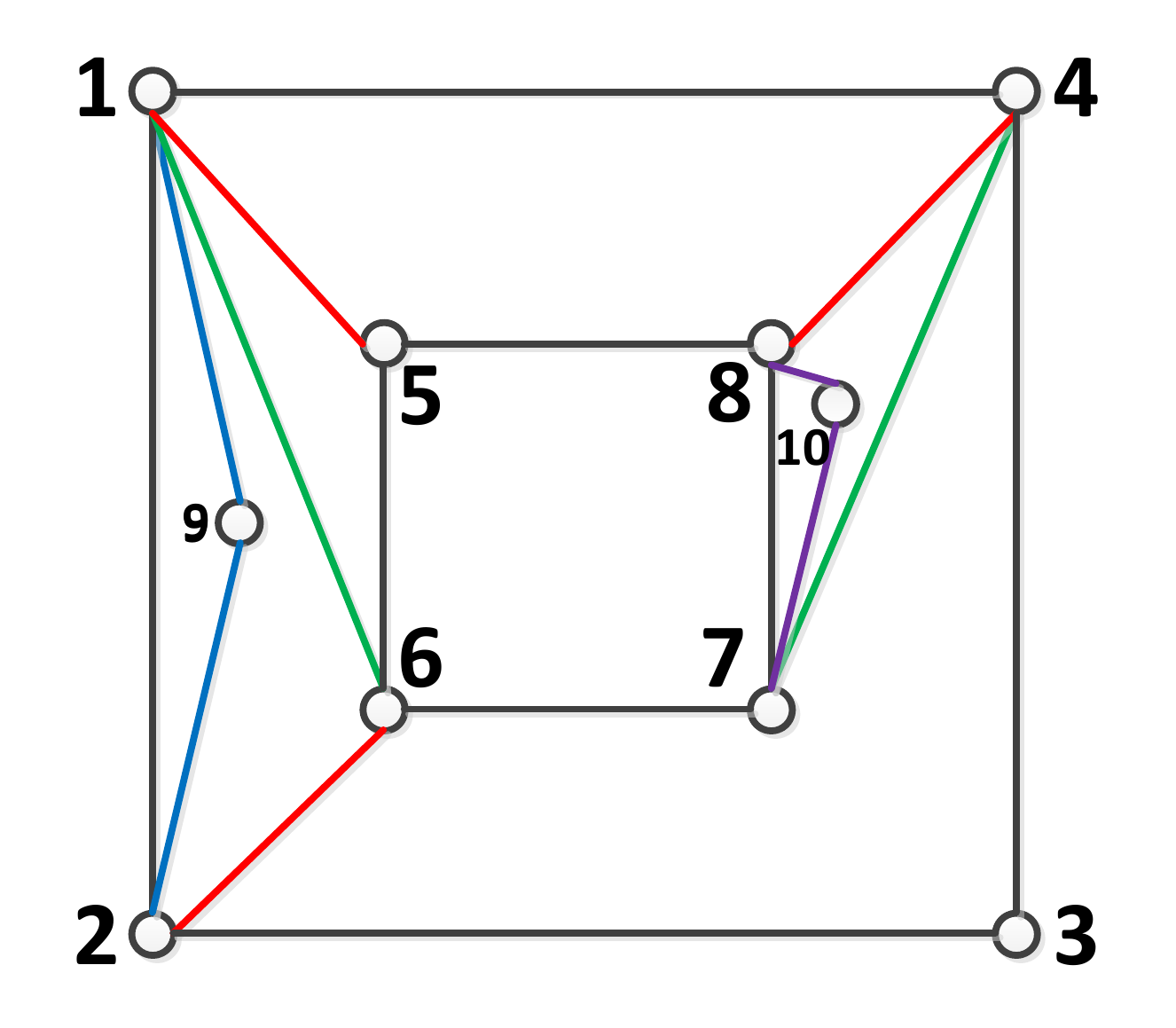}}\\
\subfigure[]{\includegraphics[width=2.1cm]{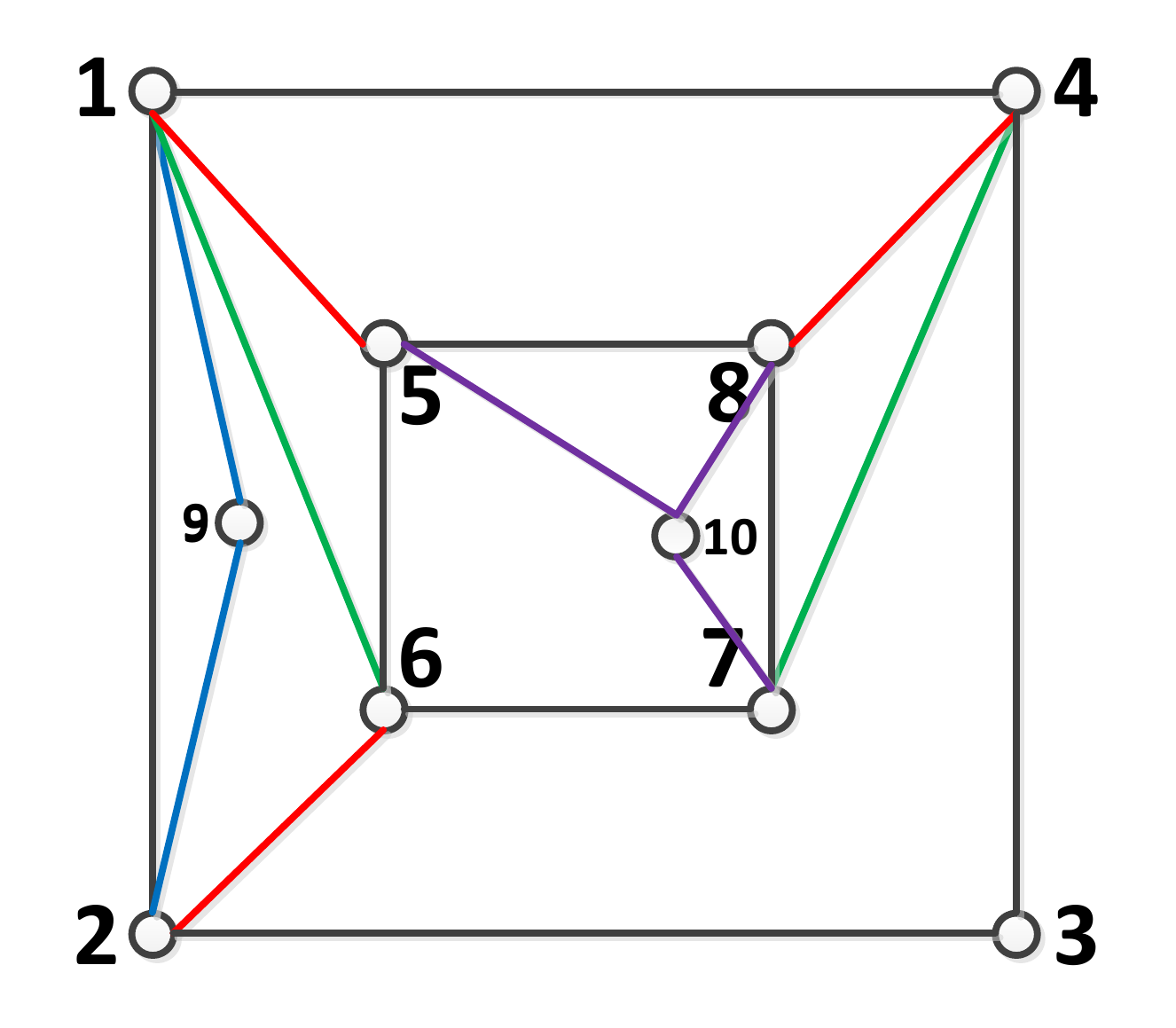}}
\subfigure[]{\includegraphics[width=2.1cm]{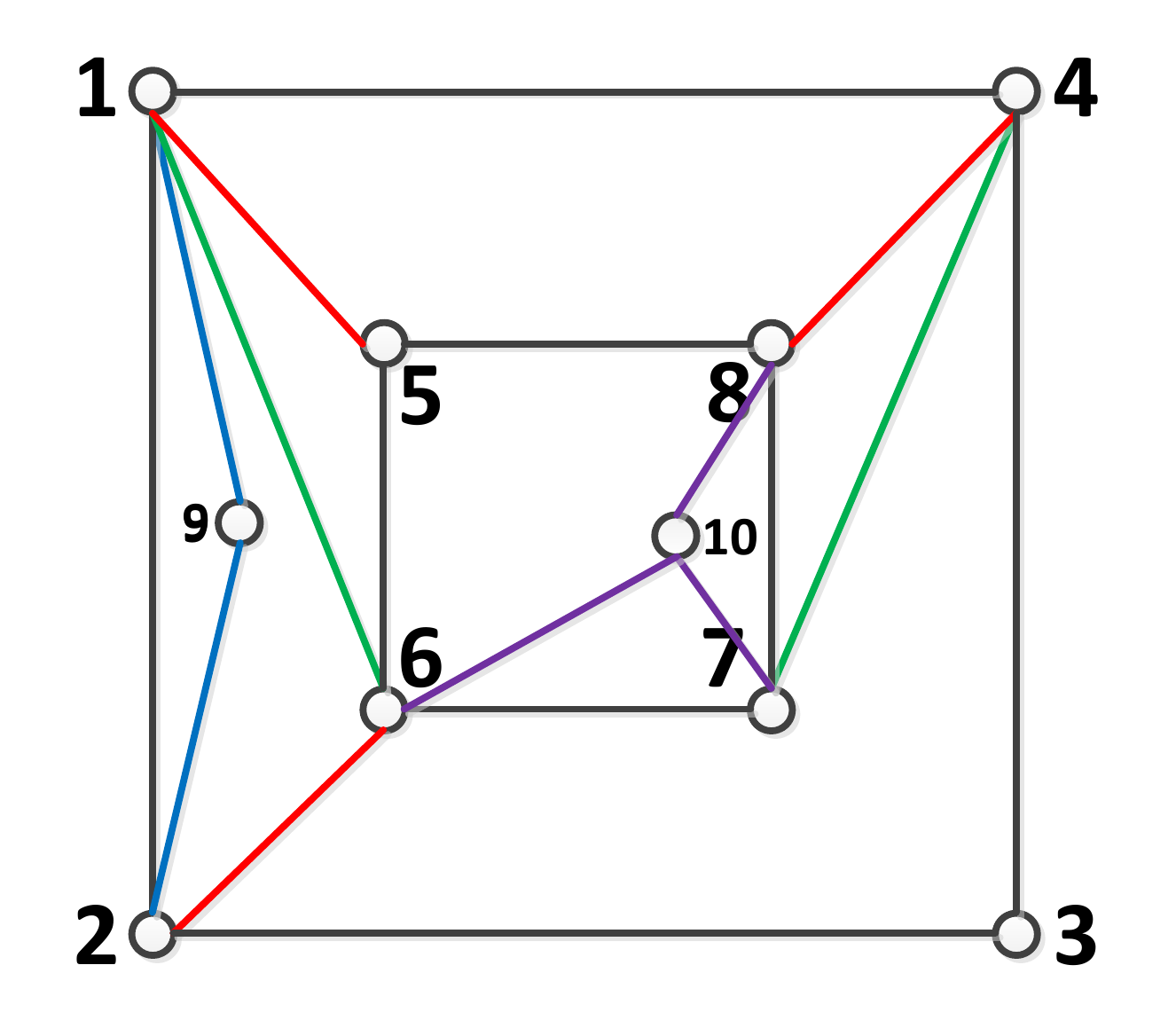}}
\subfigure[]{\includegraphics[width=2.1cm]{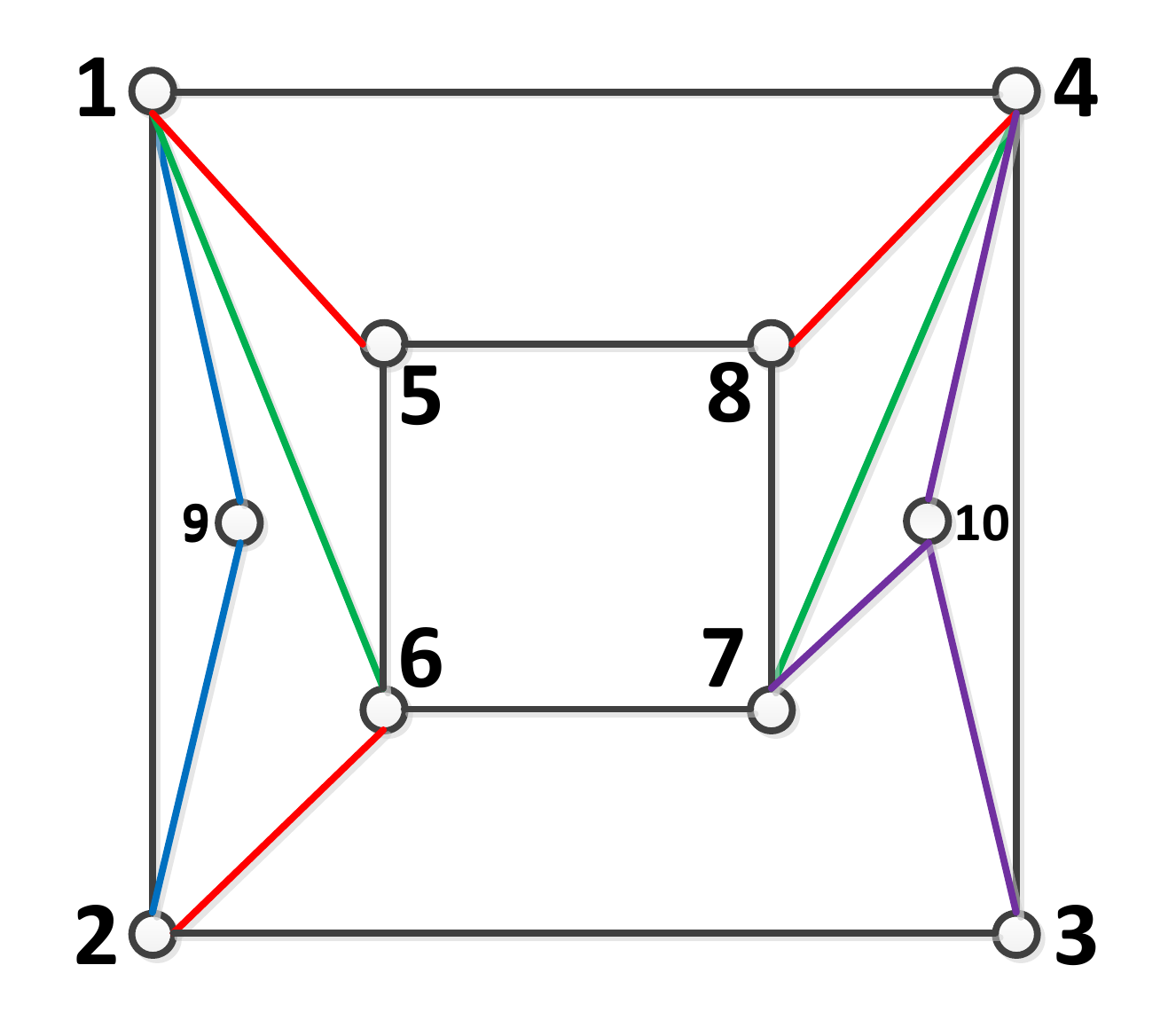}}
\subfigure[]{\includegraphics[width=2.1cm]{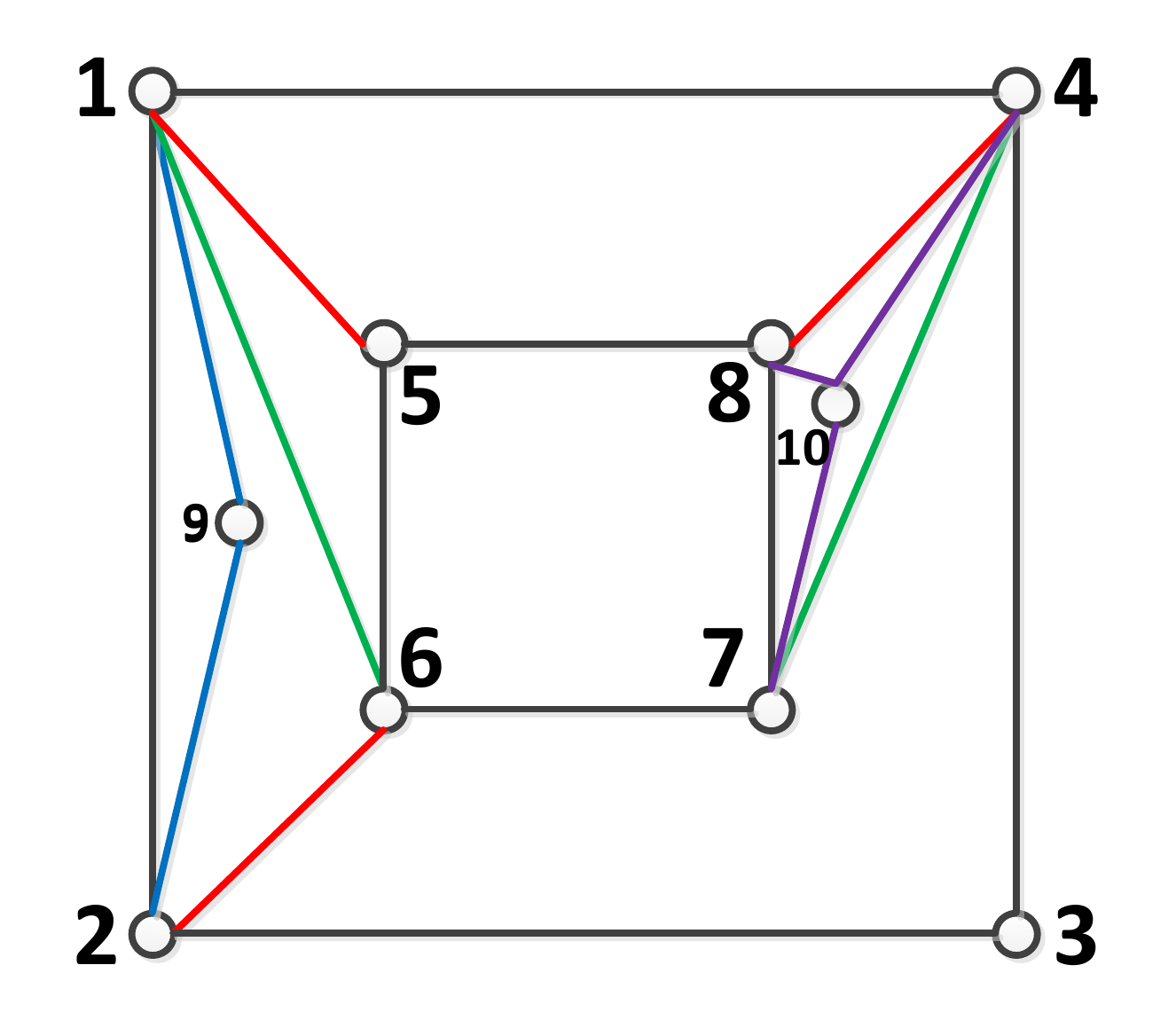}}\qquad
\subfigure[]{\includegraphics[width=2.1cm]{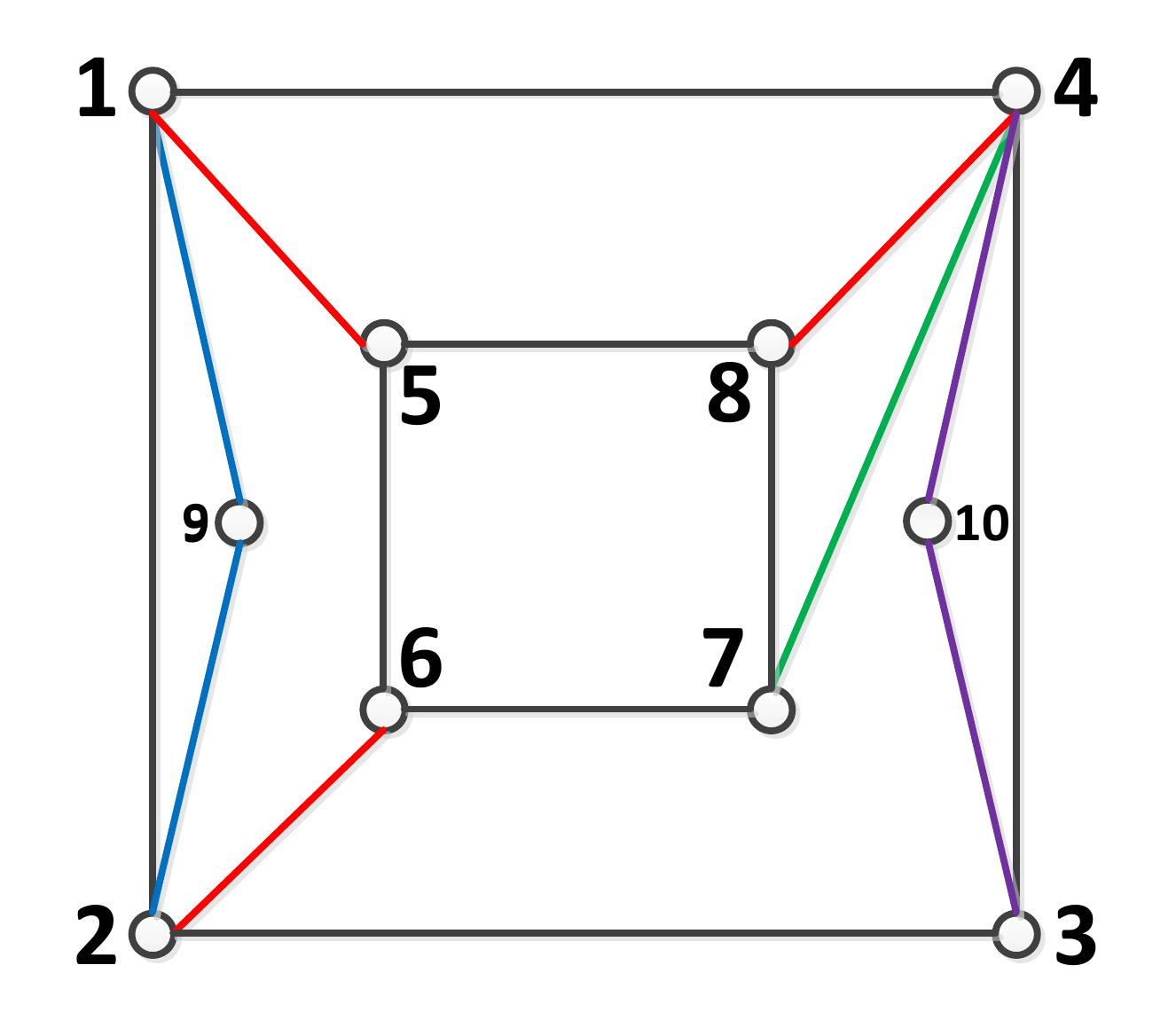}}
\subfigure[]{\includegraphics[width=2.1cm]{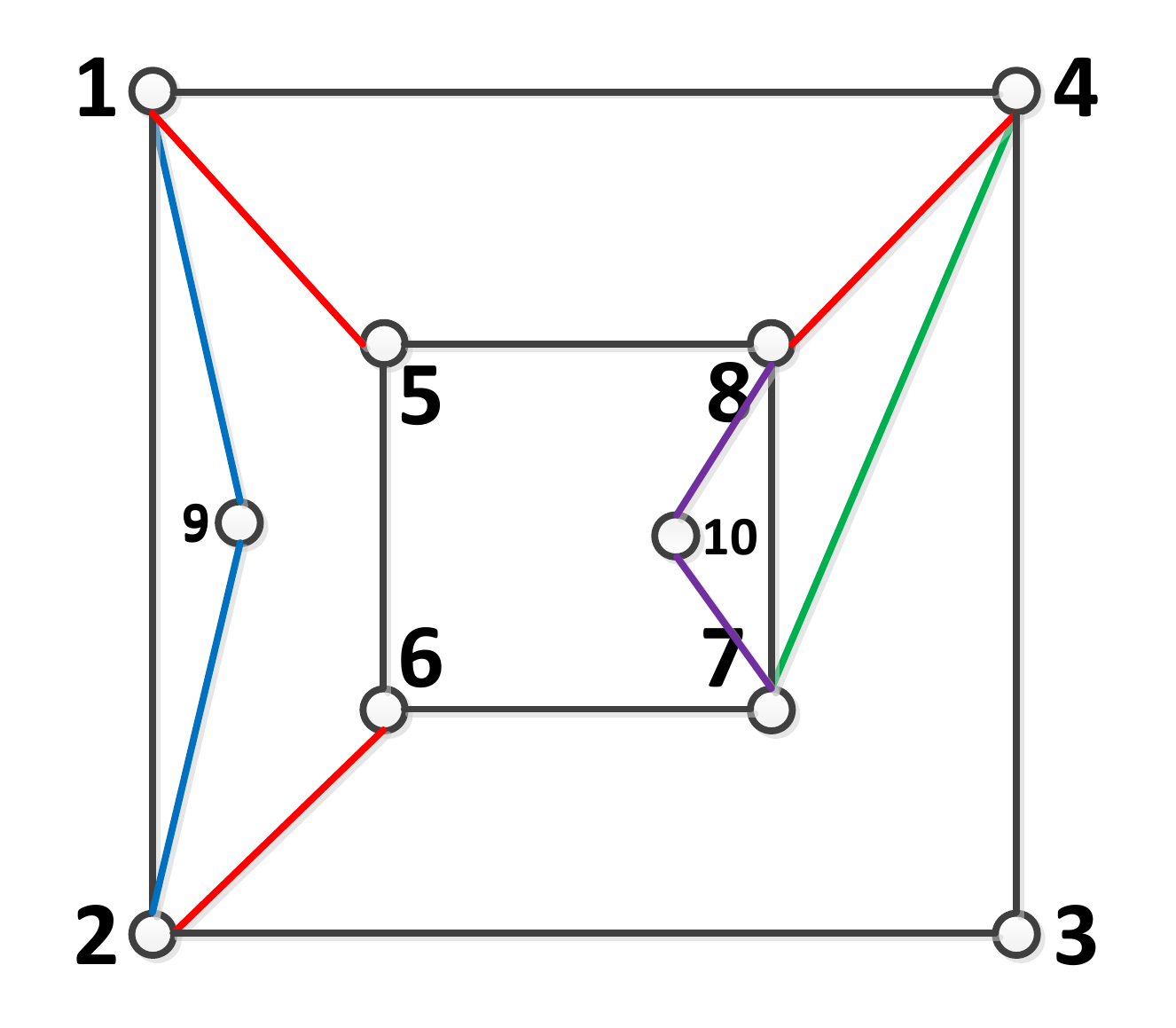}}
\subfigure[]{\includegraphics[width=2.1cm]{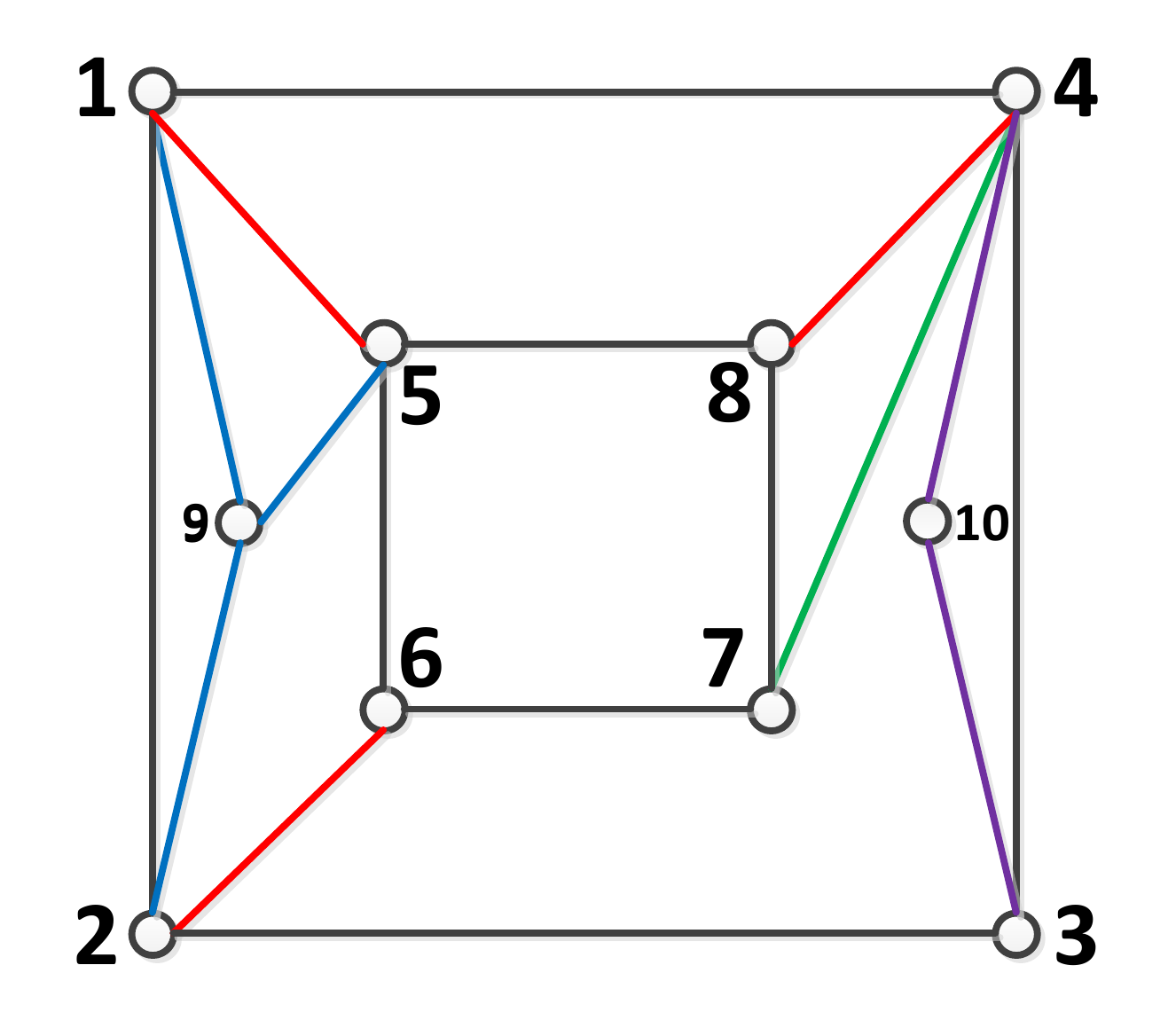}}
\subfigure[]{\includegraphics[width=2.1cm]{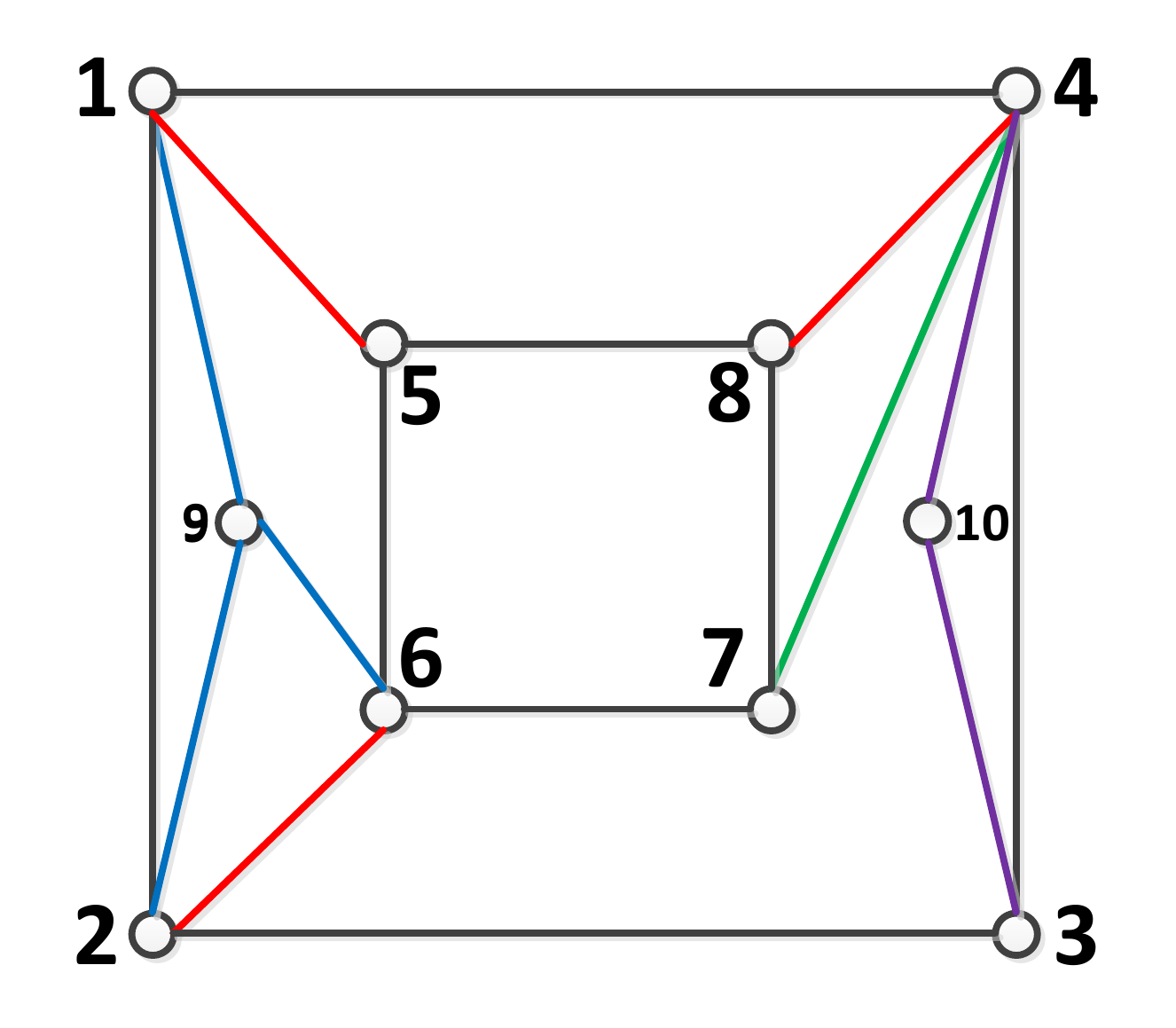}}\\
\subfigure[]{\includegraphics[width=2.1cm]{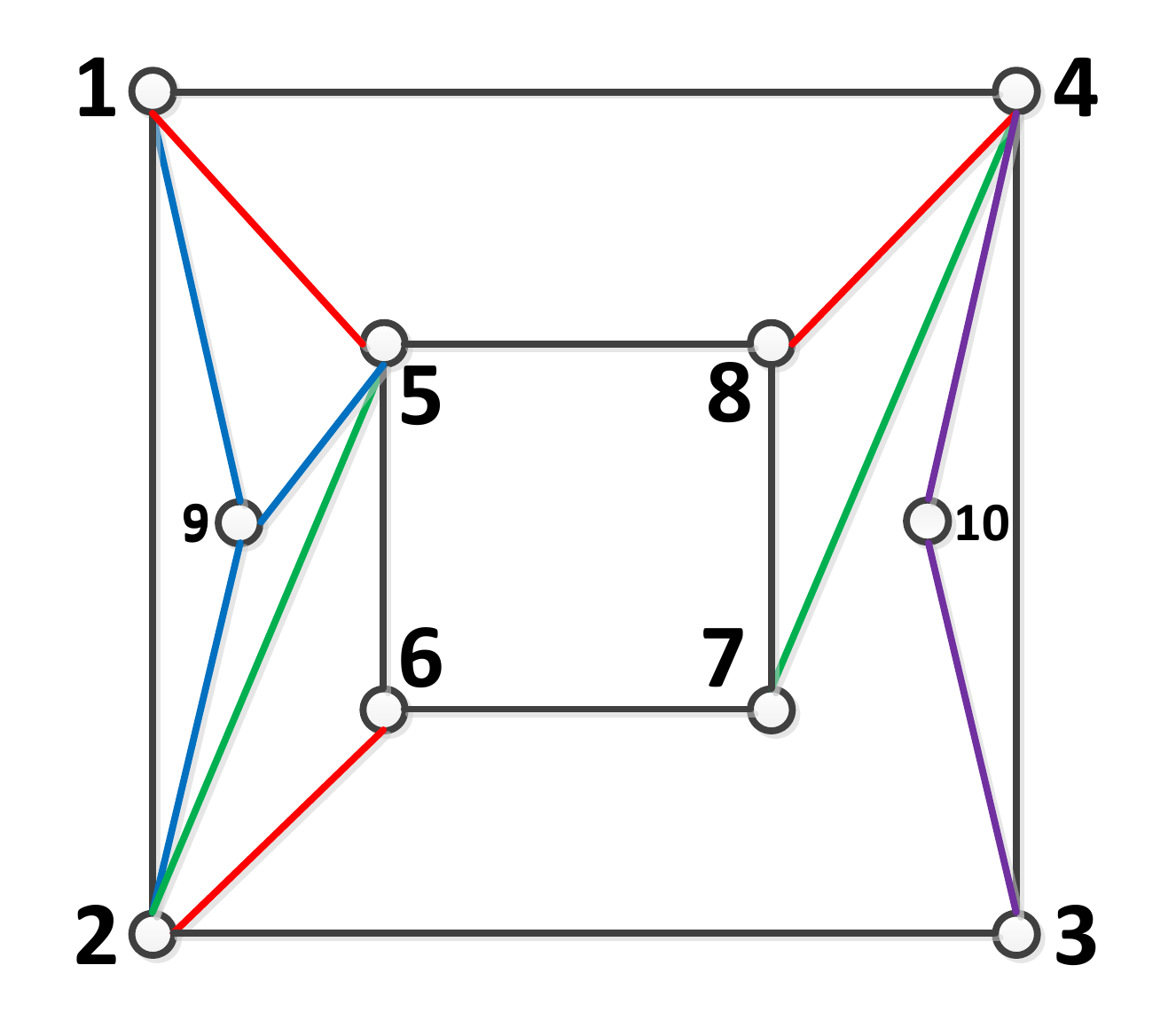}}
\subfigure[]{\includegraphics[width=2.1cm]{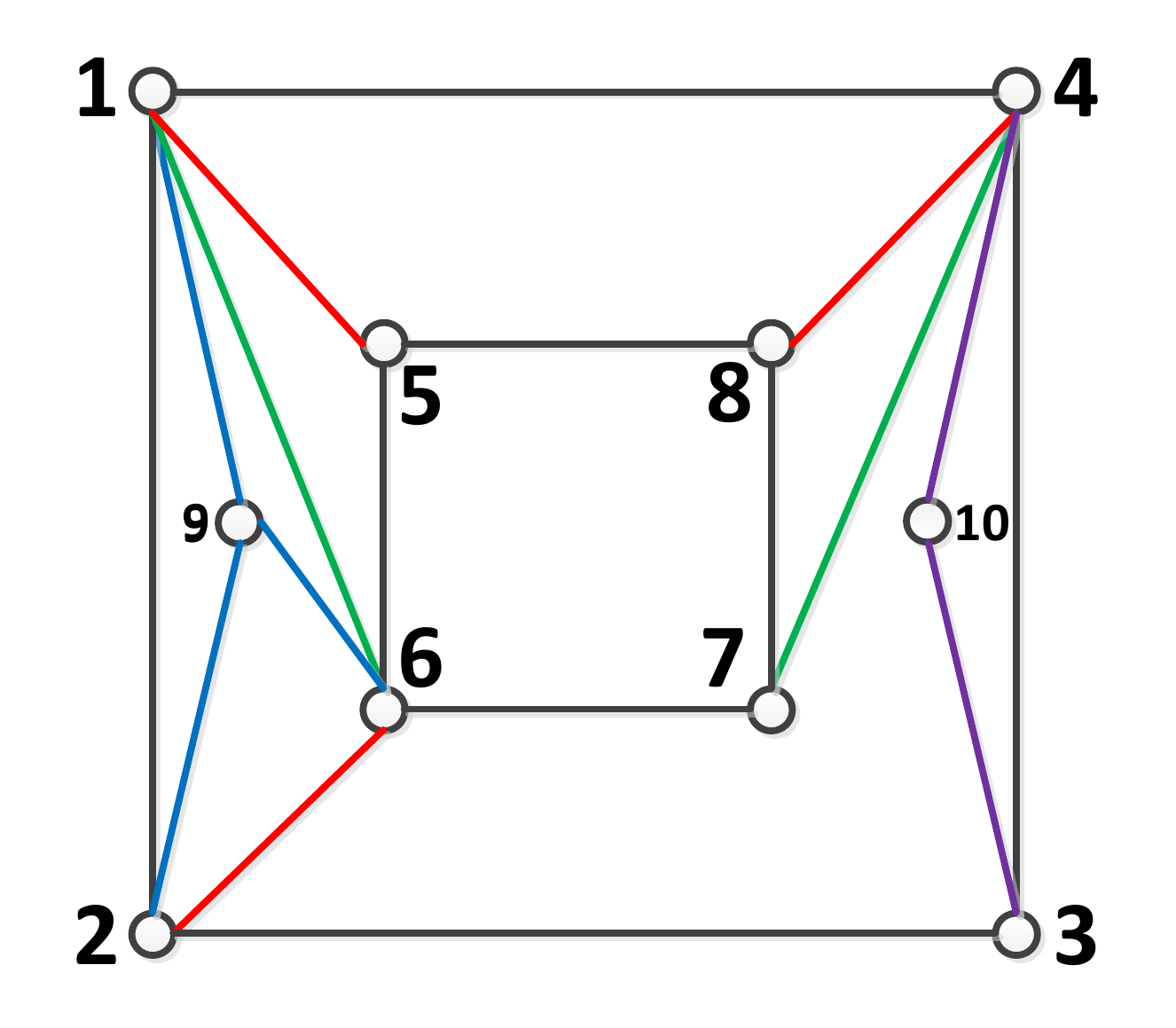}}
\caption{(a) to (n) are the topology structures corresponding to, respectively, each of the figures in Fig. \ref{Step7-1con}.}\label{Step7-1conN}
\end{center}
\end{figure}

The following theorem indicates that a graph designed by Step 7 can be a perfectly controllable graph with big probability.

\begin{theorem}
All the graphs in Fig. \ref{Step7-1conN} except the one (j) are perfectly controllable.
\end{theorem}

\begin{remark}
The above efforts are aimed at obtaining a large number of perfectly controllable graphs by following the rules proposed in the design procedure rather than proposing a construction approach that applies to all situations. By adding nodes and edges, countless graphs can be constructed from the original topology (b) of Fig.\ref{Step1con}, and accordingly the number of perfectly controllable graphs is also infinite. It is not realistic to construct all perfectly controllable graphs just by one design method.
\end{remark}

\section{Conclusions}
Different from the usual manifestation of the reported results, the rules obeyed in the proposed design procedure are the main contribution of the note. The rules were stated in the design steps, not in the form of a proposition. Starting from a simple topology structure, lots of perfectly controllable graphs were constructed by following the proposed design procedure. From our own point of view, how to obtain the desired topology is an important part of understanding the graph theoretical meaning of multi-agent controllability, not just the topology structure itself.

Besides, the concept of perfectly controllable graph has been proposed for the first time in the note, which tells us that there are controllable interconnection graphs no matter how the number and locations of leader agents are selected. A necessary and sufficient algebraic condition was derived for the perfect controllability. Many perfectly controllable graphs are constructed to verify the reliability of the results.




\begin{thebibliography}{10}





\bibitem{XiaoFTAC2018}
F. Xiao, Y. Shi, and W. Ren.
\newblock Robustness analysis of asynchronous sampled-data multi-agent networks with time-varying delays.
\newblock {\em IEEE Trans Automat Contr}, 2018, DOI: 10.1109/TAC.2017.2756860

\bibitem{MaJTIE2017}
J. Ma, Y. Zheng, and L. Wang.
\newblock Nash equilibrium topology of multi-agent systems with competitive groups.
\newblock {\em IEEE Transactions on Industrial Electronics}, 64(6):4956-4966, 2017.

\bibitem{LTianAccess2018}
Lei Tian, Zhijian Ji, Ting Hou, and Kaien Liu.
\newblock Bipartite consensus on coopetition networks with time-varying delays.
\newblock{\em IEEE Access}, 6(1): 10169-10178, Feb. 2018

\bibitem{ZhenMW2018}
Y. Zheng, J. Ma, and L. Wang.
\newblock Consensus of hybrid multi-agent systems.
\newblock {\em IEEE Transactions on Neural Networks and Learning Systems}, 29(4):1359-1365, 2018.

\bibitem{CaiN2017Complexity}
N. Cai, C. Diao and M. J. Khan.
\newblock A novel clustering method based on quasi-consensus motions of dynamical multiagent systems. \newblock {\em Complexity}, 4978613, 2017

\bibitem{XiAccess2018}
J. Xi, C. Wang, H. Liu, and Z. Wang.
\newblock Dynamic output feedback guaranteed-cost synchronization for multiagent networks with given cost budgets.
\newblock {\em IEEE Access}, 2018, DOI: 10.1109/ACCESS.2018.2819989.

\bibitem{LiuJiCyb2017}
K. Liu, Z. Ji, and W. Ren.
\newblock Necessary and sufficient conditions for consensus of second-order multi-agent systems under directed topologies without global gain dependency.
\newblock {\em IEEE Transactions on Cybernetics}, 47(8): 2089-2098, 2017

\bibitem{QiqingITC}
Q. Qi and H. Zhang.
\newblock Output feedback control and stabilization for networked control systems with packet losses.
\newblock {\em IEEE Transactions on Cybernetics}, 47(8): 2223-2234, 2017.

\bibitem{NCaiJSSC}
N. Cai, M. He, Q. Wu, and M. J. Khan.
\newblock On almost controllability of dynamical complex networks with noises.
\newblock {\em Journal of Systems Science and Complexity}, 2017, DOI: 10.1007/s11424-017-6273-7.

\bibitem{YuanGTAC2017}
Y. Sun, Y. Tian, and X.J. Xie.
\newblock Stabilization of positive switched linear systems and its application in consensus of multi-agent systems.
\newblock {\em IEEE Trans Automat Contr}, 62: 6608-6613, 2017.

\bibitem{NiuLiuCyber2018}
X. L. Niu, Y. G. Liu, and F. Z. Li.
\newblock Consensus via time-varying feedback for uncertain stochastic nonlinear multi-agent systems.
\newblock{\em IEEE Transactions on Cybernetics}, 2018, DOI: 10.1109/TCYB.2018.2808336

\bibitem{JiYA}
Z. Ji, and H. Yu.
\newblock A new perspective to graphical characteri-zation of multi-agent controllability.
\newblock {\em IEEE Transactions on Cybernetics}, 47(6): 1471-1483, 2017.

\bibitem{LiuboIET2017}
B. Liu, H. Su, F. Jiang, Y. Gao, L. Liu and J. Qian.
\newblock Group controllability of continuous-time multi-agent systems.
\newblock {\em IET Control Theory \& Applications}, 2017, DOI: 10.1049/iet-cta.2017.0870

\bibitem{JLYTAC2015}
Z. Ji, H. Lin, and H. Yu.
\newblock Protocols design and uncontrollable topologies construction for multi-agent networks.
\newblock{\em IEEE Trans. Autom. Control}, 60(3): 781-786, 2015.

\bibitem{AgGhAuto2017}
C. O. Aguilar and B. Gharesifard.
\newblock  Almost equitable partitions and new necessary conditions for network controllability.
\newblock {\em	Automatica}, 80:25-31, 2017

\bibitem{GGCo}
G. Notarstefano, and G. Parlangeli.
\newblock Controllability and observability of grid graphs via reduction and symmetries.
\newblock {\em IEEE Trans Automat Contr}, 58: 1719-1731, 2013.

\bibitem{GGCoTAC2015}
C. O. Aguilar and B. Gharesifard.
\newblock Graph controllability classes for the Laplacian leader-follower dynamics.
\newblock {\em IEEE Trans. Autom. Contr}, 60(6): 1611?1623, 2015.

\bibitem{HsuLetter2016}
S.-P. Hsu.
\newblock Controllability of the multi-agent system modeled by the
threshold graph with one repeated degree.
\newblock{\em Syst. Control Lett.}, 97: 149-156, 2016.

\bibitem{LiuJCon}
X. Liu, and Z. Ji.
\newblock Controllability of multi-agent systems based on path and cycle graphs.
\newblock {\em International Journal of Robust and NonlinearControl}, 28(1): 296-309, 2017.

\bibitem{GuanWangletter2018}
Y. Guan and L. Wang.
\newblock Controllability of multi-agent systems with directed and weighted signed networks.
\newblock {\em Syst. Control Lett.}, 116: 47-55, 2018

\bibitem{JiWCon}
Z.~Ji, Z.~D. Wang, H.~Lin, and Z.~Wang.
\newblock Interconnection topologies for multi-agent coordination under
  leader-follower framework.
\newblock {\em Automatica}, 45(12):2857--2863, 2009.

\bibitem{TanOn}
H. G. Tanner.
\newblock On the controllability of nearest neighbor interconnections.
\newblock {\em Proc. 43rd IEEE Conf. Decis. Control}, Nassau, The Bahamas, Dec. 2004, pp. 24672472.

\bibitem{JiLL}
Z. Ji, H. Lin, and H. Yu.
\newblock Leaders in multi-agent controllability under consensus algorithm and tree topology.
\newblock {\em Systems \& Control Letters}, 61(9):918-925, July 2012.

\bibitem{YWETac2016}
A. Yazicio$\breve{\text{g}}$lu, W. Abbas, and M. Egerstedt.
\newblock Graph distances and controllability of networks.
\newblock {\em IEEE Trans. Automat. Contr.}, 61(12): 4125-4130, 2016

\bibitem{GuanRNC2017}
Y. Guan, Z. Ji, L. Zhang, and L. Wang.
\newblock Controllability of multi-agent systems under directed topology.
\newblock {\em International Journal of Robust and Nonlinear Control}. 27(18): 4333-4347, 2017

\bibitem{MHMTAC2018}
S. S. Mousavi, M. Haeri, and M. Mesbahi.
\newblock On the structural and strong structural controllability of undirected networks.
\newblock{\em IEEE Trans. Automat. Contr.}, 2018, DOI: 10.1109/TAC.2017.2762620

\bibitem{LiuCRNC2012}
B. Liu, T. Chu, L. Wang, Z. Zuo, G. Chen, H. Su.
\newblock Controllability of switching networks of multi-agent systems.
\newblock{\em International Journal of Robust and Nonlinear Control}, 22: 630-644, 2012

\bibitem{ChaoIMA}
Y. Chao and Z. Ji.
\newblock Necessary and sufficient conditions for multi-agent controllability of path and star topologies by exploring the information of second-order neighbors.
\newblock{\em IMA Journal of Mathematical Control and Information}, 2016, DOI:10.1093/imamci/dnw013

\bibitem{ParlanNTAC2012}
G. Parlangeli and G. Notarstefano.
\newblock On the reachability and observability of path and cycle graphs.
\newblock{\em IEEE Trans. Autom. Control}, 57(3): 743-748, 2012.

\bibitem{ZhaoGuan2017}
B. Zhao, Y. Guan, and L. Wang
\newblock Non-fragility of multi-agent controllability.
\newblock{\em Science China Information Sciences}, 61: 052202, 2018

\bibitem{HsuIJRNC2017}
S.-P. Hsu.
\newblock A necessary and sufficient condition for the controllability of single-leader multi-chain systems.
\newblock{\em Int. J. Robust Nonlinear Control}, 27(1): 156-168, 2017.

\bibitem{LuINJRNC2018}
Z. Lu, L. Zhang,  and Long Wang.
\newblock Controllability of discrete-time multi-agent systems with switching topology.
\newblock{\em Int. J. Robust Nonlinear Control}, 28(6), 2560-2573, 2018.

\end{thebibliography}

\end{document}